\documentclass[11pt]{article}

\usepackage{styling}

\usetikzlibrary{decorations.pathreplacing}

\AddCommenter{g}{yellow}
\AddCommenter{m}{cyan}
\AddCommenter{d}{white!50!green}

\newcommand{\Exp}{\ensuremath{\operatorname{\mathbb{E}}}}
\renewcommand{\Pr}{\ensuremath{\operatorname{\mathbb{P}}}}

\newcommand{\DD}{\mathcal{D}}

\renewcommand{\SS}{\mathcal{S}}
\DeclareMathOperator{\swap}{\mathtt{swap}}
\DeclareMathOperator{\adm}{\mathtt{adm}}
\DeclareMathOperator{\lopt}{\mathtt{lopt}}

\DeclareMathOperator{\id}{\mathtt{id}}
\DeclareMathOperator{\mdev}{mdev}
\DeclareMathOperator{\dev}{tdev}
\DeclareMathOperator{\mdsp}{mdsp}
\DeclareMathOperator{\dsp}{tdsp}

\newcommand{\FourSplitLine}[4]{
\draw[{Circle[scale=0.7]}-stealth,black,very thick] 
(#1 * \Side, 0.65 * \Side) -- 
(#1 * \Side, #2 * \Side) --
(#4 * \Side, #3 * \Side) --
(#4 * \Side, - 1.5 * \Side);}
\newcommand{\SplitLine}[3]{\FourSplitLine{#1}{#2}{#2}{#3}}

\hypersetup{
  pdftitle={Naively Sorting Evolving Data is Optimal and Robust},
  pdfauthor={George Giakkoupis, Marcos Kiwi, Dimitrios Los},
  breaklinks=true
  }

\begin{document}

\date{}
\title{Naively Sorting Evolving Data is Optimal and Robust}

\author{
George Giakkoupis\thanks{George Giakkoupis was supported in part by Agence Nationale de la Recherche (ANR) under project ByBloS (ANR-20-CE25-0002), and by Inria Action Exploratoire DisEvo.}\\
\textit{Inria} \\
Rennes, France \\[-0.07cm]  {\footnotesize \texttt{george.giakkoupis@inria.fr} }
\and
Marcos Kiwi\thanks{Marcos Kiwi was supported by FB210005, BASAL funds for Centers of Excellence from ANID-Chile, and GrHyDy ANR-20-CE40-0002.} \\
\textit{Universidad de Chile} \\
Santiago, Chile \\[-0.07cm]
{\footnotesize \texttt{mk@dim.uchile.cl} }
\and Dimitrios Los\thanks{Dimitrios Los was supported by the LMS Early Career Fellowship (ref.~2024-19), and Inria Action Exploratoire DisEvo.} \\
\textit{Inria} \\
Rennes, France \\[-0.07cm]
{\footnotesize \texttt{mail@dimitrioslos.com}}}

\maketitle

\begin{abstract}
We study comparison sorting in the \emph{evolving data model}, 
introduced by Anagnostopoulos, Kumar, Mahdian and Upfal (2011), 
where the true total order changes \emph{while} the sorting algorithm is processing the input.
More precisely, each comparison operation of the algorithm is followed by a sequence of evolution steps, where an evolution step perturbs the rank of a random item by a ``small'' random value. The goal is to maintain an ordering that remains close to the true order over time.  Previous works have analyzed adaptations of classic sorting algorithms, assuming that an evolution step changes the rank of an item by just one, and that a fixed constant number~$b$ of evolution steps take place between two comparisons.
In fact, the only previous result achieving optimal linear total deviation, by Besa Vial, Devanny, Eppstein, Goodrich and Johnson (2018a), applies just for $b=1$.

We analyze a very simple sorting algorithm suggested by Mahdian (2014), which samples a random pair of adjacent items in each step and swaps them if they are out of order.
We show that the algorithm  achieves and maintains, with high probability, optimal total deviation, $O(n)$,  and optimal maximum deviation, $O(\log n)$, 
under very general model settings. 
Namely, the perturbation introduced by each evolution step is sampled from a general distribution of bounded moment generating function, and we just require that the \emph{average} number of evolution steps between two sorting steps be bounded by an (arbitrary) constant, where the average is over a linear number of steps.

The key ingredients of our proof are a novel potential function argument that inserts ``gaps'' in the list of items, and a general analysis framework which separates the analysis of sorting from that of the evolution steps, and is applicable to a variety of settings for which previous approaches do not apply. 
Our results settle conjectures and open problems in the three aforementioned works, and  provide theoretical support that simple quadratic algorithms are optimal and robust for sorting evolving data, as empirically observed by Besa Vial, Devanny, Eppstein, Goodrich  and Johnson (2018b).
\end{abstract}

\section{Introduction}
\label{sec:intro}

We consider the problem of sorting in the \emph{evolving data model}, introduced by Anagnostopoulos, Kumar, Mahdian and Upfal~\cite{AnagnostopoulosKMU09,AnagnostopoulosKMU11}.
In this model, the input data are changing slowly while the algorithm is processing them.
Changes to the data typically follow some
simple stochastic process, and they are not directly communicated to the algorithm when they occur; instead the algorithm has query access to the data.
The challenge for the algorithm is then to maintain over time 
an output close to the correct output for the current state of the data.
This model is different from the standard dynamic (graph) algorithm model~\cite{EppsteinGI10}, which assumes that changes are adversarial but are reported to the algorithm.
Also, the focus of the evolving data model is on query complexity rather than computational complexity, in particular, on the rate of queries of the algorithm relative to the rate at which the data are changing. 

Several classic problems have been studied in the evolving data model, from  stable matching with evolving preferences~\cite{KanadeLM16}, to 
community detection~\cite{AnagnostopoulosALLLM16}
and PageRank computation in evolving graphs~\cite{BahmaniKMU12} (see section Other Related Work).
We focus on the problem of sorting a set $S$ of~$n$ items, when the true underlying total order $\rho$ of $S$ changes gradually over time.
The sorting algorithm, which may only perform pairwise comparisons, 
is tasked with maintaining an ordering $\nu$ on $S$ that is a close approximation to the true order $\rho$,
where $\rho$ and $\nu$ are permutations on $S$. %

The above sorting problem arises naturally in various real-life settings where ranks are evolving over time, and comparisons tend to be slow or expensive.
For example, ranking tennis or chess players whose abilities change over time, based on head-to-head matches \cite{Glickman02}; maintaining a ranking of a set of sports teams, whose competitiveness varies during a season;
ranking political candidates, where comparisons involve running a debate or doing a poll; 
ranking movies, songs, or websites, where comparisons involve an online survey or A/B testing.

Sorting was the first problem
to be studied in the evolving data model~\cite{AnagnostopoulosKMU11}.
The precise model considered was that each step of the algorithm is followed by~$b$ evolution steps, for some constant $b \geq 1$. 
A step of the algorithm involves a pairwise comparison between the ranks in $\rho$ for two items in $S$ (specified by the algorithm),
followed by arbitrary updates to the maintained order $\nu$. 
An evolution step consists of a \emph{random adjacent rank swap}, i.e., swapping $\rho(i)$ and $\rho(i+1)$ for an index $i\in [n-1]$  sampled uniformly at random. 
No memory or computational constraints are imposed on the sorting algorithm, which is executed by a central agent.

It was shown 
in \cite{AnagnostopoulosKMU11} that repeated application of classic randomized Quicksort achieves, after the first $O(n\log n)$ steps, maximum deviation 
$\max_{s\in S}|\nu^{-1}(s) - \rho^{-1}(s)|= O(\log n)$ 
between the maintained order $\nu^{-1}(s)$ and the true order $\rho^{-1}(s)$ of any item $s\in S$ with high probability (w.h.p.),\footnote{With high probability refers to probability of at least $1 - n^{-c}$ for some constant $c > 0$.} and thus  total deviation $\sum_{s\in S}|\nu^{-1}(s) - \rho^{-1}(s)| = O(n\log n)$.\footnote{In fact, the actual bound was expressed in terms of the Kendall-tau distance which is equal to the total deviation within a factor of 2 (see \cref{def:permutation-distances} and Eq.\cref{eq:Kvsdis}.)}
Moreover, a refined algorithm, which runs Quicksort interleaved 
with copies of Quicksort on smaller ranges of the maintained list, achieves total deviation $O(n\log\log n)$.
On the lower bound side, it was  shown that for any algorithm the total deviation is $\Omega(n)$ w.h.p., and this was conjectured to be tight for any $b \geq 1$. %

Besa Vial et al.~\cite{VialDEGJ18} confirmed the conjecture for the special case of $b=1$, for repeated application of Insertion Sort. More specifically, they proved that Insertion Sort achieves $O(n)$ total deviation w.h.p., after the first $O(n^2)$ steps, and when combined with Quicksort reduces the steps required to $O(n\log n)$. Their analysis relied heavily on the assumption that $b = 1$, and so they left open the problem of achieving linear total deviation in the general case $b \geq 2$.

In parallel, experimental work by Besa Vial et al.~\cite{VialDEGJ18a} suggested that several quadratic algorithms, namely, Bubble Sort, Cocktail Sort, and Insertion Sort, all achieve linear total deviation for any constant rate $b$ of adjacent rank swaps,
and even when random adjacent swaps are replaced by more general, local rank perturbations (as we discuss below).
They posed as an interesting problem to provide theoretical analyses supporting these empirical results.

\paragraph{Our Contributions.}
Since the problem of evolving sorting was first introduced~\cite{AnagnostopoulosKMU09}, it has been suspected that the following basic randomized algorithm is asymptotically optimal, achieving linear total deviation (see, for example, the talk by Mahdian~\cite{talk2014}).
\begin{quote}
    \emph{Na\"ive Sort}: In each step, chose a uniformly random pair of adjacent items and swap them if they are out of order.\footnote{We coined the name `Na\"ive Sort' after failing to find an established name for this algorithm in the literature.} 
\end{quote}
Na\"ive Sort has several attractive properties: 
It is extremely simple; 
it has minimal memory requirements, in fact, there is no need to maintain any state from one step to the next;
and it is inherently parallelizable.
Also, it has quadratic comparison complexity (w.h.p.)~\cite{GasieniecSS23}, similarly to Bubble Sort, Cocktail Sort, and Insertion Sort.

We show that Na\"ive Sort achieves indeed optimal asymptotic deviation, in a robust manner.
We analyze Na\"ive Sort under a general model, which extends the original model in two ways.
First, instead of random adjacent rank swaps, we assume more general \emph{local rank perturbations}:
\begin{quote}
    \emph{Local rank perturbation step}: 
    We sample an integer $s$ from a distribution of zero mean and bounded moment generating function; %
    then we successively swap the rank of a uniformly random item with the ranks of its~$s$ succeeding items in $\rho$ if $s>0$, or with the ranks of its $-s$ preceding items in $\rho$ if $s<0$
    (see \cref{def:local-rank-perturbation}).

\end{quote}
Second, rather than assuming that the number of evolution steps between two consecutive sorting steps is a fixed number $b$, we just assume that it is bounded \emph{on average}, where the averaging is over a linear number of steps.
\begin{quote}
    \emph{Bounded average rate of  evolution steps}:
    Let $b_i$, for $i\geq1$, be the number of evolution steps between the $i$-th and $(i+1)$-th sorting steps.
    Then,
    \begin{equation}
        \label{eq:bounded-rate}
        \sum_{i=t+1}^{t+n} b_i \leq b\cdot n,\
        \text{ for all }
        t\geq0,
    \end{equation}
    where $b$ can be an arbitrary constant
    (see \cref{def:bounded average rate if mixing steps}).
\end{quote}
Our main result can then be stated as follows.

\begin{theorem}
    \label{thm:intro-first}
    Under evolution steps that are local rank perturbations and occur at bounded average rate, it holds for any $t=\Omega(n^2)$ large enough, %
    that after $t$ steps of Na\"ive Sort, the maximum deviation between the maintained order %
    and the true order %
    is $O(\log n)$  and the total deviation is $O(n)$ w.h.p.
\end{theorem}

Our linear bound on the total deviation is optimal due to the lower bound of \cite{AnagnostopoulosKMU11}.
The logarithmic bound on the maximum deviation is also optimal.
In fact, we prove a more general bound of $O(b\log n)$ on the maximum deviation, which applies even when $b$ in Eq.~\cref{eq:bounded-rate} is super-constant.
It is not difficult to prove that this is optimal (see~\cref{lem:max_and_total_deviation_lb}).

The requirement in \cref{thm:intro-first} that $t = \Omega(n^2)$, for the number of sorting steps $t$ before the claimed deviation bounds are achieved, is in place to cover the worst-case initialization.
We show that a more refined requirement suffices, namely, that  $t = \Omega(n\cdot (\Delta+\log n))$, where $\Delta$ is the maximum deviation 
initially
(and moreover this is tight; see \cref{lem:convergence_time_lb}).
Also, similarly to~\cite{VialDEGJ18}, 
it is possible to combine Na\"ive Sort with Quicksort to ensure it suffices that $t = \Omega(n \cdot \log n)$, for all initializations.

Prior to our work, no algorithm was known to achieve linear total deviation, even in the original model (with random adjacent rank swaps, and constant number $b$ of evolution steps between sorting steps), when $b\geq 2$.
In particular, the algorithms of \cite{AnagnostopoulosKMU11} achieve super-linear total deviation, and \cite{VialDEGJ18} prove linear total deviation only for $b=1$ (and have no  upper bounds on the maximum deviation).

Thus, \cref{thm:intro-first} settles the conjecture of \cite{AnagnostopoulosKMU11} (and the open problem in~\cite{VialDEGJ18}), that there is an algorithm achieving linear total deviation for their original model in the general case of  $b \geq 2$, and also the more specific conjecture of \cite{talk2014} that this is achievable by Na\"ive Sort.

More general models have been considered before only experimentally, in \cite{VialDEGJ18a}.
In fact, the \emph{hot spot adversary} proposed there is a special case of our model of local rank perturbations, where  $s$ is sampled from a geometric distribution with constant success probability and then its sign is flipped with probability~$1/2$.\footnote{Actually, the hot spot adversary was an inspiration for our local rank perturbation model.} 
Our results thus provide theoretical support for the empirical evidence in~\cite{VialDEGJ18a} that simple quadratic algorithms are optimal and robust for sorting evolving data, even under more general assumptions than considered empirically before.

\paragraph{Technical Overview.}
We give now an overview of the key ideas used in the proof of \cref{thm:intro-first}.
Instead of working with $\nu$ and $\rho$, it suffices to work with a single permutation $\pi$ on $[n]$ such that $\pi(i) = j$ if $\nu(i) = \rho(j)$.
The maximum and total deviation of $\pi$ (from the identity permutation) are then the same as respectively the maximum and total deviation between
$\nu^{-1}$ and $\rho^{-1}$. 
Na\"ive Sort under local rank perturbations can be described in terms of $\pi$ as follows.
Let $\pi_t$ denote the permutation after $t$ steps.
\begin{itemize}
 \item If step $t+1$ is a sorting step, then a random pair $\pi_{t}(i),\pi_t(i+1)$ is chosen and is swapped if it is out of order, i.e., if $\pi_{t}(i) > \pi_t(i+1)$
 \item If step $t+1$ is an evolution step, which we call a \emph{mixing step} from now on, then a random item $i$ and a random perturbation $s$ are picked, and $\pi^{-1}_{t}(i)$ is swapped with its $s$ succeeding or its $-s$ preceding items (if $s<0$) in the inverse permutation $\pi_t^{-1}$.
\end{itemize}

The first key idea we employ is a novel potential function argument, which relies on inserting ``gaps'' between the elements of permutation $\pi_t$.
For Na\"ive Sort in the non-evolving setting (i.e., in the absence of mixing steps), there is already an elegant potential function argument that shows $\Theta(n^2)$ running time w.h.p., by Gasieniec, Spirakis and Stachowiak \cite{GasieniecSS23}.
Their proof first applies the standard 0/1 principle, which reduces the problem to sorting a list of 0's and 1's \cite[Section 5.3.4]{Knuth98a}, and then analyzes a potential function where each 1 contributes an exponential term, whose exponent depends on the number of succeeding~0's and the number of preceding~1's.  

Unluckily, the 0/1 principle no longer applies under mixing steps. 
Yet, it is still instructive to attempt to adapt their potential function for a 0/1 list, to our permutation setting.
A natural choice is potential function  $\sum_i\left(2^{|i-\pi_t(i)|} - 1\right)\cdot 2^{\pi_t(i)}$, where $|i-\pi_t(i)|$ is the deviation of item at position $i$ from its correct position $\pi_t(i)$.
We first observe that the simpler potential $\sum_i\left(2^{|i-\pi_t(i)|} - 1\right)$ does not work, even without mixing steps. The problem is due to sorted blocks (i.e., maximal contiguous sequences of sorted items) whose contribution to the potential can drop only if the first or last item in the block is selected. 
E.g., in permutation $\pi_t=(\tfrac{n}{2} + 1, \ldots ,n,1,\ldots ,\tfrac{n}{2})$ the potential decreases only if $n$ and $1$ are swapped, and thus the simpler potential drops in expectation just by a factor of $1 - \Theta(n^{-2})$ %
(whereas we would like to drop by a factor of $1 - \Omega(n^{-1})$).
This issue is handled by the additional factor of $2^{\pi_t(i)}$, which gives more weight to the contribution of larger items.
Even though this works fine if there are no mixing steps, it no longer works with mixing steps.
Such a step may drastically increase the potential, e.g., by swapping items $i$ and $i+1$, for a large $i$, when the two items are already at the right position, thus increasing their contribution to the potential from 0 to  $2^i+2^{i+1}$.

Instead of analyzing directly the evolution of permutation $\pi_t$, in our proof we consider a larger list $l_t$ of size roughly~$nd$, for some constant integer $d > 1$. 
The list is obtained by inserting ``gaps'', denoted by~$\bot$, to the permutation list $(\pi_t(1), \pi_t(2), \ldots,\pi_t(n))$.
The sorted state of $l_t$ is the one where $l_t[d\cdot i] = i$ for each $i\in[n]$, and $l_t[j]=\bot$ for the remaining indices $j$.
In general, the placement of gaps can be arbitrary, subject to satisfying a simple \emph{local invariant}, namely, that there is no gap to the left (right) of item $i$ if its target position~$d\cdot i$ is to its left (resp.\ right).
Each sorting and mixing step on $\pi_t$ translates then directly into a similar step in $l_t$, involving the same items.

The above elementary transformation of the problem, from sorting a permutation $\pi_t$ to sorting a list $l_t$ with gaps, serves the following purpose.
The natural potential function $\sum_i\left(2^{|i-\pi_t(i)|} - 1\right)$, which, as we saw earlier, does not behave as we wanted, works extremely well when applied to $l_t$; the precise potential function in now $\sum_{l_t[j]\neq\bot} \left(e^{\alpha\cdot |j-d\cdot l_t[j]|} - 1\right)$, where $\alpha>0$ is a constant.
In particular, the problem we had before with sorted blocks is no longer an issue, because now the contribution to the potential function of a sorted block (without gaps) is dominated by the contribution of the first (or last) item of the block.
It is then not difficult to show that, in the absence of mixing steps, the exponential function drops by the desirable rate of $1-\Omega(n^{-1})$.
Hence, there is no need for the exponential weight factor used in the approach based on~\cite{GasieniecSS23}, which was a main obstacle to our handling of mixing steps.
Further, the simpler potential function facilitates the technique we describe next, for analyzing independently sorting and mixing steps.

The second key idea of the proof is a way to separate the analysis of the sorting steps from that of the mixing steps.
The idea is to move
the target position of each item~$i$ in the list $l_t$ during the process. 
For that, we rely on an additional permutation $\tau_t$ on $[n]$, and let the target position of item $i$ in list $l_t$ be $d\cdot \tau_t(i)$ ($\tau_0$ is the identity permutation, thus the initial target position of $i$ is $d\cdot i$).  
Each time a mixing step occurs, 
we modify also~$\tau_t$ by swapping $\tau_t(i)$ with its $-s$ preceding or $s$ succeeding items in $\tau_t$, where~$i$ and $s$ are the item and perturbation selected at the mixing step;
when a sorting step occurs, $\tau_t$ does not change.
In addition, in each step we must maintain a basic invariant, called \emph{admissibility condition}, 
which may lead to swapping some pairs in $\tau_t$ at the end of the step (see \cref{def:admissible}); but such swaps are favorable for our analysis.

Using $d\cdot \tau_t(l_t[j])$ as the target position for each element $l_t[j]$, instead of $d\cdot l_t[j]$, we can apply the same potential function analysis as in the case without mixing steps; precisely the potential function is
$\Phi_t = \sum_{l_t[j]\neq\bot} \left(e^{\alpha\cdot |j-d\cdot \tau_t(l_t[j])|} - 1\right)$.
In the analysis of $\Phi_t$ we can ignore mixing steps because our construction ensures that a mixing step cannot increase $\Phi_t$.
And as before we have that $\Phi_t$ decreases in expectation by a factor of $1-\Omega(n^{-1})$ with each sorting step.

We analyze separately the deviation of $\tau_t$ using another potential function, %
$\Psi_t = \sum_i e^{\alpha'\cdot|i - \tau_t(i)|} = \sum_i e^{\alpha'\cdot|\tau_t^{-1}(i)-i|}$.
In the absence of sorting steps, and if it were $s=1$ for all mixing steps, then the differences $\tau_t^{-1}(i)-i$ would evolve similarly to simple (unbiased) random walks.
Our construction ensures that sorting steps cannot increase potential $\Psi_t$, and the assumption that $s$ has bounded MGF means that  $\tau_t^{-1}(i)-i$ is still close to a simple random walk.
We thus use $\Psi_t$ to show that the maximum deviation of $\tau_t$ after $t = k \cdot n$ steps is at most $O(\sqrt{k\log n})$ w.h.p.

From the two results above, for the sorting and mixing steps, we have that starting from a maximum deviation of $k$ for $\pi_0$, and thus $\Phi_0 \leq n\cdot e^{O(k)}$, it takes $t = O(k\cdot n)$ steps before $\Phi_t = O(n)$ (as the potential drops in expectation by a factor of $1-\Omega(n^{-1})$ with each sorting step), 
and at that point the maximum deviation of $\tau_t$ is $O(\sqrt{k\log n})$.  
It follows that the maximum deviation of $\pi_t$ at that moment is $O(\log n + \sqrt{k\log n})$.
Repeating the argument, we show that for any initialization, after 
$t = \Theta(n^2)$ steps, the maximum deviation of $\pi_t$ is $O(\log n)$ w.h.p.

The last key component of the analysis is a method to refine the above argument that bounds the maximum deviation, to establish a linear bound on the total deviation for~$\pi_t$. 
As is, the above argument gives just $O(n\log n)$ total deviation.
The problem is that potential $\Psi_t$ can only be used to bound the maximum deviation of $\tau_t$, and provides little information about the total deviation of~$\tau_t$.

We show that unlike the maximum deviation of $\tau_t$ which grows to $O(\sqrt{k\log n})$ after $t = kn$ steps, most of the deviation terms $|i-\tau_t(i)|$ are concentrated around~$\sqrt{k}$.
In particular, the sum of all deviation terms larger than $k^{2/3}$ is $o(n)$ w.h.p.
This facilitates the following approach.
Starting from maximum deviation $k=\Theta(\log n)$ for~$\pi_t$, and thus $\Phi_t \leq n\cdot e^{O(k)}$, we consider a phase of $O(kn)$ steps until $\Phi_t = O(n)$. At the end of the phase we perform a ``partial target reset,'' which roughly amounts to setting  $\tau_t(i) = i$ for all  $i$ such that $|\tau_t(i) - i| \leq k^{2/3}$ and leaving unchanged the remaining entries.
The unchanged entries contribute just $o(n)$ to the total deviation of $\tau_t$, while 
the partial reset ensures that potential $\Phi_t$ is %
$n\cdot e^{O(k^{2/3})}$.
We then repeat the above, with the length of the phase reduced from $O(nk)$ to $O(nk^{2/3})$. %
After a few such phases 
we obtain that $\Phi_t = O(n)$ and also the total deviation of $\tau_t$ is $O(n)$, which together imply the desired linear total deviation of~$\pi_t$.

We formalize this argument by introducing another permutation $\sigma_t$ on $[n]$,
which is reset to the identity permutation at the beginning of each of the phases above, and is updated according to the same rules as $\tau_t$ (which is only partially reset).
Potential $\Psi_t$ is then defined in terms of $\sigma_t$ rather than $\tau_t$.
And the proof proceeds by showing that in each phase, $\tau_t$ and~$\sigma_t$ remain close to each other, while potential $\Phi_t$ drops to $O(n)$, and potential $\Psi_t$ does not increase much.

We note that previous techniques were sensitive to the specific settings analyzed. Here, we demonstrate that the use of potential function arguments allows us to overcome several obstacles that hindered those techniques, and analyze an extremely general setting. Even though we focus on the fundamental problem of sorting, we expect that techniques therein extend also to other problems with evolving data. %

\paragraph{Other Related Work.}

A problem related to evolving sorting is the well-studied problem of noisy sorting~\cite{FeigeRPU94,BravermanM08,GuX23}, where there is a fixed underlying total order, but the outcome of pairwise comparisons may be incorrect.
Depending on the precise setting, the goal  may be to find the most likely total order, e.g., when comparisons are permanently incorrect, or to find the correct order with a small number of comparisons, e.g., when failure probabilities of comparisons are independent.
Unlike noisy sorting, in evolving sorting comparisons are accurate but the underlying total order is changing.

The biased shuffling algorithm studied by Benjamini et al. \cite{Benjamini2005} (see also \cite{DiaconisR00}), can be viewed as a noisy version of na\"ive sorting, where each comparison is faulty independently with probability $f < 1/2$.
The mixing time of this noisy na\"ive sorting was shown to be $O(n^2)$. 
It would be interesting to study na\"ive sorting in a setting combining both evolving ranks and noisy comparisons.
Our current analysis does not work if we allow for noisy comparisons.

Besides sorting, several problems have been studied in the evolving data model: 
selecting the $k$-th %
element or the top-$k$ elements under evolving rankings~\cite{AnagnostopoulosKMU11,HuangLSZ17};
stable matching with evolving preferences~\cite{KanadeLM16};
label tracking on trees with evolving label-to-vertex mappings~\cite{AcharyaM22};
and various problems on evolving graphs, 
including $(s,t)$-connectivity and minimum spanning tree~\cite{AnagnostopoulosKMUV12},
densest subgraph computation~\cite{EpastoLS15},
community detection on the stochastic block model~\cite{AnagnostopoulosALLLM16},
and PageRank computations~\cite{BahmaniKMU12,OhsakaMK15,MoL21}.
In these works, the evolution that the graph undergoes typically involves changing the list of edges by adding a new edge or deleting an existing edge at each step, or changing the ranking of the edge weights.
The %
minimum spanning tree algorithm for evolving edge weights proposed in~\cite{AnagnostopoulosKMUV12}, and a generalization of it that finds a basis of minimum weight in an evolving matroid model, both use an evolving sorting algorithm as a 
component to sort the weights.

\paragraph{Road Map.}

In \cref{sec:prelims} we introduce some terminology and notation, define the processes and settings we analyze, and provide some auxiliary lemmas. In \cref{sec:max_deviation,sec:total_deviation} we establish \cref{thm:intro-first}. More specifically, in \cref{sec:max_deviation} we prove the $O(b \cdot \log n)$ bound on the maximum deviation (\cref{thm:mdev_whp}) and the bound on the convergence time (\cref{cor:max_deviation_convergence_time}). In \cref{sec:total_deviation}, we prove the $O(n)$ bound on the total deviation for any constant $b \geq 1$. In \cref{sec:lower_bounds}, we complement the upper bounds with some asymptotically tight lower bounds.

\section{Preliminaries}
\label{sec:prelims}

\subsection{Some Standard Definitions} %

Let $[n] = \{1,2,\ldots,n\}$. 
By $\SS_n$ we denote the set of permutations of $[n]$, and by $\id_n$ the identity permutation. 
For any $\pi\in \SS_n$, and $i_1,i_2\in [n]$, let $\swap(\pi, i_1, i_2)$ be the permutation $\pi'$ defined as $\pi$ with~$\pi(i_1)$ and $\pi(i_2)$ swapped, i.e.,
\[
  \pi'(i) = \begin{cases}
     \pi(i_2), & \text{if }i = i_1, \\ 
     \pi(i_1), & \text{if }i = i_2, \\
     \pi(i), & \text{otherwise}.
 \end{cases}
\]
We will use the following distance measures between two permutations.
\begin{definition} [Permutation Distances]
    \label{def:permutation-distances}
    For any permutations $\pi,\pi'\in \SS_n$,
    their
    \emph{maximum deviation}
    is
    \[
        \mdev(\pi,\pi') 
        =
        \max_{i\in [n]} |\pi(i)-\pi'(i)|
        ,
    \]
    their
    \emph{total deviation} (also known as 
    \emph{Spearman’s footrule}~\cite{DiakonisG77}) is
    \[
        \dev(\pi,\pi') = \sum_{i\in [n]} |\pi(i)-\pi'(i)|,
    \]
    and their \emph{Kendall-tau distance}~\cite{K38} (also known as \emph{bubble-sort distance}) is
    \[
      K(\pi,\pi')=|\{(i,j)\in [n]^2 : \pi(i)<\pi(j), \pi'(i)>\pi'(j)\}|.
    \]
    Also $\mdev(\pi) = \mdev(\pi,\id_n)$ is just called  the maximum deviation of $\pi$;
    and similarly for the other two distances.    
\end{definition}
The following close relation between total deviation and Kendall-tau distance was shown in
\cite{DiakonisG77},
\begin{equation}
    \label{eq:Kvsdis}
    K(\pi,\pi')\leq \dev(\pi,\pi') \leq 2\cdot K(\pi,\pi').
\end{equation}

\subsection{Evolving Sorting Model and Na\"ive Sort}
\label{sec:basic-model}

The evolving sorting model we consider in this paper is a generalization of the original model introduced in~\cite{AnagnostopoulosKMU11}.
First, we describe the original model before we present its generalization.
Then, we give a more concise representation, and apply the model to Na\"ive Sort.

\subsubsection{Original Model}
Let
$S$ be a set on $n$ items, and
without loss of generality suppose that $S=[n]$. We have an infinite sequence $(\nu_t,\rho_t)_{t\geq 0}$ on pairs of permutations from $\SS_n^2$, where $\rho_t$ denotes the underlying total order of set $S$ after~$t$ steps, and $\nu_t$ denotes the maintained approximation to $\rho_t$.
Each step is either a sorting step or a mixing step.

A \emph{sorting step} $t$ consists of a single comparison between a pair of items in $S$, which returns their relative order in $\rho_{t-1}$, followed by arbitrary modifications to~$\nu_{t-1}$, which yield $\nu_{t}$.
The pair to be compared and the subsequent modifications to $\nu_{t-1}$ are decided by the sorting algorithm.
Also, $\rho_t = \rho_{t-1}$.

A \emph{mixing step} $t$ consists of a single \emph{random adjacent rank swap}, which samples a uniformly random index $i\in [n-1]$, and swaps $\rho_{t-1}(i)$ and $\rho_{t-1}(i+1)$, i.e., $\rho_t = \swap(\rho_{t-1},i,i+1)$.
Also, $\nu_t = \nu_{t-1}$.

The order of mixing and sorting steps is such that each sorting step is followed by exactly $b\geq 1$ mixing steps, where $b$ is an integer parameter of the model that is assumed to be a constant (independent of $n$).

\subsubsection{Generalized Model}

We generalize two properties of the original model above.
First, we let a mixing step consist of a more general operation, called \emph{local rank perturbation}.

\begin{definition}[Local Rank Perturbation]
    \label{def:local-rank-perturbation}

    Let $\DD$ 
    be a distribution on the integers that has zero mean and 
    bounded moment generation function (MGF), i.e., if $D$ 
    is a random variable with distribution $\DD$, then $\Ex{D}=0$ and for some constants $\lambda, c>0$, $$\Ex{e^{\lambda\cdot |D|}} \leq c.$$
    We assume that $\DD$ is given as a parameter of the model, and is called \emph{perturbation distribution}.
    If step $t$ is a mixing step, 
    then a uniformly random integer $j$ is sampled from $[n]$, a random integer $s$ is sampled from distribution~$\DD$,
     and
    $\rho_{t}$ 
    is obtained by successively swapping $\rho_{t-1}(j)$ with $\rho_{t-1}(j+1), \rho_{t-1}(j+2),\ldots,\rho_{t-1}(\min\{j+s,n\})$ if $s>0$, or by successively swapping $\rho_{t-1}(j)$ with $\rho_{t-1}(j-1), \rho_{t-1}(j-2),\ldots,\rho_{t-1}(\max\{j+s,1\})$ if $s<0$.
    
\end{definition}

For $\DD$ being the uniform distribution over $\{-1,1\}$, we obtain a random adjacent rank swap.
While if $s$ is sampled from a geometric distribution with constant success probability and then its sign is flipped with probability $1/2$, we obtain the \emph{hot spot adversary} proposed and studied empirically in \cite{VialDEGJ18a}. 

The second generalization is that we allow the number of mixing steps between two sorting steps to vary, as long as the \emph{average} is bounded by a constant.
The formal description of this requirement is as follows. 

\begin{definition}[Bounded Average Rate of Mixing Steps]
    \label{def:bounded average rate if mixing steps}
    Let $b_i$, for $i\geq1$, be the number of mixing steps between the $i$-th and $(i+1)$-th sorting steps.
    Then,
    \begin{equation}
        \label{eq:bounded-rate-mix}
        \sum_{i=t+1}^{t+n} b_i \leq b\cdot n,\
        \text{ for all }
        t\geq0,
    \end{equation}
    where $b$ is a constant parameter of the model.
    We also assume that the sequence of $b_i$ is chosen independently of the randomness used in the sorting and evolution steps.\footnote{For the purposes of the analysis, we can assume it is fixed in advance, before the process starts.}    
\end{definition}

By letting $b_i = b$ for all $i$, we obtain the original model of \cite{AnagnostopoulosKMU11}. 
Our model also encompasses a wide class of settings, including those with $b_i$'s chosen independently from any reasonable distribution of bounded mean. %
Note that since the number of steps in the evolving sorting model is unbounded, it was necessary to restrict the range of the sum in \cref{eq:bounded-rate-mix} to involve a finite number of sorting steps. 
Our decision that the number of sorting steps be precisely $n$ was arbitrary, but we need that it is at most $O(n)$ in some parts of the analysis.

\subsubsection{Evolving Sorting with Na\"ive Sort}

For each $t\geq0$, the pair $\nu_t, \rho_t$ of permutations  gives rise to a permutation~$\pi_t\in \SS_n$, such that 
$\pi_t(i) = j$ if $\nu_t(i) = \rho_t(j)$.\footnote{Or, more concisely, $\pi_t=\rho_t^{-1}\circ \nu_t$.}
Note that the sorted case of $\nu_t = \rho_t$ corresponds to $\pi_t = \id_n$. 
Since 
$\mdev(\nu_t^{-1},\rho_t^{-1}) = \mdev(\pi_t)$ and $\dev(\nu_t^{-1},\rho_t^{-1}) = \dev(\pi_t)$, it suffices for our analysis to focus on the sequence of $\pi_t$ instead. 

For convenience, below we give the definition of sorting and mixing steps in terms of $\pi_t$, when Na\"ive Sort is used.
\begin{itemize}
    \item \emph{Sorting step}: 
    sample a uniformly random $i\in [n-1]$; if $\pi_{t-1}(i) > \pi_{t-1}(i+1)$ 
    then 
    $\pi_t = \swap(\pi_{t-1}, i,i+1)$,
    otherwise $\pi_t = \pi_{t-1}$.
    
    \item \emph{Mixing step}: 
    sample a uniformly random $i\in[n]$, and a random $s$ from distribution $\DD$; 
    if $s >0$, 
    then $\pi^{-1}_t$ is obtained by successively swapping $\pi^{-1}_{t-1}(i)$ with $\pi^{-1}_{t-1}(i+1),\pi^{-1}_{t-1}(i+2),\ldots,\pi^{-1}_{t-1}(\min\{i+s,n\})$;
    and if $s < 0$, $\pi^{-1}_t$ is obtained by successively swapping $\pi^{-1}_{t-1}(i)$ with $\pi^{-1}_{t-1}(i-1),\pi^{-1}_{t-1}(i-2),\ldots,\pi^{-1}_{t-1}(\max\{i+s,1\})$.\footnote{Equivalently, if $s >0$,
    $\pi_t$ is obtained by successively swapping item value $i$ in $\pi_{t-1}$ with item values $i+1,i+2,\ldots,\min\{i+s,n\}$;
    and if $s < 0$,  $\pi_t$ is obtained by successively swapping item value $i$ in $\pi_{t-1}$ with values $i-1,i-2,\ldots,\max\{i+s,1\}$.}    
\end{itemize}

By $(\mathcal{F}_t)_{t\geq 0}$ we denote the natural filtration of the underlying process, where $\mathcal{F}_t$ includes the random choices made in each step $t'\in \{1,\ldots, t\}$.

\subsection{Auxiliary Lists and Permutations}
\label{sec:auxiliary-process}

Next we describe some auxiliary components used in the proof. 

\subsubsection{Lists with Gaps and Target Permutations}

To facilitate the analysis of  $\pi_t$, we consider  a triple $l_t,\tau_t,\sigma_t$, where $l_t$ is a list and $\tau_t,\sigma_t$ are permutations from $\SS_n$.

List $l_t$ has size roughly~$nd$, for some constant integer $d > 1$, and it is obtained by inserting $\bot$ elements, called \emph{gaps}, to the list $(\pi_t(1), \pi_t(2), \ldots)$.
Each non-gap element $l_t[j]=i$ has a \emph{target position} indicated by permutation $\tau_t$, which is position $d\cdot \tau_t(i)$. 
For example, if $\tau_t = \id_n$ is the identity permutation, then the target position of each $l_t[j]\neq\bot$ is position $d\cdot l_t[j]$.
As we will see, sorting steps will tend to bring non-gap elements closer to their target position, and $l_t$ will be considered sorted if every element $l_t[j]\neq\bot$ is precisely at its target position, i.e., $d\cdot \tau_t(l_t[j]) = j$.
For example, if $\tau_t = \id_n$ then $l_t$ is sorted if $l_t[d\cdot i] = i$ for all $i\in[n]$ (and $l_t[j]=\bot$ for the remaining indices $j$).

In general, the placement of gaps in $l_t$ can be arbitrary, subject to a simple \emph{local optimality} requirement, which states that there should be no gap to the right (left) of an element if its target position is to its right (resp.\ left).
In other words, we cannot move any single element in $l_t$ closer to its target position without reordering non-gap elements.

To formally describe $l_t$, we will need the following definitions of a \emph{$d$-padding} list, and operation $\lopt$ that makes a list locally optimal.

\begin{definition}[$d$-Padding \& Local Optimimality]\label{def:padding}
    Let $\pi,\tau\in \SS_n$ be  permutations and $d> 1$ an integer. 
    A \emph{$d$-padding} of $\pi$ is a list $l$ of length $N = (n+1)\cdot d-1$ such that $l$ contains list $(\pi(1),\ldots,\pi(n))$ as a (not necessarily contiguous) subsequence, and the remaining $N - n$
    elements are equal to $\bot$.
    We will just say that $l$ is a $d$-padding, to denote that it is a $d$-padding of some permutation from $\SS_n$.
    
    We say that $l$ is \emph{locally optimal}
    w.r.t.\ permutation $\tau$, 
    if for all $j\in [N]$ with $l[j]\neq \bot$,
    \begin{align}
        & \big(
        j < d\cdot \tau(l[j])
        \implies
        l[j+1] \neq \bot
        \big) \notag \\
        & \qquad \ \land\ 
        \big(
        j > d\cdot \tau(l[j])
         \implies
        l[j-1] \neq \bot
        \big). \label{eq:loc-opt-padding}
    \end{align}
  
    We let $\lopt(l,\tau)$ be the list computed by the following iterative procedure: 
    As long as there is $j\in[N]$ with $l[j]\neq \bot$ that violates condition \cref{eq:loc-opt-padding}, consider the smallest such $j$, and if $j < d\cdot \tau(l[j])$ then
    swap elements $l[j]$ and $l[j+1]$, while if $j > d\cdot \tau(l[j])$ then
    swap elements $l[j-1]$ and $l[j]$.
    Once all $j\in[N]$ with $l[j]\neq \bot$ satisfy \cref{eq:loc-opt-padding}, output the resulting list.
    It is easy to verify that the procedure terminates, and outputs a locally optimal list w.r.t.\ $\tau$.
\end{definition}

We will use the following notions of maximum and total displacement to measure the distance of list $l_t$ from its sorted state as indicated by the target permutation $\tau_t$.

\begin{definition}[Displacements]
    \label{def:list-permutation-distances}
    For a $d$-padding $l$ 
    and a permutation $\tau\in\SS_n$, the \emph{maximum displacement} of~$l$ w.r.t.\ $\tau$ is
    \[
        \mdsp(l,\tau)
        =
        \max_{j\in [N]\colon l[j]\neq \bot} |j-d\cdot \tau(l[j])|,
    \]
    and the \emph{total displacement} of $l$ w.r.t.\ $\tau$ is
    \[
        \dsp(l,\tau) = \sum_{j\in [N]\colon l[j]\neq \bot} |j-d\cdot \tau(l[j])|
        .
    \]
    In the above notation, we may omit $\tau$ if $\tau = \id_n$. %
\end{definition}

We will use the following relations between the displacement of a $d$-padding of a permutation, and the deviation of the permutation.
The proof is given in Appendix~\ref{apndx}.

\begin{restatable}{lemma}{DevDspLem}
    \label{lem:dsp_bounds_dev}
    If $l$ is a $d$-padding of $\pi\in \SS_n$, 
    then 
    $
     \mdev(\pi) \leq \frac{2}{d} \cdot \mdsp(l)
    $
    and
    $
      \dev(\pi) \leq \dsp(l).
    $
\end{restatable}

The role of target permutations $\tau_t$ and $\sigma_t$ is to facilitate analyzing independently sorting from mixing steps.
Both evolve in a similar manner, so we focus on $\tau_t$ for now.
Roughly speaking, $\tau_t$ keeps track of all mixing steps, swapping $\tau_t(i)$ and $\tau_t(i+1)$ whenever $\pi^{-1}_t(i)$ and $\pi^{-1}_t(i+1)$ are swapped during a mixing step.
In addition, at the end of each step we may need to perform some additional swaps to $\tau_t$ to ensure that a basic \emph{admissibility condition} is satisfied.
This condition states that if $\pi_t(i)$ and $\pi_t(i+1)$ are sorted in $\pi_t$ (and thus also in $l_t$), then $\tau_t(\pi_t(i))$ and $\tau_t(\pi_t(i+1))$ must also be sorted in $\tau_t$.
This is a natural invariant to maintain, since if it is violated for some $i$, then simply swapping $\tau_t(\pi_t(i))$ and $\tau_t(\pi_t(i+1))$ in  $\tau_t$ brings $\tau_t$ closer to the identity permutation, which can only help our analysis.
We formally define admissibility, and a procedure to make a permutation admissible, as follows.

\begin{definition}[Admissibility]
    \label{def:admissible}
    For permutations $\tau,\pi \in \SS_n$, we say that $\tau$ is \emph{admissible} w.r.t.\ $\pi$ if for all $i\in[n-1]$,
    \begin{equation}
        \label{eq:admissible}
        \pi(i) < \pi(i+1) \implies \tau(\pi(i)) < \tau(\pi(i+1)).
    \end{equation}

    We let $\adm(\tau,\pi)$ be the permutation computed by the following iterative procedure: 
    As long as there is $i\in[n-1]$ that violates condition \cref{eq:admissible}, consider the smallest such $i$, and replace $\tau$ by $\swap(\tau, \pi(i),\pi(i+1))$.
    Once all $i\in[n-1]$ satisfy \cref{eq:admissible}, output the resulting permutation.
    It is easy to verify that the procedure terminates, and outputs an admissible permutation w.r.t.\ $\pi$.
\end{definition}

We are now ready to formally define the sequences of $l_t$, $\tau_t$, and $\sigma_t$. 
We will not yet discuss their initialization (for $t=0$).
In fact, our proof strategy is to occasionally intervene, e.g., to reset (or partly reset) the target distributions to the identity permutation, as described in \cref{sec:total_deviation}.
Sequences $\tau_t$ and $\sigma_t$ differ only on how their values are modified in those interventions.

\begin{definition} [Auxiliary Quantities $l_t,\tau_t,\sigma_t$]
    \label{def:l-tau-sigma}
    Let $d> 1$ be an integer.
    Let $\tau_0$ and $\sigma_0$ be permutations in $\SS_n$ that are admissible w.r.t.\ $\pi_0$,
    and let $l_0$ be a $d$-padding of $\pi_0$ that is locally optimal w.r.t.\ $\tau_0$. 
    For $t\geq 1$, let $\tau'_t$ and $\sigma'_t$ be permutations defined as follows:
    \begin{itemize}
        \item if step $t$ is a mixing step (and $i,s$ are the values sampled at that step), 
        then
        $\tau'_t$ is obtained from $\tau_{t-1}$ by successively swapping $\tau_{t-1}(i)$ with $\tau_{t-1}(i+1),\tau_{t-1}(i+2),\ldots,\tau_{t-1}(\min\{i+s,n\})$ if $s>0$, or with
        $\tau_{t-1}(i-1),
        \ldots,\tau_{t-1}(\max\{i+s,1\})$ if $s<0$;
        and $\sigma'_t$ is obtained identically from $\sigma_{t-1}$;
        
        \item
        if $t$ is a sorting step, then
        $\tau'_t = \tau_{t-1}$ and
        $\sigma'_t = \sigma_{t-1}$; 
    \end{itemize}
    Then,
    \begin{itemize}
        \item $\tau_t = \adm(\tau'_t,\pi_t)$ and 
        $\sigma_t = \adm(\sigma'_t,\pi_t)$; 
        \item $l_t =\lopt(l'_t,\tau_t)$, where $l'_t$ is a $d$-padding of $\pi_t$ with the same gaps as $l_{t-1}$.
    \end{itemize}
\end{definition}

\subsubsection{Potential Functions}

Introducing $l_t$, $\tau_t$, and $\sigma_t$ makes it possible to analyze independently the sorting steps from the mixing steps. 
Our analysis uses exponential potential functions to study the displacement of~$l_t$ w.r.t.\ $\tau_t$, and independently the deviations of $\tau_t$ and~$\sigma_t$. 
Using $l_t$ instead of $\pi_t$ is critical, as our potential functions do not work when applied to $\pi_t$ directly.
The conditions of local optimality  and admissibility are also critical, in ensuring that the potential functions are well-behaved.
Below, we define the two main potential functions we use.

\begin{definition}[Potential Functions]
    \label{def:potential-functions}   
    For any $t\geq 0$ and $j\in [N]$ such that $l_t[j]\neq \bot$, and for any $\alpha>0$,
    let
    \[
        \phi_t(j,\alpha) = 
        e^{\alpha\cdot|j - d\cdot \tau_t(l_t[j])|},
    \] 
    and
    \[
        \Phi_t (\alpha)
        = 
        \sum_{j\in [N]\colon l_t[j]\neq \bot} 
        (\phi_t(j,\alpha) - 1).
    \]
    Also, for any $i\in [n]$, let
    \[
        \psi_t(i,\alpha) = 
        e^{\alpha\cdot|\sigma_t(i) - i|},
    \]  
    and
    \[
        \Psi_t (\alpha)
        = 
        \sum_{i\in [n]} 
        \psi_t(i,\alpha)
        .
    \]
    We will generally omit parameter $\alpha$ from the above notation; the dependency on $\alpha$ will be made explicit only when needed or relevant.
\end{definition}

\subsubsection{Initialization and Interventions}

We discuss now the initialization and interventions to the auxiliary sequences $l_t,\tau_t,\sigma_t$.
Our analysis divides an execution into \emph{phases} of consecutive steps, and at the beginning of each phase we \emph{reset} the auxiliary sequences.
Specifically, in the analysis for the maximum displacement, 
at the beginning of a phase starting after step $t$, we reset $\tau_t$ and $\sigma_t$ to the identity permutation $\id_n$, and replace $l_t$ by  $\lopt(l_t,\id_n)$.
(Note that $\tau_t = \sigma_t$ throughout each phase; in fact just $\tau_t$ is used in the analysis of maximum displacement.)

In the analysis of total displacement, 
at the beginning of each phase
just~$\sigma_t$ is reset to $\id_n$, while~$\tau_t$ changes to $\hat\tau_t$ by a more refined transformation;
and $l_t$ is replaced by $\lopt(l_t,\hat\tau_t)$.
Ideally, we would like to obtain $\hat\tau_t$ from $\tau_t$ by resetting to zero all deviations $|\tau_t(i) - i|$ that are below a certain threshold $\theta$, in a way that does not increase any deviation that is above that threshold.
That way, we can ensure that $\dev(\hat\tau_t)$ is bounded by the sum of the deviations above the threshold, and that $\mdev(\hat\tau_t,\tau_t)$ is at most $\theta$.  
The actual transformation, which we call \emph{$\theta$-filtering}, is a bit different:
it keeps unchanged the elements of $\tau_t$ with deviation above $\theta$, and sorts the remaining ones.
However, it still achieves (within constant factors) the desired bounds on $\dev(\hat\tau_t)$ and $\mdev(\hat\tau_t,\tau_t)$ mentioned above, as we shown in \cref{lem:theta-filtering}.  

\begin{definition}[$\theta$-Filtering]
    \label{def:theta-filtering}
    For $\theta\geq0$, the \emph{$\theta$-filtering} 
    of permutation $\tau\in \SS_n$ is the permutation $\hat\tau\in \SS_n$ defined as follows. %
    Let $I = \{i\in[n]\colon |\tau(i) - i| > \theta\}$.
    Then,
    \begin{itemize}
        \item $\hat\tau(i) = \tau(i)$ for all $i\in I$, and
        \item $\hat\tau(i) < \hat\tau(j)$ for all $i,j\in[n]\setminus I$ with $i<j$.
    \end{itemize}
\end{definition}

\begin{restatable}{lemma}{ThetaFiltering}
\label{lem:theta-filtering}
    If $\hat\tau$ is the $\theta$-filtering of $\tau\in \SS_n$,
    then
    \[
        \mdev(\hat\tau,\tau)\leq 2\theta, 
    \]
    and
    \[
        \dev(\hat\tau)
        \leq 
        4 \sum_{i: | \tau(i) - i | > \theta} | \tau(i) - i |
        .
    \]
\end{restatable}

\begin{proof}
    We can obtain $\hat\tau$ from $\tau$ using the following iterative procedure: 
    As long as there are indices $i,j\in[n]\setminus I$ such that $i < j$ and $\tau(i) > \tau(j)$, i.e., pair $(\tau(i),\tau(j))$ is an inversion, replace $\tau$ by $\swap(\tau,i,j)$.
    By observing that each step does not increase $\max\{|\tau(i)-i|,\,|\tau(j)-j|\}$ (which is at most $\theta$), it follows by \cref{lem:admissible_swaps_are_good}~$(ii)$ that $|\hat\tau(i)-i|\leq \theta$ for all $i\in [n]\setminus I$.
    Therefore, for every $i \in [n] \setminus I$,
    \[
        |\hat\tau(i)-\tau(i)|
        \leq
        |\hat\tau(i)-i| + |\tau(i)-i|
        \leq
        2\theta.
    \]
    And since $\hat\tau(i) = \tau(i)$ for $i\in I$, it follows $\mdev(\hat\tau,\tau)\leq 2\theta$.

    To bound the total deviation of $\hat\tau$, we first bound its Kendall-tau distance from $\id_n$, $K(\hat\tau)$.
    Let $P = \{(i,j)\colon i<j,\, \hat\tau(i) > \hat\tau(j)\}$ be the set of index pairs for which we have an inversion; clearly $K(\hat\tau) = |P|$.
    For each $(i,j) \in P$, we have that at least one of $i,j$ is in $I$, by the definition of $\hat\tau$.

    Let $i\in I$, and suppose that $\hat\tau(i) > i$ (the case $\hat\tau(i) < i$ is symmetric). 
    We bound the number of pairs $(j,i),(i,j)\in P$ such that $|\hat\tau(j)-j| \leq |\hat\tau(i)-i|$.
    Then, summing over all $i\in I$ will bound~$|P|$.  
    Note that  $(j,i)\notin P$ for all $j$ such that $|\hat\tau(j)-j| \leq |\hat\tau(i)-i|$, because $j < i$ implies 
    \[
        \hat\tau(j) \leq j + |\hat\tau(j)-j| < i + |\hat\tau(i) - i| = \hat\tau(i). 
    \]
    On the other hand, there are at most $2\cdot |\hat\tau(i) - i|$ pairs $(i,j)\in P$ with $|\hat\tau(j)-j| \leq |\hat\tau(i)-i|$, since for $j > \hat\tau(i)+|\hat\tau(i) - i|$,
    \[
        \hat\tau(j) \geq j - |\hat\tau(j)-j| > \hat\tau(i).  
    \]

    It follows that 
    \[
        |P| \leq \sum_{i\in I} (2\cdot |\hat\tau(i) - i|)
        =
        2\sum_{i\in I} |\tau(i) - i|.
    \]
    Finally, from \cref{eq:Kvsdis},
    \[ 
        \dev(\hat\tau) \leq 2K(\hat\tau) = 2|P| \leq 4\sum_{i\in I} |\tau(i) - i|.
        \qedhere
    \]
\end{proof}

\subsection{Blocks}

For the analysis of the displacement of list $l_t$, we decompose the list into sorted \emph{blocks}, defined as follows (see \cref{fig:left_and_right_block_examples} for an illustration).

\begin{definition}[Blocks]
    \label{def:block}
    Let $\pi,\tau \in \SS_n$ be permutations such that $\tau$ is admissible w.r.t.\ $\pi$, and let~$l$ be a $d$-padding of $\pi$ that is locally optimal w.r.t.\ $\tau$. 
    A \emph{right block} of list~$l$ is a contiguous sublist $l[j..k]$
    satisfying the following properties:
    \begin{enumerate}
        \item (absence of gaps) 
        $l[i]\neq\bot$, for all $j\leq i\leq k$, 
        \item (target monotonicity) $\tau(l[j]) < \tau(l[j+1])<\dots< \tau(l[k])$,
        \item (targets to the right) $d \cdot \tau(l[j]) > j$,
        \item (maximality) decreasing $j$'s value or increasing $k$'s violates one of the previous three properties.
    \end{enumerate}
    The \emph{head} of right block $l[j..k]$ is its rightmost element~$l[k]$, and its \emph{tail} is the leftmost element $l[j]$.
    We define a \emph{left block} $l[j.. k]$ similarly, by replacing property~$3)$ by 
    \begin{itemize}
        \item [$3')$] (targets to the left)
        $d\cdot \tau(l[k]) < k$.
    \end{itemize}
    The \emph{head} of left block $l[j..k]$ is $l[j]$ and its \emph{tail} is $l[k]$.
    
    A \emph{stationary block} is a single element $l[i]$ for which $d\cdot\tau(l[i]) = i$. 
    We refer to elements belonging to stationary blocks as \emph{stationary elements}.
\end{definition}

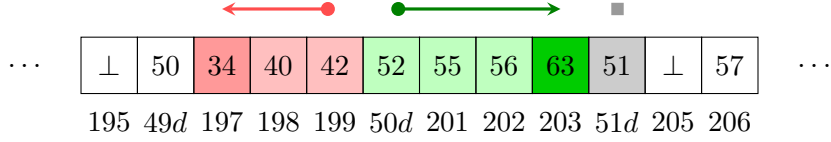
\begin{figure}[t]
\centering
\scalebox{0.68}{
\begin{tikzpicture}
\newcommand{\Side}{0.75}

\fill[color=red!40!white] (2 * \Side,0) rectangle ++(\Side,\Side);
\fill[color=red!25!white] (3 * \Side,0) rectangle ++(\Side,\Side);
\fill[color=red!25!white] (4 * \Side,0) rectangle ++(\Side,\Side);
\fill[color=green!25!white] (5 * \Side,0) rectangle ++(\Side,\Side);
\fill[color=green!25!white] (6 * \Side,0) rectangle ++(\Side,\Side);
\fill[color=green!25!white] (7 * \Side,0) rectangle ++(\Side,\Side);
\fill[color=green!80!black] (8 * \Side,0) rectangle ++(\Side,\Side);
\fill[color=white!80!black] (9 * \Side,0) rectangle ++(\Side,\Side);

\foreach \x in {0,...,11} 
  \draw[draw=black] (\x * \Side,0) rectangle ++(\Side,\Side);

\foreach \x in {195,...,206} 
  \node at (-0.5 * \Side + \x * \Side - 194 * \Side, -0.5 * \Side) {$\x$};

\node[fill=white] at (-0.5 * \Side + 2 * \Side, -0.5 * \Side) {$49d$};
\node[fill=white] at (-0.5 * \Side + 6 * \Side, -0.5 * \Side) {$50d$};
\node[fill=white] at (-0.5 * \Side + 10 * \Side, -0.5 * \Side) {$51d$};

\node at (-1 * \Side, 0.5 * \Side) {$\ldots$};
\node at (0.5 * \Side, 0.5 * \Side) {$\bot$};
\node at (1.5 * \Side, 0.5 * \Side) {$50$};
\node at (2.5 * \Side, 0.5 * \Side) {$34$};
\node at (3.5 * \Side, 0.5 * \Side) {$40$};
\node at (4.5 * \Side, 0.5 * \Side) {$42$};
\node at (5.5 * \Side, 0.5 * \Side) {$52$};
\node at (6.5 * \Side, 0.5 * \Side) {$55$};
\node at (7.5 * \Side, 0.5 * \Side) {$56$};
\node at (8.5 * \Side, 0.5 * \Side) {$63$};
\node at (9.5 * \Side, 0.5 * \Side) {$51$};
\node at (10.5 * \Side, 0.5 * \Side) {$\bot$};
\node at (11.5 * \Side, 0.5 * \Side) {$57$};
\node at (13 * \Side, 0.5 * \Side) {$\ldots$};

\draw[{Circle[scale=0.85]}-stealth,black!50!green,very thick] (5.5 * \Side, 1.5 * \Side) -- (8.5 * \Side, 1.5 * \Side);

\draw[stealth-{Circle[scale=0.85]},very thick,white!30!red] (2.5 * \Side, 1.5 * \Side) -- (4.5 * \Side, 1.5 * \Side);

\draw[{Square[scale=0.85]}-,white!60!black,very thick] (9.4 * \Side, 1.5 * \Side) -- (9.5 * \Side, 1.5 * \Side);

\end{tikzpicture}}
    \caption{An example of a left block (shown in \textcolor{red}{red}), a right block (shown in \textcolor{black!50!green}{green}), and a stationary block (shown in \textcolor{white!40!black}{gray}) of list $l$, where $d = 4$ and  $\tau = \id_n$.
    The head of the left block is denoted by darker red and the head of the right block with darker green.}
    \label{fig:left_and_right_block_examples}
\end{figure}

In the next lemma, we show the crucial property that
the contribution to the potential function of the head of a block is at least a constant factor of the contribution of the entire block. 
We will make use of this property in proving the drop inequality in \cref{lem:new_sorting_step_drift} for the $\Phi$ potential. Intuitively, the larger the number $d-1$ of gaps between two adjacent targets, the more dominant the contribution of the head compared to that of the entire block.

\begin{restatable}{lemma}{HeadPotential}
\label{lem:potential_of_header_dominates}
Consider an arbitrary list $l$ and the potential function $\Phi_t = \Phi_t(\alpha)$ on this list for any real $\alpha > 0$ and any integer $d > 1$. For every step $t \geq 0$ and right block $l_t[i..j]$, it holds that
\begin{align*}
\sum_{k = i}^j \phi_t(k) 
 & \leq \phi_t(j) \cdot \frac{1}{1 - e^{-\alpha \cdot (d-1)}}.
\end{align*}
Similarly, for every left block $l_t[i..j]$, it holds that
\begin{align*}
\sum_{k = i}^j \phi_t(k) 
 & \leq \phi_t(i) \cdot \frac{1}{1 - e^{-\alpha \cdot (d-1)}}.
\end{align*}
\end{restatable}

\begin{proof}
Consider a right block $l_t[i..j]$ and any index $k \in [i, j)$. Then, we have that $k < j$ (by definition), $\tau_t(l_t[k]) < \tau_t(l_t[j])$ (by definition of a right block) and $\tau_t(l_t[j]) - \tau_t(l_t[k]) \geq j - k$. Therefore, we can upper bound the displacement of $k$ as follows
\begin{align*}
d \cdot \tau_t(l_t[k]) - k 
 & \leq d \cdot (\tau_t(l_t[j]) - j + k) - k \notag \\
 & = d \cdot \tau_t(l_t[j]) - j - (d-1) \cdot (j - k) \label{eq:upper_bound_displacement}.
\end{align*}
Hence, the contribution of the block is given by
\begin{align*}
\sum_{k = i}^j \phi_t(k)
 & = \sum_{k = i}^j e^{\alpha \cdot (d \cdot \tau_t(l_t[k]) - k)} \\
 & %
 \leq\sum_{k = i}^j e^{\alpha \cdot (d \cdot \tau(l_t[j]) - j - (d-1) \cdot (j - k))} \\
 & = \phi_t(j) \cdot \sum_{k = i}^j e^{- \alpha \cdot (d-1)  \cdot (j - k) } \\
 & = \phi_t(j) \cdot \frac{1 - e^{-\alpha \cdot (d-1) \cdot (j - i + 1)}}{1 - e^{-\alpha \cdot (d-1)}} \\
 & \leq \phi_t(j) \cdot \frac{1}{1 - e^{-\alpha \cdot (d-1)}}.
\end{align*}
By a symmetric argument and bounding instead the displacement $k - d\cdot \tau_t(l_t[k])$, we also get the claim for the left block.
\end{proof}

Some of our proofs in the sections to come hold only for sufficiently large values of $n$.

\section{Tight Bounds on the Maximum Deviation} \label{sec:max_deviation}

In this section, we give a bound of $O(b \cdot \log n)$ for the maximum displacement. In  \cref{sec:lower_bounds} (\cref{lem:max_and_total_deviation_lb}), we prove an  asymptotically matching lower bound for the process in~\cite{AnagnostopoulosKMU11} with any $b \in \big[1, \frac{n}{\log^2 n}\big]$. In \cref{sec:mdev_phi_potential}, we will prove a drift inequality for the $\Phi$ potential, establishing that in sorting steps it drops in expectation by a multiplicative factor of $1 - \Omega(\alpha/n)$. In \cref{sec:mdev_psi_potential}, we will analyze the evolution of $\sigma$ as a result of mixing steps, this time by analyzing the $\Psi$ potential. Finally, in \cref{sec:mdev_main_theorem}, we combine the two analyses to deduce the bound on the maximum deviation.

\subsection{Analysis of the \texorpdfstring{$\Phi$}{Φ} Potential} \label{sec:mdev_phi_potential}

Our main tool will be the $\Phi$ potential, which for a smoothing parameter $\alpha$ and padding length of $d-1$ was defined as
\begin{align*}
  \Phi_t 
   = \sum_{j \in [N] : l_t[j] \neq \bot } (\phi_t(j) - 1)  
   = \sum_{j \in [N] : l_t[j] \neq \bot } (e^{\alpha \cdot |j - d \cdot \tau_t(l_t[j])|} - 1).
\end{align*}

Our goal is to prove the following drift inequality for~$\Phi$ and establish that it is non-increasing.

\begin{lemma}[Drift Inequality] \label{lem:new_sorting_step_drift}
Consider the auxiliary process $(l_t,\tau_t,\sigma_t)_{t\geq 0}$ and the potential function $\Phi_t = \Phi_t(\alpha)$ with any integer $d > 1$ and any real $\alpha \in \big[ \frac{\log 20}{d-1}, 1 \big]$. Further, let $\mathcal{G}_{t+1}$ be the event that in step~$t+1$ we do sorting. Then, for any step $t \geq 0$, it holds that~$\Phi_{t+1} \leq \Phi_t$ and
\[
 \Ex{ \left. \Phi_{t+1} \, \right| \, \mathcal{F}_t} 
   \leq \Phi_t \cdot \left( 1 - \frac{\alpha}{4 (n-1)} \cdot \Ex{\mathbf{1}_{\mathcal{G}_{t+1}} \middle| \mathcal{F}_t} \right).
\]
\end{lemma}

Regarding the constraints on $\alpha$ and $d$, the intuition is that~$\alpha$ needs to be small enough so that we can apply some Taylor estimates, and $d$ needs to be large enough so that large displacements are amplified.

\begin{proof}
In any mixing step the value of the potential cannot increase, as the first step of the $\adm$ operation also swaps the targets of the elements mixed. These include the following swaps between adjacent elements:

\begin{itemize}
  \item Swapping two elements of the same right (left) block (\cref{subfig:right_internal,subfig:left_internal}).
  \item Swapping two stationary elements (\cref{subfig:stationary}).
  \item Swapping a stationary element with the tail of a left (right) block (\cref{subfig:stationary_and_right,subfig:left_and_stationary}).
\end{itemize}

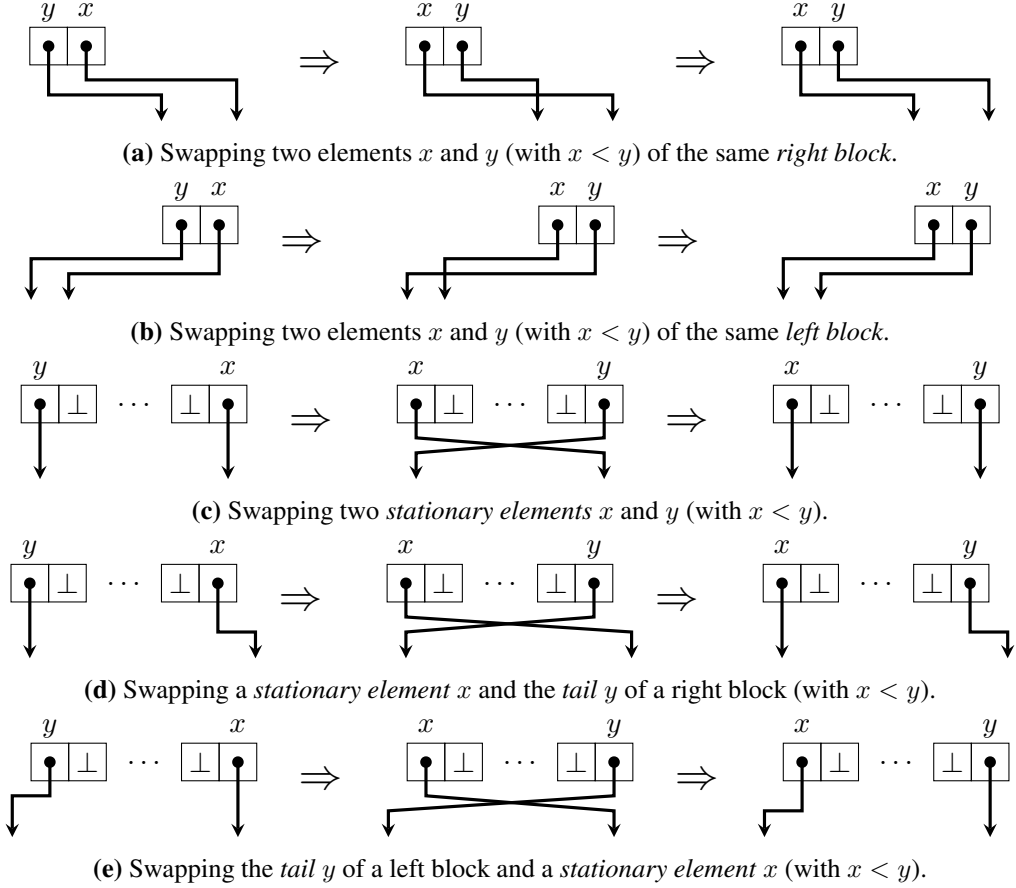
\begin{figure*}[t]
\centering
\begin{subfigure}[t]{\textwidth}
\centering
\begin{tikzpicture}
\newcommand{\Side}{0.5}
\foreach \x in {0,...,1} 
  \draw[draw=black] (\x * \Side,0) rectangle ++(\Side,\Side);

\node[anchor=north] at (0.5 * \Side, 0.95) {$y$};
\node[anchor=north] at (1.5 * \Side, 0.95) {$x$};
\SplitLine{0.5}{-0.8}{3.5}
\SplitLine{1.5}{-0.4}{5.5}

\node at (3.85, 0) {\scalebox{1.5}{$\Rightarrow$}};

\begin{scope}[shift={(5, 0)}]
\foreach \x in {0,...,1} 
  \draw[draw=black] (\x * \Side,0) rectangle ++(\Side,\Side);

\node[anchor=north] at (0.5 * \Side, 0.95) {$x$};
\node[anchor=north] at (1.5 * \Side, 0.95) {$y$};
\SplitLine{0.5}{-0.8}{5.5}
\SplitLine{1.5}{-0.4}{3.5}

\node at (3.85, 0) {\scalebox{1.5}{$\Rightarrow$}};
\end{scope}

\begin{scope}[shift={(10, 0)}]
\foreach \x in {0,...,1} 
  \draw[draw=black] (\x * \Side,0) rectangle ++(\Side,\Side);

\node[anchor=north] at (0.5 * \Side, 0.95) {$x$};
\node[anchor=north] at (1.5 * \Side, 0.95) {$y$};
\SplitLine{0.5}{-0.8}{3.5}
\SplitLine{1.5}{-0.4}{5.5}
\end{scope}
\end{tikzpicture}
\caption{Swapping two elements $x$ and $y$ (with $x < y$) of the same \emph{right block}.}
\label{subfig:right_internal}
\end{subfigure}

\begin{subfigure}[t]{\textwidth}
\centering
\begin{tikzpicture}
\newcommand{\Side}{0.5}
\foreach \x in {0,...,1} 
  \draw[draw=black] (2 + \x * \Side,0) rectangle ++(\Side,\Side);

\node[anchor=north] at (4.5 * \Side, 0.95) {$y$};
\node[anchor=north] at (5.5 * \Side, 0.95) {$x$};
\SplitLine{4.5}{-0.4}{0.5}
\SplitLine{5.5}{-0.8}{1.5}

\node at (3.85, 0) {\scalebox{1.5}{$\Rightarrow$}};

\begin{scope}[shift={(5, 0)}]
\foreach \x in {0,...,1} 
  \draw[draw=black] (2 + \x * \Side,0) rectangle ++(\Side,\Side);

\node[anchor=north] at (4.5 * \Side, 0.95) {$x$};
\node[anchor=north] at (5.5 * \Side, 0.95) {$y$};
\SplitLine{5.5}{-0.8}{0.5}
\SplitLine{4.5}{-0.4}{1.5}

\node at (3.85, 0) {\scalebox{1.5}{$\Rightarrow$}};
\end{scope}

\begin{scope}[shift={(10, 0)}]
\foreach \x in {0,...,1} 
  \draw[draw=black] (2 + \x * \Side,0) rectangle ++(\Side,\Side);

\node[anchor=north] at (4.5 * \Side, 0.95) {$x$};
\node[anchor=north] at (5.5 * \Side, 0.95) {$y$};
\SplitLine{5.5}{-0.8}{1.5}
\SplitLine{4.5}{-0.4}{0.5}
\end{scope}
\end{tikzpicture}
\caption{Swapping two elements $x$ and $y$ (with $x < y$) of the same \emph{left block}.}
\label{subfig:left_internal}
\end{subfigure}

\begin{subfigure}[t]{\textwidth}
\centering
\begin{tikzpicture}
\newcommand{\Side}{0.5}
\foreach \x in {0,...,1} 
  \draw[draw=black] (\x * \Side,0) rectangle ++(\Side,\Side);
\foreach \x in {4,...,5} 
  \draw[draw=black] (\x * \Side,0) rectangle ++(\Side,\Side);

\node[anchor=north] at (0.5 * \Side, 0.95) {$y$};
\node[anchor=north] at (5.5 * \Side, 0.95) {$x$};
\node at (1.5 * \Side, 0.5 * \Side) {$\bot$};
\node at (4.5 * \Side, 0.5 * \Side) {$\bot$};
\node at (3 * \Side, 0.5 * \Side) {$\ldots$};
\SplitLine{0.5}{-0.4}{0.5}
\SplitLine{5.5}{-0.8}{5.5}

\node at (3.85, 0) {\scalebox{1.5}{$\Rightarrow$}};

\begin{scope}[shift=({5, 0})]
\foreach \x in {0,...,1} 
  \draw[draw=black] (\x * \Side,0) rectangle ++(\Side,\Side);
\foreach \x in {4,...,5} 
  \draw[draw=black] (\x * \Side,0) rectangle ++(\Side,\Side);

\node[anchor=north] at (0.5 * \Side, 0.95) {$x$};
\node[anchor=north] at (5.5 * \Side, 0.95) {$y$};
\node at (1.5 * \Side, 0.5 * \Side) {$\bot$};
\node at (4.5 * \Side, 0.5 * \Side) {$\bot$};
\node at (3 * \Side, 0.5 * \Side) {$\ldots$};
\FourSplitLine{0.5}{-0.4}{-0.8}{5.5}
\FourSplitLine{5.5}{-0.4}{-0.8}{0.5}

\node at (3.85, 0) {\scalebox{1.5}{$\Rightarrow$}};
\end{scope}

\begin{scope}[shift=({10, 0})]
\foreach \x in {0,...,1} 
  \draw[draw=black] (\x * \Side,0) rectangle ++(\Side,\Side);
\foreach \x in {4,...,5} 
  \draw[draw=black] (\x * \Side,0) rectangle ++(\Side,\Side);

\node[anchor=north] at (0.5 * \Side, 0.95) {$x$};
\node[anchor=north] at (5.5 * \Side, 0.95) {$y$};
\node at (1.5 * \Side, 0.5 * \Side) {$\bot$};
\node at (4.5 * \Side, 0.5 * \Side) {$\bot$};
\node at (3 * \Side, 0.5 * \Side) {$\ldots$};
\SplitLine{0.5}{-0.8}{0.5}
\SplitLine{5.5}{-0.8}{5.5}
\end{scope}
\end{tikzpicture}
\caption{Swapping two \emph{stationary elements} $x$ and $y$ (with $x < y$).}
\label{subfig:stationary}
\end{subfigure}

\begin{subfigure}[t]{\textwidth}
\centering
\begin{tikzpicture}
\newcommand{\Side}{0.5}
\foreach \x in {0,...,1} 
  \draw[draw=black] (\x * \Side,0) rectangle ++(\Side,\Side);
\foreach \x in {4,...,5} 
  \draw[draw=black] (\x * \Side,0) rectangle ++(\Side,\Side);

\node[anchor=north] at (0.5 * \Side, 0.95) {$y$};
\node[anchor=north] at (5.5 * \Side, 0.95) {$x$};
\node at (1.5 * \Side, 0.5 * \Side) {$\bot$};
\node at (4.5 * \Side, 0.5 * \Side) {$\bot$};
\node at (3 * \Side, 0.5 * \Side) {$\ldots$};
\SplitLine{0.5}{-0.4}{0.5}
\SplitLine{5.5}{-0.8}{6.5}

\node at (3.85, 0) {\scalebox{1.5}{$\Rightarrow$}};

\begin{scope}[shift=({5, 0})]
\foreach \x in {0,...,1} 
  \draw[draw=black] (\x * \Side,0) rectangle ++(\Side,\Side);
\foreach \x in {4,...,5} 
  \draw[draw=black] (\x * \Side,0) rectangle ++(\Side,\Side);

\node[anchor=north] at (0.5 * \Side, 0.95) {$x$};
\node[anchor=north] at (5.5 * \Side, 0.95) {$y$};
\node at (1.5 * \Side, 0.5 * \Side) {$\bot$};
\node at (4.5 * \Side, 0.5 * \Side) {$\bot$};
\node at (3 * \Side, 0.5 * \Side) {$\ldots$};
\FourSplitLine{0.5}{-0.4}{-0.8}{6.5}
\FourSplitLine{5.5}{-0.4}{-0.8}{0.5}

\node at (3.85, 0) {\scalebox{1.5}{$\Rightarrow$}};
\end{scope}

\begin{scope}[shift=({10, 0})]
\foreach \x in {0,...,1} 
  \draw[draw=black] (\x * \Side,0) rectangle ++(\Side,\Side);
\foreach \x in {4,...,5} 
  \draw[draw=black] (\x * \Side,0) rectangle ++(\Side,\Side);

\node[anchor=north] at (0.5 * \Side, 0.95) {$x$};
\node[anchor=north] at (5.5 * \Side, 0.95) {$y$};
\node at (1.5 * \Side, 0.5 * \Side) {$\bot$};
\node at (4.5 * \Side, 0.5 * \Side) {$\bot$};
\node at (3 * \Side, 0.5 * \Side) {$\ldots$};
\SplitLine{0.5}{-0.8}{0.5}
\SplitLine{5.5}{-0.8}{6.5}
\end{scope}
\end{tikzpicture}
\caption{Swapping a \emph{stationary element} $x$ and the \emph{tail} $y$ of a right block (with $x < y$).}
\label{subfig:stationary_and_right}
\end{subfigure}

\begin{subfigure}[t]{\textwidth}
\centering
\begin{tikzpicture}
\newcommand{\Side}{0.5}
\foreach \x in {0,...,1} 
  \draw[draw=black] (\x * \Side,0) rectangle ++(\Side,\Side);
\foreach \x in {4,...,5} 
  \draw[draw=black] (\x * \Side,0) rectangle ++(\Side,\Side);

\node[anchor=north] at (0.5 * \Side, 0.95) {$y$};
\node[anchor=north] at (5.5 * \Side, 0.95) {$x$};
\node at (1.5 * \Side, 0.5 * \Side) {$\bot$};
\node at (4.5 * \Side, 0.5 * \Side) {$\bot$};
\node at (3 * \Side, 0.5 * \Side) {$\ldots$};
\SplitLine{0.5}{-0.4}{-0.5}
\SplitLine{5.5}{-0.8}{5.5}

\node at (3.85, 0) {\scalebox{1.5}{$\Rightarrow$}};

\begin{scope}[shift=({5, 0})]
\foreach \x in {0,...,1} 
  \draw[draw=black] (\x * \Side,0) rectangle ++(\Side,\Side);
\foreach \x in {4,...,5} 
  \draw[draw=black] (\x * \Side,0) rectangle ++(\Side,\Side);

\node[anchor=north] at (0.5 * \Side, 0.95) {$x$};
\node[anchor=north] at (5.5 * \Side, 0.95) {$y$};
\node at (1.5 * \Side, 0.5 * \Side) {$\bot$};
\node at (4.5 * \Side, 0.5 * \Side) {$\bot$};
\node at (3 * \Side, 0.5 * \Side) {$\ldots$};
\FourSplitLine{0.5}{-0.4}{-0.8}{5.5}
\FourSplitLine{5.5}{-0.4}{-0.8}{-0.5}

\node at (3.85, 0) {\scalebox{1.5}{$\Rightarrow$}};
\end{scope}

\begin{scope}[shift=({10, 0})]
\foreach \x in {0,...,1} 
  \draw[draw=black] (\x * \Side,0) rectangle ++(\Side,\Side);
\foreach \x in {4,...,5} 
  \draw[draw=black] (\x * \Side,0) rectangle ++(\Side,\Side);

\node[anchor=north] at (0.5 * \Side, 0.95) {$x$};
\node[anchor=north] at (5.5 * \Side, 0.95) {$y$};
\node at (1.5 * \Side, 0.5 * \Side) {$\bot$};
\node at (4.5 * \Side, 0.5 * \Side) {$\bot$};
\node at (3 * \Side, 0.5 * \Side) {$\ldots$};
\SplitLine{0.5}{-0.8}{-0.5}
\SplitLine{5.5}{-0.8}{5.5}
\end{scope}
\end{tikzpicture}
\caption{Swapping the \emph{tail} $y$ of a left block and a \emph{stationary element} $x$ (with $x < y$).}
\label{subfig:left_and_stationary}
\end{subfigure}
\caption{The five cases of swaps between adjacent element for which the value of the potential remains the same (all of which correspond to mixing steps). The figures show the following three stages: $(i)$~the original list $l_t$ (with targets $\tau_t$), $(ii)$~the list $l_t'$ after the mixing step (with targets $\tau_t'$), and $(iii)$~the list $l_{t+1}$ after the $\adm$ operation (with targets $\tau_{t+1}$).} %
\label{fig:cases_with_no_change}
\end{figure*}

Now the remaining swaps between adjacent elements correspond to sorting steps. We will only consider the case for swapping the \emph{head} $j$ of a right block as the case for swapping the head of a left block is symmetric. Further note that because of the $\lopt$ operation, the next element $l_t[j+1] \neq \bot$. Further, $l_t[j + 1] < l_t[j]$ as otherwise $\tau_t(l_t[j+1]) > \tau_t(l_t[j])$ and $l_t[i..j]$ would not be a right block.
\begin{itemize}
  \item When swapping $j$ with the head of a left block, then both displacements decrease (\cref{subfig:right_head_left}).
 \item When swapping $j$ with the tail of a right block, then only the displacement of the head decreases by $1$ and the other displacement increases by $1$ (\cref{subfig:right_head_right}). As we will show below, the change of the largest displacement (i.e., that of the head element) dominates and so overall the potential decreases. 
 \item When swapping $j$ with a stationary element, then again the displacement of the head decreases by~$1$ and the other displacement increases by $1$ (\cref{subfig:right_head_stationary}).
\end{itemize}

Further, note that since the permutation $\tau_t$ is admissible w.r.t.~$\pi_t$, we also have that the head of a right block~$j$ can always be swapped with its element to the right and that element is not a gap.

By the case analysis above, it follows that none of the swaps can increase the potential. Also, the $\lopt$ operation cannot increase the potential and so we deduce that $\Phi_{t+1} \leq \Phi_t$ for any step $t \geq 0$.

Our goal will be to show that the expected decrease of the value of the potential for the head $j$ of the block covers $(i)$ the possible increase of the potential for the element to its right, and $(ii)$ also gives the decrease factor for all elements of its block. We will assume the pessimistic case where the displacement of the element to its right does increase.

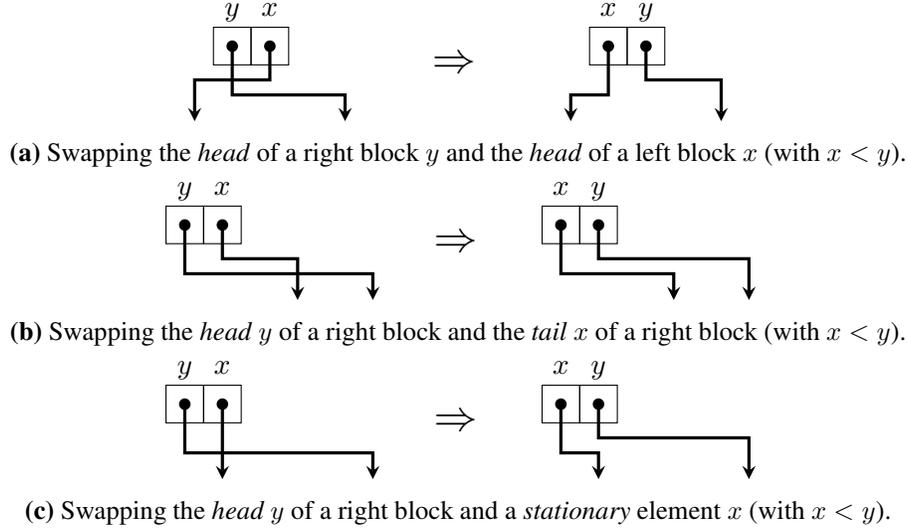
\begin{figure*}[t]
\centering

\begin{subfigure}[t]{\textwidth}
\centering
\begin{tikzpicture}
\newcommand{\Side}{0.5}
\foreach \x in {0,...,1} 
  \draw[draw=black] (\x * \Side,0) rectangle ++(\Side,\Side);

\node[anchor=north] at (0.5 * \Side, 0.95) {$y$};
\node[anchor=north] at (1.5 * \Side, 0.95) {$x$};
\SplitLine{0.5}{-0.8}{3.5}
\SplitLine{1.5}{-0.4}{-0.5}

\node at (3.2, 0) {\scalebox{1.5}{$\Rightarrow$}};

\begin{scope}[shift={(5, 0)}]
\foreach \x in {0,...,1} 
  \draw[draw=black] (\x * \Side,0) rectangle ++(\Side,\Side);

\node[anchor=north] at (0.5 * \Side, 0.95) {$x$};
\node[anchor=north] at (1.5 * \Side, 0.95) {$y$};
\SplitLine{0.5}{-0.8}{-0.5}
\SplitLine{1.5}{-0.4}{3.5}

\end{scope}
\end{tikzpicture}
\caption{Swapping the \emph{head} of a right block $y$ and the \emph{head} of a left block $x$ (with $x < y$).}
\label{subfig:right_head_left}
\end{subfigure}

\begin{subfigure}[t]{\textwidth}
\centering
\begin{tikzpicture}
\newcommand{\Side}{0.5}
\foreach \x in {0,...,1} 
  \draw[draw=black] (\x * \Side,0) rectangle ++(\Side,\Side);

\node[anchor=north] at (0.5 * \Side, 0.95) {$y$};
\node[anchor=north] at (1.5 * \Side, 0.95) {$x$};
\SplitLine{0.5}{-0.8}{5.5}
\SplitLine{1.5}{-0.4}{3.5}

\node at (3.85, 0) {\scalebox{1.5}{$\Rightarrow$}};

\begin{scope}[shift={(5, 0)}]
\foreach \x in {0,...,1} 
  \draw[draw=black] (\x * \Side,0) rectangle ++(\Side,\Side);

\node[anchor=north] at (0.5 * \Side, 0.95) {$x$};
\node[anchor=north] at (1.5 * \Side, 0.95) {$y$};
\SplitLine{0.5}{-0.8}{3.5}
\SplitLine{1.5}{-0.4}{5.5}

\end{scope}
\end{tikzpicture}
\caption{Swapping the \emph{head} $y$ of a right block and the \emph{tail} $x$ of a right block (with $x < y$).}
\label{subfig:right_head_right}
\end{subfigure}

\begin{subfigure}[t]{\textwidth}
\centering
\begin{tikzpicture}
\newcommand{\Side}{0.5}
\foreach \x in {0,...,1} 
  \draw[draw=black] (\x * \Side,0) rectangle ++(\Side,\Side);

\node[anchor=north] at (0.5 * \Side, 0.95) {$y$};
\node[anchor=north] at (1.5 * \Side, 0.95) {$x$};
\SplitLine{0.5}{-0.8}{5.5}
\SplitLine{1.5}{-0.4}{1.5}

\node at (3.85, 0) {\scalebox{1.5}{$\Rightarrow$}};

\begin{scope}[shift={(5, 0)}]
\foreach \x in {0,...,1} 
  \draw[draw=black] (\x * \Side,0) rectangle ++(\Side,\Side);

\node[anchor=north] at (0.5 * \Side, 0.95) {$x$};
\node[anchor=north] at (1.5 * \Side, 0.95) {$y$};
\SplitLine{0.5}{-0.8}{1.5}
\SplitLine{1.5}{-0.4}{5.5}

\end{scope}
\end{tikzpicture}
\caption{Swapping the \emph{head} $y$ of a right block and a \emph{stationary} element $x$ (with $x < y$).}
\label{subfig:right_head_stationary}
\end{subfigure}

\caption{The three cases of sorting steps involving swapping the \emph{head} of a right block.}
\end{figure*}

Consider a right block $l_t[i..j]$.  As mentioned above, we can assume that when swapping the head $j$ of the block with the  next element to its right, its displacement decreases by $1$. Hence, we have that
\begin{align*}
 \Exp\left[ \left. \phi_{t+1}(j) \, \right| \, \mathcal{F}_t, \mathcal{G}_{t+1} \right] 
 & \leq \left( 1  - \frac{1}{n-1} \right) \cdot \phi_t(j) + \frac{1}{n-1} \cdot \phi_t(j) \cdot e^{-\alpha}  \\
 & = \phi_t(j) \cdot \left( 1  + \frac{1}{n-1} \cdot (e^{-\alpha} - 1) \right)  \\
 & \leq \phi_t(j) \cdot \left( 1  - \frac{\alpha}{2(n-1)} \right),
\end{align*}
using the Taylor estimate $e^{-\alpha} \leq 1 - \alpha/2$ (for any $\alpha \leq 1$).

Further, we can pessimistically assume that the displacement of element $j+1$ increases by $1$. So,
\begin{align}
\Exp\left[ \left. \phi_{t+1}(j+1) \, \right| \, \mathcal{F}_t, \mathcal{G}_{t+1} \right] 
 & \leq \left( 1  - \frac{1}{n-1} \right) \cdot \phi_t(j+1) + \frac{1}{n-1} \cdot \phi_t(j+1) \cdot e^{\alpha} \notag \\
 & = \phi_t(j+1) \cdot \left( 1  + \frac{1}{n-1} \cdot (e^{\alpha} - 1) \right) \notag \\
 & \leq \phi_t(j+1) \cdot \left( 1  + \frac{2\alpha}{n-1} \right), \label{eq:phi_j_plus_one}
\end{align}
using the Taylor estimate $e^{\alpha} \leq 1 + 2\alpha$ (for any $\alpha \leq 1$).

Moreover, in the cases where the displacement of $j$ increases by $1$ (\cref{subfig:right_head_right,subfig:right_head_stationary}), then $\tau_t(l_t[j]) \geq \tau_t(l_t[j-1]) + d$ and so $|d \cdot \tau_t(l_t[j]) - j| \geq |d \cdot \tau_t(l_t[j+1]) - (j+1)| + d-1$. Therefore, we can bound its contribution as follows
\[
\phi_t(j+1) \leq \phi_t(j) \cdot e^{-\alpha (d-1)} \leq \frac{1}{20} \cdot \phi_t(j),
\]
using that $\alpha \cdot (d-1) \geq \log 20$. Returning to~\cref{eq:phi_j_plus_one},
\begin{align} 
 \Ex{ \left. \phi_{t+1}(j+1) \, \right| \, \mathcal{F}_t, \mathcal{G}_{t+1} } 
  & \leq \phi_t(j+1) \cdot \left( 1 - \frac{\alpha}{4(n-1)} \right)
  + \phi_t(j+1) \cdot \frac{9 \alpha}{4(n-1)} \notag \\
  & \leq \phi_t(j+1) \cdot \left( 1 - \frac{\alpha}{4(n-1)} \right)
  + \phi_t(j) \cdot \frac{\alpha}{8(n-1)}. \label{eq:sorting_step_next_element}
\end{align}
Now, for the contribution of the elements of the block~$l_t[i..j]$ we have that

\begin{align}
\Ex{\left. \sum_{\ell = i}^{j} \phi_{t+1}(\ell) \, \right| \, \mathcal{F}_t, \mathcal{G}_{t+1} } 
 &  = \sum_{\ell = i}^{j-1} \phi_t(\ell) +  \phi_t(j) \cdot \left( 1  - \frac{\alpha}{2(n-1)} \right) \notag \\
 &  = \sum_{\ell = i}^{j}  \phi_t(\ell) \cdot \left( 1 - \frac{\alpha}{4(n-1)}  \right) + \sum_{\ell = i}^{j-1}  \phi_t(\ell) \cdot \frac{\alpha}{4 (n-1)}  - \phi_t(j) \cdot \frac{\alpha}{4 (n-1)} \notag \\
 &  \leq \sum_{\ell = i}^{j}  \phi_t(\ell) \cdot \left( 1 - \frac{\alpha}{4 (n-1)} \right) - \phi_t(j) \cdot \frac{\alpha}{8(n-1)}, \label{eq:sorting_step_block}
\end{align}
using in the last inequality that \begin{align*}
\sum_{\ell = i}^{j-1}  \phi_t(\ell) \cdot \frac{\alpha}{4(n-1)} 
  &  \stackrel{(a)}{\leq} \phi_t(j) \cdot \left( \frac{1}{1 - e^{-\alpha (d-1)}}  - 1\right) \cdot \frac{\alpha}{4 (n-1)} \\
  &  = \phi_t(j) \cdot \frac{e^{-\alpha (d-1)}}{1 - e^{-\alpha (d-1)}} \cdot \frac{\alpha}{4 (n-1)} \\
  &  \stackrel{(b)}{\leq} \phi_t(j) \cdot \frac{\alpha}{8 (n-1)},
\end{align*}
where $(a)$ follows by \cref{lem:potential_of_header_dominates} and $(b)$ since $\alpha (d-1) \geq \log 20$.
Therefore, combining \cref{eq:sorting_step_next_element} and \cref{eq:sorting_step_block}, we have that
\begin{align*}
\Ex{ \sum_{\ell = i}^{j+1} \phi_{t+1}(\ell) \, \middle| \, \mathcal{F}_t, \mathcal{G}_{t+1} }
& \leq 
  \sum_{\ell = i}^{j+1}  \phi_t(\ell)\cdot\left( 1 - \frac{\alpha}{4 (n-1)} \right),
\end{align*}
which also implies that
\begin{align*}
\Ex{\sum_{\ell = i}^{j+1} (\phi_{t+1}(\ell) - 1) \, \middle| \, \mathcal{F}_t, \mathcal{G}_{t+1} }
 & \leq  \sum_{\ell = i}^{j+1}  (\phi_t(\ell)  - 1) \cdot \left( 1 - \frac{\alpha}{4 (n-1)} \right).
\end{align*}
Further, by a symmetric argument, for a left block $l_t[i..j]$ we have that
\begin{align*}
\Ex{\sum_{\ell = i - 1}^{j} (\phi_{t+1}(\ell) - 1) \, \middle| \, \mathcal{F}_t, \mathcal{G}_{t+1} } 
 &  \leq \sum_{\ell = i - 1}^{j}  (\phi_t(\ell)  - 1) \cdot \left( 1 - \frac{\alpha}{4 (n-1)} \right).
\end{align*}
By aggregating over all blocks and using that swaps between stationary elements  do not contribute to the potential, we have that
\[
\Ex{\Phi_{t+1} \, \middle| \, \mathcal{F}_t, \mathcal{G}_{t+1} }
  \leq \Phi_t \cdot \left( 1 - \frac{\alpha}{4 (n-1)} \right).
\]
Combining with the fact that deterministically $\Phi_{t+1} \leq \Phi_t$, we conclude that
\[
 \Ex{ \left. \Phi_{t+1} \, \right| \, \mathcal{F}_t} 
   \leq \Phi_t \cdot \left( 1 - \frac{\alpha}{4 (n-1)} \cdot \Ex{\mathbf{1}_{\mathcal{G}_{t+1}} \middle| \mathcal{F}_t} \right). \qedhere
\]
\end{proof}

\subsection{Analysis of the \texorpdfstring{$\Psi$}{Ψ} Potential} \label{sec:mdev_psi_potential}

We now show that for any element $i$, the displacement of its target in $n \cdot k$ steps is at most $O(\sqrt{k \cdot \log n})$. The high level idea is that the targets perform an unbiased random walk, while the $\adm$ and $\lopt$ operations do not increase the exponential potential. 

We define another exponential potential function \[
  \Psi_t = \Psi_t(\alpha) = \sum_{i = 1}^n e^{\alpha \delta_t(i)},
\]
for a smoothing parameter $\alpha \in [1/n, 1]$ and where $\delta_t(i) = |\sigma_t^{-1}(i) - i|$. 

\begin{lemma} \label{lem:mgf_psi_and_target_displacement_bounds}
Consider the auxiliary process $(l_t,\tau_t,\sigma_t)_{t\geq 0}$ starting with the identity permutation $\sigma_0 = \id_n$, with a perturbation distribution $\DD$ satisfying for $D \sim \mathcal{D}$ that $\Ex{D} = 0$ and $\Ex{e^{3\lambda |D|}} \leq c'$ for some constants $\lambda \in (0, 1/3]$ and $c' \geq 1$. Further, consider the potential $\Psi_t = \Psi_t(\alpha)$ for any real $\alpha \in [1/n, \lambda]$. Then, $(i)$ for any step $t \geq 0$, it holds that
\[
\Ex{\Psi_{t+1} \,\middle| \, \mathcal{F}_t}
 \leq \Psi_t \cdot e^{\frac{\alpha^2}{n} \cdot \frac{3c'}{\lambda^2}} + \alpha \cdot \frac{4c'}{\lambda},
\]
$(ii)$ for any integer $k \geq 1$, it holds that
\[
\Ex{\Psi_{nk}} \leq 3 \cdot \frac{n}{\alpha} \cdot e^{\frac{3c'}{\lambda^2} \cdot k \cdot \alpha^2},
\]
and $(iii)$  for any integer $k \in [\frac{3}{c'} \cdot \log n, n]$, it holds that
\[
  \Pro{\max_{i \in [n]} \delta_{nk}(i) \leq \sqrt{108 \cdot \frac{c'}{\lambda^2} \cdot k \cdot \log n} } \geq 1- n^{-6}.
\]
\end{lemma}
\begin{proof}
\emph{Statement (i).} Let $\Psi_t'$ be the potential after the $t$-th mixing step, but before applying the $\adm$ and $\lopt$ operations. We will bound the expected value of $\psi_{t+1}'(i)$ for an arbitrary element $i \in [n]$.
Let $ \mathcal{T}_{t+1}(i)$ be the event that element $i$ was selected in step $t+1$. Then,
\begin{align}
& \Ex{\psi_{t+1}'(i) \,\middle| \, \mathcal{F}_t} \notag \\
 & \quad  = \Ex{\psi_{t+1}'(i) \,\middle| \, \mathcal{F}_t, \mathcal{T}_{t+1}(i)} \cdot \frac{1}{n} + \Ex{\psi_{t+1}'(i) \,\middle| \, \mathcal{F}_t, \neg\mathcal{T}_{t+1}(i)} \cdot \Big( 1 - \frac{1}{n} \Big) \notag \\
 & \quad = \psi_t(i) + \Ex{\psi_{t+1}'(i) - \psi_t(i) \,\middle| \, \mathcal{F}_t, \mathcal{T}_{t+1}(i)} \cdot \frac{1}{n} + \Ex{\psi_{t+1}'(i) - \psi_t(i) \,\middle| \, \mathcal{F}_t, \neg \mathcal{T}_{t+1}(i)} \cdot \Big( 1 - \frac{1}{n}\Big). \label{eq:mgf_combine}
\end{align}
We will now bound the increase of the potential separately for the two cases depending on whether $\mathcal{T}_{t+1}(i)$ holds or not. 

\begin{figure}[t]
    \centering
\scalebox{0.9}{
\begin{tikzpicture}
\newcommand{\Side}{0.75}
\begin{scope}[shift=({0, 0})]
\foreach \x in {1,...,9} 
  \draw[draw=black] (\x * \Side,0) rectangle ++(\Side,\Side);

\node[anchor=north] at (4.5 * \Side, 0.82 * \Side) {$i$};
\draw[{Circle[scale=0.7]}-stealth,black,very thick] 
(4.5 * \Side, 0.1 * \Side) -- 
(4.5 * \Side, -0.625 * \Side) --
(7.5 * \Side, -0.625 * \Side) --
(7.5 * \Side, - 1.25 * \Side);

\path[thick, black!60!green,->] (4.5 * \Side, 1.75 * \Side) edge[bend left] (5.5 * \Side, 1.75 * \Side);
\path[thick, black!60!green,->] (4.5 * \Side, 1.75 * \Side) edge[bend left] (6.5 * \Side, 1.75 * \Side);
\path[thick, black!60!green,->] (4.5 * \Side, 1.75 * \Side) edge[bend left] (7.5 * \Side, 1.75 * \Side);
\path[thick, red,->] (4.5 * \Side, 1.75 * \Side) edge[bend left] (8.5 * \Side, 1.75 * \Side);
\path[thick, red,->] (4.5 * \Side, 1.75 * \Side) edge[bend left] (9.5 * \Side, 1.75 * \Side);

\path[thick, red,->] (4.5 * \Side, 1.75 * \Side) edge[bend right] (3.5 * \Side, 1.75 * \Side);
\path[thick, red,->] (4.5 * \Side, 1.75 * \Side) edge[bend right] (2.5 * \Side, 1.75 * \Side);
\path[thick, red,->] (4.5 * \Side, 1.75 * \Side) edge[bend right] (1.5 * \Side, 1.75 * \Side);

\path[thick,dashed,red,->] (4.5 * \Side, 1.75 * \Side) edge[bend right] (0.5 * \Side, 1.75 * \Side);
\path[thick,dashed,red,->] (4.5 * \Side, 1.75 * \Side) edge[bend right] (-0.5 * \Side, 1.75 * \Side);
\path[thick,dashed,red,->] (4.5 * \Side, 1.75 * \Side) edge[bend left] (10.5 * \Side, 1.75 * \Side);

\path[thick,->] (7.5 * \Side, 1.2 * \Side) edge[bend left] (8.5 * \Side, 1.2 * \Side);
\path[thick,->] (7.5 * \Side, 1.2 * \Side) edge[bend left] (9.5 * \Side, 1.2 * \Side);
\path[thick,->] (7.5 * \Side, 1.2 * \Side) edge[bend left] (10.5 * \Side, 1.2 * \Side);
\end{scope}

\end{tikzpicture}}

    \caption{Arrows showing the possible displacements of element $i$ by $D_{t+1}$. Red arrows represent an increase on~$\delta_t(i)$, while green represent a decrease. Dashed arrows correspond to reaching a destination that is out of boundaries and in all cases lead to an increase (and their lengths upper bound the true increase). Solid black arrows represent the process corresponding to the additive term $\Ex{e^{\alpha |D_{t+1}|}}$ used to upper bound the contributions for when $i$ surpasses the position $\sigma_t^{-1}(i)$.}
    \label{fig:psi_case_1}
\end{figure}
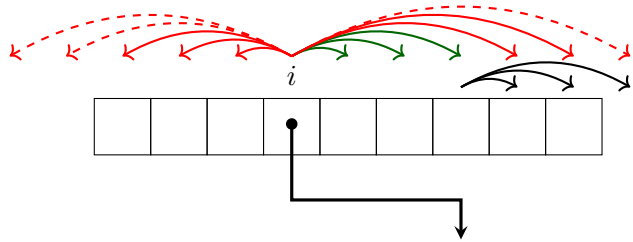

\subparagraph{Case 1 [$\mathcal{T}_{t+1}(i)$ holds]:} When selecting element $i$, $\delta_t(i)$ will change by exactly $D_{t+1}$, unless $(i)$ the element reaches one of the two boundaries of the array, or $(ii)$ the element surpasses the position $\sigma_t^{-1}(i)$. 

For case $(i)$, we upper bound the change by pessimistically assuming that the element carries on even outside the boundaries of the array (as this would mean larger $\delta_t(i)$). For case $(ii)$, we upper bound the increase by (always) having an additive term $\Ex{e^{\alpha |D_{t+1}| } - 1}$, i.e., at most the contribution of an element starting with $\delta_t(i) = 0$ (\cref{fig:psi_case_1}).

Therefore, 
\begin{align*}
\Ex{\psi_{t+1}'(i) - \psi_t(i) \, \middle| \, \mathcal{F}_t, \mathcal{T}_{t+1}(i) }
 &  \leq \psi_t(i) \cdot \left( \Ex{e^{\alpha D_{t+1}}} - 1 \right) + \Ex{e^{\alpha |D_{t+1}|} - 1} \\
 & \stackrel{(a)}{\leq} \psi_t(i) \cdot \left(\alpha \cdot \Ex{D_{t+1}} + \alpha^2 \cdot \frac{c'}{\lambda^2} \right) + \alpha \cdot \Ex{|D_{t+1}|} + \alpha^2 \cdot \frac{c'}{\lambda^2} \\
 &   \stackrel{(b)}{\leq} \psi_t(i) \cdot \alpha^2 \cdot \frac{c'}{\lambda^2} + \alpha \cdot \frac{2c'}{\lambda},
\end{align*}
using in $(a)$ \cref{lem:mgf_taylor} (since $\alpha \leq \lambda$) and in $(b)$ that $\Ex{D_{t+1}} = 0$, $\Ex{|D_{t+1}|} \leq \frac{1}{\lambda} \cdot \Ex{e^{\lambda |D_{t+1}|}} \leq \frac{c'}{\lambda}$ and $\alpha^2 \cdot \frac{c'}{\lambda^2} \leq \alpha \cdot \frac{c'}{\lambda}$ (since $\alpha \leq \lambda$).

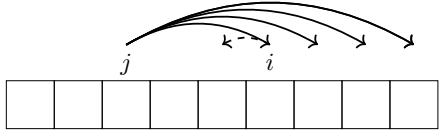
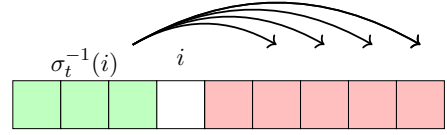
\begin{figure}
    \centering
\hfill
\begin{subfigure}[t]{0.45\textwidth}
\centering
\scalebox{0.85}{
\begin{tikzpicture}
\newcommand{\Side}{0.75}
\begin{scope}[shift=({0, 0})]
\foreach \x in {1,...,9} 
  \draw[draw=black] (\x * \Side,0) rectangle ++(\Side,\Side);

\node[anchor=north] at (3.5 * \Side, 0.8 * \Side) {$j$};
\node[anchor=north] at (6.5 * \Side, 0.8 * \Side) {$i$};

\path[thick,->] (3.5 * \Side, 1.75 * \Side) edge[bend left] (6.5 * \Side, 1.75 * \Side);
\path[thick,->] (3.5 * \Side, 1.75 * \Side) edge[bend left] (7.5 * \Side, 1.75 * \Side);
\path[thick,->] (3.5 * \Side, 1.75 * \Side) edge[bend left] (8.5 * \Side, 1.75 * \Side);
\path[thick,->] (3.5 * \Side, 1.75 * \Side) edge[bend left] (9.5 * \Side, 1.75 * \Side);
\path[thick,->] (3.5 * \Side, 1.75 * \Side) edge[bend left] (9.5 * \Side, 1.75 * \Side);

\path[thick,dashed,->] (6.5 * \Side, 1.75 * \Side) edge[bend right] (5.5 * \Side, 1.75 * \Side);
\end{scope}

\end{tikzpicture}}
\caption{Element $j$ can displace element $i$ by one place (shown with the dashed arrow) iff $|D_{t+1}| \geq |\sigma_t^{-1}(i) - \sigma_t^{-1}(j)|$ (corresponding to the solid arrows).}
\label{fig:psi_case_2_j_moves_i}
\end{subfigure}
\hfill
\begin{subfigure}[t]{0.45\textwidth}
\centering
\scalebox{0.85}{
\begin{tikzpicture}
\newcommand{\Side}{0.75}
\begin{scope}[shift=({0, 0})]
\foreach \x in {6,...,9} 
  \draw[draw=black,fill=green!25!white] (\x * \Side,0) rectangle ++(\Side,\Side);

\draw[draw=black] (5 * \Side,0) rectangle ++(\Side,\Side);

\foreach \x in {-1,...,4} 
  \draw[draw=black,fill=red!25!white] (\x * \Side,0) rectangle ++(\Side,\Side);

\node[anchor=north] at (5.5 * \Side, 0.8 * \Side) {$i$};

\draw[{Circle[scale=0.7]}-stealth,black,very thick] 
(5.5 * \Side, 0.1 * \Side) -- 
(5.5 * \Side, -0.625 * \Side) --
(8.5 * \Side, -0.625 * \Side) --
(8.5 * \Side, - 1.25 * \Side);

\draw [decorate,decoration={brace,amplitude=5pt,mirror,raise=-4ex}]  (8.5 * \Side,1.95 * \Side) -- (5.5 * \Side,1.95 * \Side) node[midway,yshift=-0.5em]{$\delta_t(i)$};

\begin{scope}[opacity=0.0]
\path[thick,->] (3.5 * \Side, 1.75 * \Side) edge[bend left] (6.5 * \Side, 1.75 * \Side);
\path[thick,->] (3.5 * \Side, 1.75 * \Side) edge[bend left] (7.5 * \Side, 1.75 * \Side);
\path[thick,->] (3.5 * \Side, 1.75 * \Side) edge[bend left] (8.5 * \Side, 1.75 * \Side);
\path[thick,->] (3.5 * \Side, 1.75 * \Side) edge[bend left] (9.5 * \Side, 1.75 * \Side);
\path[thick,->] (3.5 * \Side, 1.75 * \Side) edge[bend left] (9.5 * \Side, 1.75 * \Side);
\end{scope}

\end{scope}

\end{tikzpicture}}
\caption{The elements $j$ that could decrease $\delta_t(i)$ if selected shown in green and the others in red. Note that there are always at least $\delta_t(i)$ green elements. }
\label{fig:psi_case_2_many_green}
\end{subfigure}\hfill
\caption{Visualizations for Case 2, where element $j \neq i$ was selected in step $t+1$.}
    \label{fig:psi_case_2}
\end{figure}

\subparagraph{Case 2 [$\mathcal{T}_{t+1}(i)$ does not hold]:} Let $j$ be the element that was selected in step $t+1$. Element $i$ will be displaced by one place iff $|D_{t+1}| \geq |\sigma_t^{-1}(i) - \sigma_t^{-1}(j)|$ and in the direction of~$i$ (\cref{fig:psi_case_2_j_moves_i}). We can upper bound this probability using Markov's inequality as follows 
\begin{align} 
\Pro{D_{t+1} \geq |\sigma_t^{-1}(i) - \sigma_t^{-1}(j)| } 
  & = \Pro{e^{\lambda D_{t+1}} \geq e^{\lambda |\sigma_t^{-1}(i) - \sigma_t^{-1}(j)|}} \notag \\
  & \leq \Ex{e^{\lambda D_{t+1}}} \cdot e^{-\lambda |\sigma_t^{-1}(i) - \sigma_t^{-1}(j)|} \notag \\
  & \leq c' \cdot e^{-\lambda |\sigma_t^{-1}(i) - \sigma_t^{-1}(j)|}, \label{eq:displacement_distribution_markov}
\end{align}
and similarly\begin{align*}
\Pro{D_{t+1} \leq -|\sigma_t^{-1}(i) - \sigma_t^{-1}(j)| }
  & = \Pro{e^{\lambda D_{t+1}} \leq e^{-\lambda |\sigma_t^{-1}(i) - \sigma_t^{-1}(j)|}} \\
  & \leq \Ex{e^{-\lambda D_{t+1}}} \cdot e^{-\lambda |\sigma_t^{-1}(i) - \sigma_t^{-1}(j)|}  \\
  & \leq c' \cdot e^{-\lambda |\sigma_t^{-1}(i) - \sigma_t^{-1}(j)|}.
\end{align*}
Note that the number of elements $j$ that increase and the number of elements that decrease $\delta_t(i)$ need not be the same. We pessimistically assume that there are more elements on the side that increase $\delta_t(i)$, and without loss of generality we assume these to be on the left side.
Further, on the side that decrease $\delta_t(i)$ there should be at least $\delta_t(i)$ elements, namely the elements between positions $i$ and~$\sigma_t^{-1}(i)$ (\cref{fig:psi_case_2_many_green}). Hence, 
\clearpage
\begin{align*}
& \Ex{\psi_{t+1}'(i) - \psi_t(i) \, \middle| \, \mathcal{F}_t, \neg \mathcal{T}_{t+1}(i) } \\
 & \quad \leq \psi_t(i) \cdot \Bigg( \sum_{\ell \geq 1} \frac{1}{2(n{-}1)} \cdot \Pro{D_{t+1} \geq \ell} \cdot (e^{\alpha} - 1)  + \sum_{\ell =  1}^{\delta_t(i)} \frac{1}{2(n{-}1)} \cdot \Pro{D_{t+1} \leq -\ell} \cdot (e^{-\alpha} - 1) \Bigg) \\
 & \quad \stackrel{(a)}{=} \psi_t(i) \cdot \Bigg( \frac{1}{4(n{-}1)} \Ex{|D_{t+1}|} \cdot (e^{\alpha} + e^{-\alpha} - 2)  - \sum_{\ell \geq \delta_t(i) + 1} \frac{1}{2(n{-}1)} \cdot \Pro{D_{t+1} \leq -\ell} \cdot (e^{-\alpha} - 1)\Bigg) \\
 & \quad \stackrel{(b)}{\leq} \psi_t(i) \cdot \Bigg( \frac{1}{2(n{-}1)}\Ex{|D_{t+1}|} \cdot \alpha^2 + \sum_{\ell \geq \delta_t(i) + 1} \frac{1}{2(n{-}1)} \cdot \Pro{D_{t+1} \leq -\ell} \cdot \alpha \Bigg) \\
 & \quad \stackrel{(c)}{\leq} \psi_t(i) \cdot \Bigg( \frac{1}{n-1} \cdot \frac{2c'}{\lambda} \cdot \alpha^2 + e^{-\lambda \delta_t(i)}\cdot \sum_{\ell \geq 1} \frac{1}{2(n{-}1)} \cdot c' \cdot e^{-\lambda \ell} \cdot \alpha \Bigg) \\
 & \quad \stackrel{(d)}{\leq} \psi_t(i) \cdot \frac{\alpha^2}{n{-}1} \cdot \frac{2c'}{\lambda} + \frac{\alpha}{n{-}1} \cdot \frac{c'}{\lambda},
\end{align*}
using in $(a)$ that $\frac{1}{2} \Ex{|D_{t+1}|} = \sum_{\ell \geq 1} \Pro{D_{t+1} \geq \ell} = \sum_{\ell \geq 1} \Pro{D_{t+1} \leq -\ell}$ since $\Ex{D_{t+1}} = 0$, in $(b)$ that $e^{\alpha} \leq 1 + \alpha + \alpha^2$ (since $\alpha \leq 1$) and $1 - \alpha \leq e^{-\alpha}$, in $(c)$ the inequality~\cref{eq:displacement_distribution_markov} and that $\Ex{|D_{t+1}|} \leq \frac{1}{\lambda} \cdot \Ex{e^{\lambda |D_{t+1}|}} \leq \frac{c'}{\lambda}$, %
and in $(d)$ that $\sum_{\ell \geq 1} e^{-\lambda \ell} \leq \frac{1}{1 - e^{-\lambda}} \leq \frac{2}{\lambda}$ (since $\lambda \leq 1$) and $\psi_t(i) = e^{\alpha \delta_t(i)}$ (and $\alpha \leq \lambda$). %

Thus, combining the above two cases using \cref{eq:mgf_combine}, we have that
\begin{align*}
\Ex{\psi_{t+1}'(i) \,\middle| \, \mathcal{F}_t}
 & \leq \psi_t(i) \cdot \left(1 + \frac{\alpha^2}{n} \cdot \frac{3c'}{\lambda^2} \right) + \frac{\alpha}{n} \cdot \frac{4c'}{\lambda},
\end{align*}
and by aggregating
\begin{align} 
\Ex{\Psi_{t+1}' \,\middle| \, \mathcal{F}_t}
 & \leq \Psi_t \cdot \left(1 + \frac{\alpha^2}{n} \cdot \frac{3c'}{\lambda^2} \right) + \alpha \cdot \frac{4c'}{\lambda} \leq \Psi_t \cdot e^{\frac{\alpha^2}{n} \cdot \frac{3c'}{\lambda^2}} + \alpha \cdot \frac{4c'}{\lambda}. \label{eq:mgf_drop_psi}
\end{align}
After the swaps due to mixing, the $\adm$ operation may cause further swaps. By \cref{lem:admissible_swaps_are_good}, these swaps do not increase the sum of the two displacements of the swapped elements, nor the largest displacement of the two. Thus, by the convexity of $\Psi_t$ (being the sum of convex functions), this implies that the $\adm$ operation does not increase the potential. Further, the $\lopt$ operation does not affect $\sigma_t$, and so we conclude that the drop inequality \cref{eq:mgf_drop_psi} holds over the entire mixing step, i.e.,
\begin{align*}
\Ex{\Psi_{t+1} \,\middle| \, \mathcal{F}_t}
 & \leq \Psi_t \cdot e^{\frac{\alpha^2}{n} \cdot \frac{3c'}{\lambda^2}} + \alpha \cdot \frac{4c'}{\lambda}.
\end{align*}

\paragraph{\it Statement (ii).} Applying for $nk$ steps, we get that 
\begin{align*}
\Exp\left[\Psi_{nk} \right]
 &  \leq n \cdot e^{nk \cdot \frac{\alpha^2}{n} \cdot \frac{3c'}{\lambda^2}} + \alpha \cdot \frac{4c'}{\lambda} \cdot \left(1 + e^{\frac{\alpha^2}{n} \cdot \frac{3c'}{\lambda^2}} %
 + \ldots + e^{(nk-1) \cdot \frac{\alpha^2}{n} \cdot \frac{3c'}{\lambda^2}} \right) \\
 & \leq n \cdot e^{k \cdot \alpha^2 \cdot \frac{3c'}{\lambda^2}} + \alpha \cdot \frac{4c'}{\lambda} \cdot \frac{e^{k \cdot \alpha^2 \cdot \frac{3c'}{\lambda^2}} - 1}{e^{\frac{\alpha^2}{n} \cdot \frac{3c'}{\lambda^2}} - 1} \\
 & \stackrel{(a)}{\leq} n \cdot e^{k \cdot \alpha^2 \cdot \frac{3c'}{\lambda^2}} + \alpha \cdot \frac{4c'}{\lambda} \cdot \frac{e^{k \cdot \alpha^2 \cdot \frac{3c'}{\lambda^2}} - 1}{\frac{\alpha^2}{n} \cdot \frac{3c'}{\lambda^2}} \\
 & \stackrel{(b)}{\leq} 3 \frac{n}{\alpha} \cdot e^{k \cdot \alpha^2 \cdot \frac{3c'}{\lambda^2}} \\
 & \stackrel{(c)}{\leq} n^3 \cdot e^{k \cdot \alpha^2 \cdot \frac{3c'}{\lambda^2}},
\end{align*}
using in $(a)$ that $e^x \geq 1 + x$, in $(b)$ that $\lambda \leq 1$ and in $(c)$ that $\alpha \geq 1/n$.

\paragraph{\it Statement (iii).} By Markov's inequality, we have that
\[
\Pr\left[ \Psi_{nk} \leq n^{9} \cdot e^{k \cdot \alpha^2 \cdot \frac{3c'}{\lambda^2}} \right] \geq 1 - n^{-6}.
\]
Therefore, with probability at least $1-n^{-6}$, for any $i \in [n]$, 
\begin{align} \label{eq:mgf_delta_bound}
\delta_{nk}(i) 
 \leq \frac{1}{\alpha} \cdot \log \Psi_{nk} 
 \leq \frac{1}{\alpha} \cdot \left(9 \log n + k \cdot \alpha^2 \cdot \frac{3c'}{\lambda^2}\right) 
  = 9 \cdot \frac{\log n}{\alpha} + k \cdot \frac{3c'}{\lambda^2} \cdot \alpha.
\end{align}
By choosing $\alpha$ such that $9 \cdot \frac{\log n}{\alpha} = k \cdot \frac{3c'}{\lambda^2} \cdot \alpha$ and using that $k \in [\frac{3}{c'} \cdot \log n, n]$, we get that
\[
\alpha = \sqrt{\frac{3\lambda^2}{c'} \cdot \frac{\log n}{k} } \in [1/n, \lambda].
\]
Substituting this in \cref{eq:mgf_delta_bound}, we conclude that
\[
  \delta_{n k}(i) \leq 2 \cdot 9 \cdot \frac{\log n}{\alpha} \leq \sqrt{108 \cdot \frac{c'}{\lambda^2} \cdot k \cdot \log n}. \qedhere
\]
\end{proof}

\subsection{Proof of Bound on Maximum Deviation}
\label{sec:mdev_main_theorem}

Now, we are ready to prove the upper bound for the maximum deviation.
\begin{theorem}\label{thm:mdev_whp} 
Consider a process with bounded average rate of mixing steps for any $b \geq 1$ with a perturbation distribution $\DD$ satisfying for $D \sim \mathcal{D}$ that $\Ex{D} = 0$ and $\Ex{e^{3\lambda |D|}} \leq c'$ for some constants $\lambda \in (0, 1/3]$ and $c' \geq 1$. Then, there exists a constant $c := c(c', \lambda) > 0$ such that for any $m \geq 128 \cdot (b+1) \cdot n^2$, it holds that\[
  \Pro{\mdev(\pi_m) \leq c \cdot b \cdot \log n} \geq 1 - n^{-2}.
\]
\end{theorem}

\begin{proof}
Let $p = \frac{1}{4 \cdot (b+1)}$. We will analyze the process for the last $m = \frac{16d}{p} \cdot n^2$ steps, where $d > 1$ (to be specified below). We will split these steps into $m/(nk)$ groups of $nk$ steps each, where $k = \frac{16^2 \cdot 108 \cdot d^2}{p^2} \cdot \frac{c'}{\lambda^2} \cdot \log n$. At the end of each group of steps we will reset the targets to their original positions, so $\sigma_t = \tau_t = \id_n$. We will consider a \emph{capped process} in which if any of the mixing steps increases the displacement of the target within the current group to more than $M = \frac{1}{\lambda} \sqrt{108 \cdot c' \cdot k \cdot \log n}$, then the current and any future mixing step is not performed. By \cref{lem:mgf_psi_and_target_displacement_bounds}~$(iii)$, we have that
\begin{align}\label{eq:mdev_capped_process_agrees} 
  \Pr\Big[ \bigcap_{t \in [0, m]} \Big\{\max_{i \in [n]} \delta_t(i) \leq M \Big\} \big]
   & \geq 1- n^{-6} \cdot m  \geq 1 - \frac{1}{2n^3},
\end{align}
and so the original process agrees with the capped process at all steps 
with probability at least $1 - \frac{1}{2} n^{-3}$. So, from now on we will be working with the capped process.

Consider the potential $\Phi_t = \Phi_t(\alpha)$ with $\alpha = \frac{\log 20}{d - 1}$. Let $\mathcal{G}_{t+1}$ be the event that step $t+1$ was a sorting step. By \cref{lem:new_sorting_step_drift}, since $\alpha \in [\frac{\log 20}{d - 1}, 1]$, in any step $t$, it holds that
\[
  \Ex{\Phi_{t+1} \, \middle| \, \mathcal{F}_t} \leq \Phi_t \cdot \left( 1 - \frac{\alpha}{4 \cdot (n-1)} \cdot \Ex{\mathbf{1}_{\mathcal{G}_{t+1}} \middle| \mathcal{F}_t } \right).
\]
We denote by $\Phi_{t + nk}'$ the value of the potential before resetting the targets. By aggregating over $nk$ steps, for any step $t$ being the start of a group, we have that
\begin{align}
\Ex{ \Phi_{t+nk}' \, \middle| \, \mathcal{F}_t } 
   & \leq \Phi_t \cdot \prod_{i = 1}^{nk} \left( 1 - \frac{\alpha}{4 (n-1)} \cdot \Ex{\mathbf{1}_{\mathcal{G}_{t+i}} \middle| \mathcal{F}_t} \right) \notag \\ 
   & \leq \Phi_t \cdot e^{-\frac{\alpha}{4(n-1)} \cdot \sum_{i = 1}^{nk} \Ex{\mathbf{1}_{\mathcal{G}_{t+i}} \middle| \mathcal{F}_t}} \notag \\
   & \leq \Phi_t \cdot e^{-\frac{p\alpha}{4} \cdot \frac{nk}{n-1}}, \label{eq:phi_bound_with_reallocation}
\end{align}
where in the last inequality we used that $nk \geq (b+1) \cdot n$ and so by \cref{lem:many_sorting} there is at least a $p$ fraction of sorting steps among the $nk$ steps.

Now, recall that in the capped process, at the end of the~$nk$ steps, the displacement of any single target changes by at most $M$, so each term in the potential increases by at most a factor of $e^{\alpha \cdot d \cdot M}$ plus an additive term of $e^{\alpha \cdot d \cdot M}$. Therefore, %
\begin{align*}
\Ex{ \Phi_{t+nk} \, \middle| \, \mathcal{F}_t } 
  & \leq \Phi_t \cdot e^{-\frac{p\alpha}{4} \cdot \frac{nk}{n-1} + \alpha \cdot d \cdot M } + n \cdot e^{\alpha \cdot d \cdot M} 
    \leq \Phi_t \cdot e^{-\frac{p\alpha}{8} \cdot k} + n \cdot e^{\alpha \cdot d \cdot M},
\end{align*}
using that $\frac{p\alpha}{8} \cdot k \geq 2\alpha \cdot d \cdot M$ (since $k = \frac{16^2 \cdot 108 \cdot d^2}{p^2} \cdot \frac{c'}{\lambda} \cdot \log n$).

Trivially, at any step, each element is at most $n \cdot d$ units away from their target, so we have that $\Phi_0 \leq n \cdot e^{\alpha \cdot nd}$. By taking the expectation over the $\frac{m}{nk}$ groups we have that 
\begin{align*}
\Ex{ \Phi_{m} } 
   & \stackrel{(a)}{\leq} n \cdot e^{\alpha \cdot nd} \cdot e^{-2\alpha \cdot nd} + n \cdot e^{\alpha \cdot d \cdot M} \cdot \big( 1 + e^{-\frac{p\alpha}{8} \cdot k} + e^{-2 \cdot \frac{p\alpha}{8} \cdot k} + \ldots \big)\\
   & \leq 1 + n \cdot e^{\alpha \cdot d \cdot M} \cdot \frac{1}{1 - e^{-\frac{p\alpha}{8} \cdot k}} \\
   & \stackrel{(b)}{\leq} 2n \cdot e^{\alpha \cdot d \cdot M},
\end{align*}
using in $(a)$ that $\frac{p\alpha}{8} \cdot k \cdot \frac{m}{nk} \geq 2\alpha \cdot nd$ and in $(b)$ that $e^{-\frac{p\alpha}{8} \cdot k} \leq \frac{1}{2}$ (since $k = \omega(1/p)$).

By Markov's inequality we have that \[
  \Pro{\Phi_m \leq 2n^4 \cdot e^{\alpha \cdot d \cdot M}} \geq 1 - n^{-3},
\]
and so we can deduce that
\[
  \Pro{\mdsp(l_m) \leq \frac{1}{\alpha} (\log 2 + 4 \log n) + d \cdot M} \geq 1 - n^{-3}.
\]
Using \cref{lem:dsp_bounds_dev}, we have that for $c = 32 \cdot \frac{16 \cdot 108 \cdot d \cdot c'}{\lambda^2}$
(since $\frac{5}{\alpha} \log n = \frac{20}{\log 20} \log n \leq d \cdot M$) the capped process satisfies,
\[
  \Pro{\mdev(\pi_m) \leq  c \cdot b \cdot \log n} \geq 1 - n^{-3}. 
\]
Finally, by taking the union-bound with \cref{eq:mdev_capped_process_agrees}, the capped process agrees with the original process and so  we conclude that the same bound also holds for the maximum deviation of the original process
\[
  \Pro{\mdev(\pi_m) \leq  c \cdot b \cdot \log n} \geq 1 - \frac{1}{n^3} - \frac{1}{2n^{3}} \geq 1 - n^{-2}.
\]
By setting $d = 2$, we get the conclusion.
\end{proof}
From the preceding result bounding the maximum deviation, we immediately obtain an $O(n \cdot b \cdot \log n)$ bound on the total deviation. In \cref{thm:total_deviation_whp} we will improve this to $O(n)$ for any constant $b \geq 1$.

\begin{remark}
Consider the process where every sorting step is followed by $b \in [1, \frac{n}{\log^2 n}]$ mixing steps. Then by \cref{lem:max_and_total_deviation_lb}~$(i)$ the  $O(b \cdot \log n)$ bound on the maximum deviation implied by~\cref{thm:mdev_whp} is asymptotically tight.
\end{remark}

\begin{remark}
The proof in \cref{thm:mdev_whp} also applies to a model where every $nk$ steps each rank can be perturbed adversarially by at most $M$ places.
\end{remark}

In the proof of \cref{thm:mdev_whp}, we assumed the worst possible initial maximum deviation, i.e., that of~$nd$ from their target positions. This allowed us to prove that the convergence time to a configuration with $O(b \cdot \log n)$ maximum deviation is \Whp~$O(n^2 \cdot b)$. We can improve this to $O(n \cdot b \cdot \mdev(\pi_0))$ for sufficiently large initial maximum deviation $\mdev(\pi_0)$. By \cref{lem:convergence_time_lb}, this is tight for any constant $b \geq 1$.

\begin{corollary} \label{cor:max_deviation_convergence_time} 
Consider a process with bounded average rate of mixing steps for any $b \geq 1$ with a perturbation distribution $\DD$ satisfying for $D \sim \mathcal{D}$ that $\Ex{D} = 0$ and $\Ex{e^{3\lambda |D|}} \leq c'$ for some constants $\lambda \in (0, 1/3]$ and $c' \geq 1$. Then, there exists a constant $c > 0$ such that for any $m \geq 320 \cdot (b+1) \cdot n \cdot \mdev(\pi_0)$, we have that\[
  \Pro{\mdev(\pi_m) \leq c \cdot b \cdot \log n} \geq 1 - n^{-2}.
\]
\end{corollary}

\begin{remark}
Consider the process where every sorting step is followed by $b$ mixing steps for any constant~$b \geq 1$. Then, by \cref{lem:convergence_time_lb} the  $O(b \cdot n \cdot \mdev(\pi_m))$ bound on the convergence time implied by~\cref{thm:mdev_whp} is asymptotically tight.
\end{remark}

\section{Tight Bound on the Total Deviation} \label{sec:total_deviation}
In this section, we will extend the ideas used in proving \cref{thm:mdev_whp} to show that \Whp~the total deviation is $O(n)$ for any constant $b \geq 1$. This matches the lower bound in \cite[Theorem 1]{AnagnostopoulosKMU11} for any constant $b \geq 1$.

We cannot directly apply the idea of resetting the targets for intervals of length $\ell = o(n \log n)$, because in this regime the worst-case displacement due to mixing is $\omega(\ell/n)$ and no longer $O(\sqrt{(\ell / n) \log n})$. For instance, when $\ell = \Theta(n)$, the maximum displacement due to mixing is \Whp~$\Theta(\log n/\log \log n)$. Therefore, the drop of the potential function by a factor of $e^{-\Theta(\ell/n)}$ can no longer counteract the worst-case displacement of $\omega(\ell/n)$, which would increase the potential by a factor of $e^{\omega(\ell/n)}$. So, using this technique we cannot show that overall the potential decreases in expectation over these $\ell$ steps.

To overcome this, we will take advantage of the fact that there are not too many elements with a large displacement. For elements with a small $O(\ell)$ displacement, we will reset their targets and handle their contribution using the argument in the proof of~\cref{thm:mdev_whp}. For the rest of the elements that have a large displacement, we prove that they don't contribute much to the total deviation. For instance, for $\ell = \Theta(n (\log n)^{2/3})$ we prove that the number of elements~$i$ with $\delta_t(i) = \Omega((\log n)^{2/3})$ is $O(n \cdot e^{-(\log n)^{1/3}})$ and so they contribute $o(n)$ to the total displacement. Applying this idea $\Theta(\log \log \log n)$ times for disjoint intervals of decreasing length and aggregating all the contributions of the large displacements, we are able to prove the $O(n)$ bound on total deviation.

In \cref{sec:tdev_psi_concentration}, we prove that the $\Psi$ potential is concentrated and in \cref{sec:tdev_large_elements_bound} we use this to prove the bound on the number (and total displacement) of elements with a large displacement. Finally, in \cref{sec:tdev_main_theorem}, we combine these to prove the $O(n)$ bound on the total deviation.

\subsection{Concentration of the \texorpdfstring{$\Psi$}{Ψ} Potential}
\label{sec:tdev_psi_concentration}

In \cref{sec:mdev_psi_potential}, we analyzed the expectation of the $\Psi$ potential. In the lemma below we will show that $\Psi$ is also concentrated. We will use this lemma in \cref{sec:tdev_large_elements_bound} to bound the sum of large displacements.

\begin{lemma}[Concentration of $\Psi$] \label{lem:psi_concentration}
Consider the auxiliary process $(l_t,\tau_t,\sigma_t)_{t\geq 0}$ starting with the identity permutation $\sigma_0 = \id_n$, with a perturbation distribution $\DD$ satisfying for $D \sim \mathcal{D}$ that $\Ex{D} = 0$ and $\Ex{e^{3\lambda |D|}} \leq c'$ for some constants $\lambda \in (0, 1/3]$ and $c' \geq 1$. Further, let $C \in [\frac{3}{\lambda}, n]$ be arbitrary. Then, for any $k \in [1, C \cdot \log n]$ and $\Psi_t = \Psi_t(\alpha)$ with any real $\alpha \in \big[\frac{1}{n},  \frac{\lambda}{6 \cdot 2 \cdot \sqrt{108 \cdot c' \cdot C}}\big]$, it holds that 
\[
\Pro{ \Psi_{nk} \leq 4 \cdot \frac{n}{\alpha} \cdot e^{k \cdot \alpha^2 \cdot \frac{3c'}{\lambda^2}}} 
  \geq 1 - n^{-4}.
\]
\end{lemma}
Our goal is to use a variant of Azuma's inequality (\cref{lem:azuma_modified}) to prove concentration for the $\Psi$ potential. To this end, we introduce a capped process, and an auxiliary potential, defined in a way that it is a super-martingale. Next, we show that this auxiliary potential satisfies a bounded difference inequality and hence, we can apply Azuma's inequality. 
\begin{proof}
We define the \emph{capped process} as the process following the original process, until in some step $t$ some element $i$ has displacement $\delta_t(i) > \tilde{c} \cdot \log n$ for $\tilde{c} = \frac{1}{\lambda} \sqrt{108 \cdot c' \cdot C}$. From that step $t$, in the capped process, we stop performing  mixing steps.  

Applying \cref{lem:mgf_psi_and_target_displacement_bounds}~$(iii)$ (with $k = C \cdot \log n$) for all steps of the interval $[1, nk]$, and taking the union-bound, we have that the two processes agree with high probability. Specifically,
\begin{align} \label{eq:whp_all_delta_i_small}
\Pr\Big[\bigcap_{t \in [1, nk]} \Big\{ \max_{i \in [n]}\delta_t(i) \leq \tilde{c} \cdot \log n \Big\} \Big]  
 & \geq 1 - n^{-6} \cdot n \cdot k \geq 1 - \frac{1}{2}n^{-4}.
\end{align}

\subparagraph{The potential functions.} Let $\widetilde{\Psi}_t = \widetilde{\Psi}_t(\alpha)$ be the potential function for the capped process, for any smoothing parameter $\alpha \in \big[\frac{1}{n},  \frac{1}{6\tilde{c}} \big]$. Further, let $\Psi_t'$ be the value of the potential before applying the $\adm$ and $\lopt$ operations at the end of the step. By \cref{lem:admissible_swaps_are_good}~$(i)$ and since $\widetilde{\Psi}$ is convex, these operations cannot increase the value of the potential, so we have that $\widetilde{\Psi}_t \leq \widetilde{\Psi}_t'$.

For any step $t$ of the capped and the original process, we have that $\widetilde{\Psi}_{t+1} \leq \widetilde{\Psi}_t$ and so we can deduce by \cref{lem:mgf_psi_and_target_displacement_bounds}~$(i)$, since $\alpha \in \big[\frac{1}{n}, \lambda\big]$ that for any step $t$ it holds that
\begin{align} \label{eq:tilde_psi_drift}
\Ex{\widetilde{\Psi}_{t+1}' \,\middle| \, \mathcal{F}_t}
  \leq \widetilde{\Psi}_t \cdot e^{\frac{\alpha^2}{n} \cdot \frac{3c'}{\lambda^2}} + \alpha \cdot \frac{4c'}{\lambda}.
\end{align}
Further, we define the auxiliary potential functions
\begin{align} \label{eq:hat_psi_potential}
  \widehat{\Psi}_t' = \left( \widetilde{\Psi}_t' + \frac{\alpha \cdot \frac{4c'}{\lambda}}{e^{\frac{\alpha^2}{n} \cdot \frac{3c'}{\lambda^2}} - 1} \right) \cdot e^{-t \cdot \frac{\alpha^2}{n} \cdot \frac{3c'}{\lambda^2}},
\end{align}
and
\begin{align*}
   \widehat{\Psi}_t = \left( \widetilde{\Psi}_t + \frac{\alpha \cdot \frac{4c'}{\lambda}}{e^{\frac{\alpha^2}{n} \cdot \frac{3c'}{\lambda^2}} - 1} \right) \cdot e^{-t \cdot \frac{\alpha^2}{n} \cdot \frac{3c'}{\lambda^2}},
\end{align*}
for which $\widehat{\Psi}_t', \widehat{\Psi}_t, \widehat{\Psi}_{t+1}', \widehat{\Psi}_{t+1}, \ldots $ is a sub-martingale since deterministically
\[
 \widehat{\Psi}_t \leq \widehat{\Psi}_t',
\]
and
\begin{align*}
\Ex{\widehat{\Psi}_{t+1}' \, \middle| \, \mathcal{F}_t }
  & = \Big( \Ex{\widetilde{\Psi}_{t+1}' \, \middle| \, \mathcal{F}_t } + \frac{\alpha \cdot \frac{4c'}{\lambda}}{e^{\frac{\alpha^2}{n} \cdot \frac{3c'}{\lambda^2}} - 1} \Big) \cdot e^{- (t+1) \cdot \frac{\alpha^2}{n} \cdot \frac{3c'}{\lambda^2}} \\
  & \stackrel{(\ref{eq:tilde_psi_drift})}{\leq}  \Big( \widetilde{\Psi}_t \cdot e^{\frac{\alpha^2}{n} \cdot \frac{3c'}{\lambda^2}} + \alpha \cdot \frac{4c'}{\lambda} + \frac{\alpha \cdot \frac{4c'}{\lambda}}{e^{\frac{\alpha^2}{n} \cdot \frac{3c'}{\lambda^2}} - 1}\Big)  \cdot e^{-(t+1) \cdot \frac{\alpha^2}{n} \cdot \frac{3c'}{\lambda^2}} \\
  & = \Big( \widetilde{\Psi}_t \cdot e^{\frac{\alpha^2}{n} \cdot \frac{3c'}{\lambda^2}} + \frac{\alpha \cdot \frac{4c'}{\lambda}}{e^{\frac{\alpha^2}{n} \cdot \frac{3c'}{\lambda^2}} - 1} \cdot e^{\frac{\alpha^2}{n} \cdot \frac{3c'}{\lambda^2}} \Big) \cdot e^{-(t+1) \cdot \frac{\alpha^2}{n} \cdot \frac{3c'}{\lambda^2} } \\
  & = \Big( \widetilde{\Psi}_t + \frac{\alpha \cdot \frac{4c'}{\lambda}}{e^{\frac{\alpha^2}{n} \cdot \frac{3c'}{\lambda^2}} - 1} \Big) \cdot e^{-t \cdot \frac{\alpha^2}{n} \cdot \frac{3c'}{\lambda^2}} = \widehat{\Psi}_t.
\end{align*}

\subparagraph{Bounded difference.} Since we are in the capped process, any mixing step can change $\widetilde{\Psi}$ in at most $2\tilde{c} \cdot \log n$ indices $j$. Therefore,
\begin{align}
  \left| \widetilde{\Psi}_{t+1}' - \widetilde{\Psi}_t \right| 
  & \leq (2 \tilde{c} \cdot \log n) \cdot \max_{j \in [n]} \widetilde{\psi}_t(j) 
    \leq (2 \tilde{c} \cdot \log n) \cdot e^{\alpha \cdot \tilde{c} \log n} 
    \leq 2 \tilde{c} \cdot n^{1/6} \cdot \log n, \label{eq:delta_tilde_psi_bounded}
\end{align}
using that $\alpha \leq \frac{1}{6\tilde{c}}$. Similarly it also holds that\begin{align}
\label{eq:bound_on_tilde_psi}
 \widetilde{\Psi}_t \leq \widetilde{\Psi}_t' \leq n \cdot e^{\alpha \cdot \tilde{c} \log n} \leq n \cdot n^{1/6}.
\end{align}
Now we can bound the change of the potential due to a mixing step for any step $t \leq nk$
as follows
\begin{align*}
 \Big| \widehat{\Psi}_{t+1}' - \widehat{\Psi}_t \Big| 
  & \stackrel{(a)}{\leq} \Bigg(\Big| \widetilde{\Psi}_{t+1}' - \widetilde{\Psi}_t \Big| + \widetilde{\Psi}_{t+1} \cdot \left(1 - e^{-\frac{\alpha^2}{n} \cdot \frac{3c'}{\lambda^2}} \right)  + \alpha \cdot \frac{4c'}{\lambda} \cdot \frac{1 - e^{-\frac{\alpha^2}{n} \cdot \frac{3c'}{\lambda^2}}}{e^{\frac{\alpha^2}{n} \cdot \frac{3c'}{\lambda^2}} - 1} \Bigg) \cdot e^{-t \cdot \frac{\alpha^2}{n} \cdot \frac{3c'}{\lambda^2}} \\
  & \stackrel{(b)}{\leq} 2 \tilde{c} \cdot n^{1/6} \cdot \log n + n \cdot n^{1/6} \cdot \frac{\alpha^2}{n} \cdot \frac{3c'}{\lambda^2} + \alpha \cdot \frac{4c'}{\lambda} \\
  & \leq n^{1/3},
\end{align*}%
using in $(a)$ the definition of the $\widehat{\Psi}$ potential in \cref{eq:hat_psi_potential} and in $(b)$ the bounds in \cref{eq:bound_on_tilde_psi,eq:delta_tilde_psi_bounded}. Further, recall that for any step $t \geq 0$, it holds $\widetilde{\Psi}_t \leq \widetilde{\Psi}_t'$.

\subparagraph{Applying Azuma's inequality.} Now we are ready to apply Azuma's inequality (\cref{lem:azuma_modified}) for the super-martingale $\widehat{\Psi}_0, \widehat{\Psi}_1', \widehat{\Psi}_1, \widehat{\Psi}_2', \widehat{\Psi}_2\ldots$ and parameters $\Lambda = \frac{n}{\alpha}$ and $\Delta = n^{1/3}$, %
\begin{align} \label{eq:azuma_application}
& \Pro{ \widehat{\Psi}_{nk} \leq \widehat{\Psi}_0 + \frac{n}{\alpha} } \geq 1 - \exp\left( - \frac{(n/\alpha)^2}{3nk \cdot (n^{1/3})^2  + 2 \cdot 4 \cdot \frac{n}{\alpha} \cdot n^{1/3} } \right),
\end{align}
since using the definition of $\widehat{\Psi}_0$ and the fact that we are starting from the identity permutation
\begin{align*}
\widehat{\Psi}_0 + \frac{n}{\alpha} 
 & = \left( n + \frac{\alpha \cdot \frac{4c'}{\lambda}}{e^{\frac{\alpha^2}{n} \cdot \frac{3c'}{\lambda^2}} - 1} \right) \cdot 1 + \frac{n}{\alpha}
   \leq n + \frac{4}{3} \cdot \lambda \cdot \frac{n}{\alpha} + \frac{n}{\alpha}
   \leq 4 \cdot \frac{n}{\alpha}.
\end{align*}
So when $\widehat{\Psi}_{nk} \leq \widehat{\Psi}_0 + \frac{n}{\alpha}$, we also have that\[
  \widetilde{\Psi}_{nk}
    = \widehat{\Psi}_{nk}  \cdot e^{n k \cdot \frac{\alpha^2}{n} \cdot \frac{3c'}{\lambda^2}} - \frac{\alpha \cdot \frac{4c'}{\lambda}}{e^{\frac{\alpha^2}{n} \cdot \frac{3c'}{\lambda^2}} - 1}
    \leq 4 \cdot \frac{n}{\alpha} \cdot e^{k \cdot \alpha^2 \cdot \frac{3c'}{\lambda^2}}.
\]
Therefore, returning to \cref{eq:azuma_application} since $k = O(\log n)$ and $\alpha = O(1)$, we deduce that
\begin{align} \label{eq:whp_bounded_diff_tilde_psi}
\Pro{ \widetilde{\Psi}_{nk} \leq 4 \cdot \frac{n}{\alpha} \cdot e^{k \cdot \alpha^2 \cdot \frac{3c'}{\lambda^2}} }
  \geq 1 - \frac{1}{2}n^{-4}.
\end{align}

Finally, by taking the union-bound of \cref{eq:whp_all_delta_i_small} and \cref{eq:whp_bounded_diff_tilde_psi}, we have that the capped process agrees with the original process, so $\Psi_{nk} = \widetilde{\Psi}_{nk}$ and thus
\begin{align*}
\Pro{ \Psi_{nk} \leq 4 \cdot \frac{n}{\alpha} \cdot e^{k \cdot \alpha^2 \cdot \frac{3c'}{\lambda^2}}} 
  & \geq 1 - \frac{1}{2} n^{-4} - \frac{1}{2} n^{-4} = 1 - n^{-4}. \qedhere
\end{align*}
\end{proof}

\subsection{Bounding the Sum of Large Displacements}
\label{sec:tdev_large_elements_bound}

In this section, we use the concentration of the $\Psi$ potential to bound the sum of large displacements. 

\begin{lemma} \label{lem:sum_large_delta_bound}
Consider the auxiliary process $(l_t,\tau_t,\sigma_t)_{t\geq 0}$ starting with the identity permutation $\sigma_0 = \id_n$, with a perturbation distribution $\DD$ satisfying for $D \sim \mathcal{D}$ that $\Ex{D} = 0$ and $\Ex{e^{3\lambda |D|}} \leq c'$ for some constants $\lambda \in (0, 1/3]$ and $c' \geq 1$.
Let $C \geq 1$ be an arbitrary constant and $R = \max\{ 6 \cdot 2 \cdot \frac{c'}{\lambda^2} \sqrt{108 \cdot C}, 32^3 \cdot 16 \}$. Then, for any $k \in [(64 R \log (16R))^3, C \cdot \log n]$, it holds that
\[
  \Pr\left[ \sum_{i \in [n] : \delta_{nk}(i) \geq k^{2/3}} \delta_{nk}(i) \leq n \cdot e^{-\frac{1}{2R} \cdot k^{1/3}} \right] \geq 1 - n^{-3}.
\]
\end{lemma}
\begin{proof}
Let $N_{\ell}$ be the number of elements $i$ with $\delta_{nk}(i) \geq \ell$, i.e.,
\[
N_{\ell} = \left| \left\{ i \in [n] : \delta_{nk}(i) \geq \ell\right\} \right|.
\]
Then, we can express the total displacement $D_{\geq k^{2/3}}$ of elements $i$ with $\delta_{nk}(i) \geq k^{2/3}$ as 
\begin{align} \label{eq:displacement_def}
  D_{\geq k^{2/3}}
    = N_{k^{2/3}} \cdot (k^{2/3} - 1) + \sum_{\ell = k^{2/3}}^{\infty} N_{\ell}.
\end{align}
We will obtain bounds for $N_{\ell}$ using the potential functions $\Psi_t^{\ell} = \Psi_t^{\ell}(\alpha_{\ell})$ for $\ell = k^{2/3}, \ldots, k$ with smoothing parameters
\[
  \alpha_\ell = \begin{cases}
      \frac{1}{R} \cdot \frac{\ell}{k}, & \text{for }\ell \leq k, \\
      \frac{1}{R}, & \text{for }\ell > k.
  \end{cases}
\]
Applying \cref{lem:psi_concentration} (as $k \leq C \cdot \log n$ and $\frac{1}{n} \leq \alpha_{\ell} \leq \frac{1}{R} \leq \frac{\lambda}{6 \cdot 2 \cdot \sqrt{108 \cdot c' \cdot C}}$), taking the union bound over the $(k - k^{2/3}+ 1)$ potential functions, we get that
\begin{align}
 & \Pro{\bigcap_{\ell \in [k^{2/3}, k]} \left\{ \Psi_{nk}^\ell \leq 4 \cdot \frac{n}{\alpha_\ell} \cdot e^{k \cdot \frac{3c'}{\lambda^2} \cdot \alpha_\ell^2} \right\} } \notag \\
 & \quad \geq 1 - n^{-4} \cdot (k - k^{2/3} + 1) \notag \\
 & \quad \geq 1 - n^{-3}, \label{eq:large_all_psi_small}
\end{align}
Now, assuming that these bounds hold on the $\Psi$'s we will bound $N_{\ell}$. We start by considering $\ell \in [k^{2/3}, k]$, 
\begin{align}\label{eq:whp_n_ell_first_bound}
  N_{\ell} 
    & \leq \Psi_{nk}^\ell \cdot e^{-\alpha_\ell \cdot \ell} \notag \\
    & \leq 4 \cdot \frac{n}{\alpha_\ell} \cdot e^{k \cdot \frac{3c'}{\lambda^2} \cdot \alpha_\ell^2} \cdot e^{-\alpha_\ell \cdot \ell} \notag \\
    & = 4Rn \cdot \frac{k}{\ell} \cdot e^{\frac{1}{R} \cdot \frac{\ell^2}{k} \cdot \left( \frac{3c'}{\lambda^2} \cdot \frac{1}{R} - 1\right) } \notag \\
    & \stackrel{(a)}{\leq} 4Rn \cdot \frac{k}{\ell} \cdot e^{-\frac{3}{4R} \cdot \frac{\ell^2}{k}  } \notag \\
    & \stackrel{(b)}{\leq} 4Rn \cdot k^{1/3} \cdot e^{-\frac{3}{4R} \cdot k^{1/3} } \notag \\
    & \stackrel{(c)}{\leq} \frac{1}{4} n \cdot e^{-\frac{1}{2R} \cdot k^{1/3}},
\end{align}
using in $(a)$ that $R \geq 12 \cdot \frac{c'}{\lambda^2}$, in $(b)$ that $\ell \geq k^{2/3}$ and in $(c)$ that 
\begin{align}
e^{\frac{1}{4R} \cdot k^{1/3}} 
  & = \left( e^{\frac{1}{16R} \cdot k^{1/3}} \right)^3 \cdot e^{\frac{1}{16R} \cdot k^{1/3}} \notag \\
  & \geq \left( \frac{1}{16R} \cdot k^{1/3} \right)^3 \cdot e^{4 \log (16R)} \notag \\
  & = \frac{1}{(16R)^3} \cdot k \cdot (16R)^4  \notag \\
  & = 16Rk, \label{eq:exp_k_lb}
\end{align}
since  $e^x \geq x$ and $k \geq (64 R \log (16R))^3$. 

Similarly, for $\ell > k$, we get
\begin{align} 
  N_{\ell} 
    & \leq \Psi_{n \cdot k}^\ell \cdot e^{-\alpha_\ell \cdot \ell} \notag \\
    & \leq 4Rn \cdot e^{k \cdot \frac{3c'}{\lambda^2} \cdot \frac{1}{R^2}} \cdot e^{-\frac{1}{R} \ell} \notag \\
    & \leq 4Rn \cdot e^{-\frac{3}{4R} \cdot \ell} \notag \\
    & \leq \frac{1}{4} n \cdot e^{-\frac{1}{2R} \cdot k^{1/3}},\label{eq:whp_n_ell_second_bound}
\end{align}
using \cref{eq:exp_k_lb} in the last inequality.

Finally, the first term in \cref{eq:displacement_def} can be bounded similarly to
\cref{eq:whp_n_ell_first_bound} using the steps until $(b)$, so
\begin{align} 
N_{k^{2/3}} \cdot (k^{2/3} - 1) 
    & \leq 4Rn \cdot k^{1/3} \cdot e^{-\frac{3}{4R} \cdot k^{1/3} } \cdot k^{2/3} \notag \\
    & \leq \frac{1}{4} n \cdot e^{-\frac{1}{2R} \cdot k^{1/3}}, \label{eq:whp_second_term_bound}
\end{align}
using \cref{eq:exp_k_lb} in the last step.

Combining \cref{eq:whp_n_ell_first_bound}, \cref{eq:whp_n_ell_second_bound} and \cref{eq:whp_second_term_bound}, we get
\begin{align*}
D_{\geq k^{2/3}} & \leq \frac{1}{4} n \cdot e^{-\frac{1}{2R} \cdot k^{1/3}} + 2 \cdot \frac{1}{4} n \cdot e^{-\frac{1}{2R} \cdot k^{1/3}} \\
& \leq n \cdot e^{-\frac{1}{2R} \cdot k^{1/3}}.
\end{align*}
Finally, by  \cref{eq:large_all_psi_small}, we conclude that
\[
  \Pro{D_{\geq k^{2/3}} \leq n \cdot e^{-\frac{1}{2R} \cdot k^{1/3}}} \geq 1 - n^{-3}. \qedhere
\]
\end{proof}

\subsection{Proof of Bound on Total Deviation} \label{sec:tdev_main_theorem}

Now, we are ready to prove the $O(n)$ bound on total deviation with high probability.

\begin{theorem}[Total Deviation] \label{thm:total_deviation_whp}
Consider Na\"ive Sort under random adjacent rank swaps for any constant $b \geq 1$. Then, there exists a constant $c = c(b) > 0$ such that for any $m \geq 16 \cdot 128 \cdot 4 \cdot (b+1) \cdot n^2$, it follows that\[
  \Pro{\dev(\pi_m) \leq cn} \geq 1 - n^{-2}.
\]
\end{theorem}
\begin{proof}
Let $p = \frac{1}{4 \cdot (b+1)}$. We will perform the analysis in $\kappa + 2$ phases $(0, \rho], (\rho, t_0], (t_0, t_1], \ldots, (t_{\kappa-1}, t_{\kappa}]$ each consisting of  $\rho, n \cdot k_0, n \cdot k_1, \ldots, n \cdot k_\kappa$ many steps, where $k_0 = k_1 = \frac{16^2 \cdot 108 \cdot d^2}{p^2} \cdot \frac{c'}{\lambda^2} \cdot \log n$, $k_{i+1} = \frac{16d}{p} \cdot k_i^{2/3}$ (for $d = 128$) and $\rho = m - n \cdot k_0 - n \cdot  k_1 - \ldots - n \cdot k_{\kappa}$  (the remaining steps). We choose $\kappa$ to be the smallest integer such that $k_\kappa \geq \max\{ (\frac{8 \cdot 16d}{p})^4, (64 R \log (16R))^3 \}$, where $R = \max\{ 6 \cdot 2 \cdot \frac{c'}{\lambda^2} \sqrt{108 \cdot C}, 32^3 \cdot 16 \}$ the constant in \cref{lem:sum_large_delta_bound}. It follows that $\kappa = \Theta(\log \log \log n)$.

We will be using two potential functions $\Phi_t = \Phi_t(\alpha)$ and $\widetilde{\Phi}_t = \widetilde{\Phi}_t(\tilde{\alpha})$ with $\alpha$ being a large constant factor larger than $\widetilde{\alpha}$, so it holds that $\Phi_t \geq \widetilde{\Phi}_t$ for any step $t \geq 0$.\footnote{A similar interplay between two exponential potentials has been used to prove concentration in the context of balanced allocations (e.g.,~\cite{LS22Queries}). However, here because admissible steps can cause a large decrease in the potentials we require a slight variant of Azuma's inequality (\cref{lem:azuma_modified}).} More concretely, we set $\alpha = \frac{\log 20}{127}$ and $\tilde{\alpha} = \frac{1}{42} \cdot \frac{\log 20}{127}$ (and $d$ defined as above). First, we are going to prove that \Whp~$\Phi_t = \mathrm{poly}(n)$ for any $t \in [t_0, t_{\kappa}]$ and then we use this to show that $|\widetilde{\Phi}_{t+1} - \widetilde{\Phi}_t| \leq n^{1/3}$, which allows us to prove that \Whp~$\widetilde{\Phi}_t = O(n)$ for any $t = t_1, t_2, \ldots , t_{\kappa}$ (before resetting the targets).

\definecolor{beige}{HTML}{fffade}
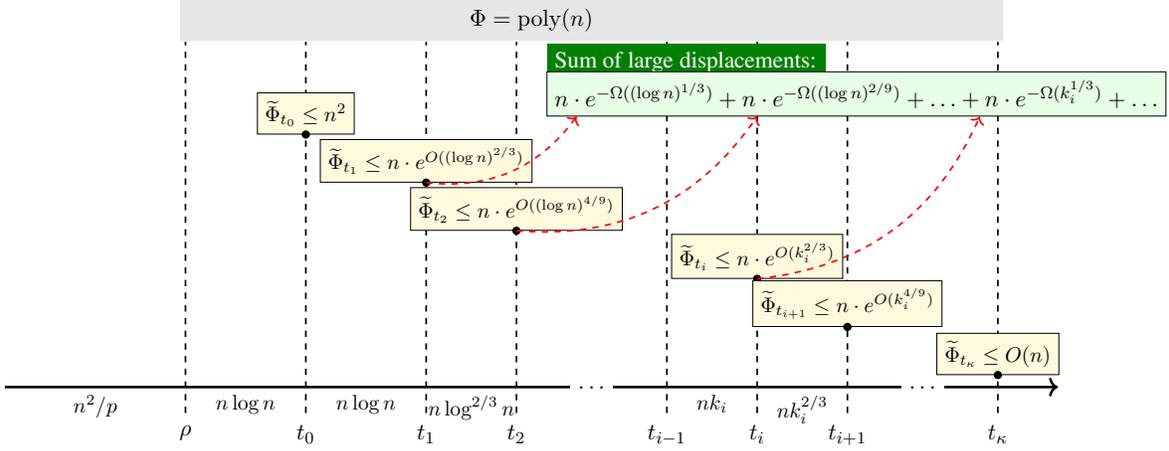
\begin{figure*}
    \centering
\scalebox{0.8}{
 \begin{tikzpicture}
  
  \draw[very thick,->] (0, 0) -- (17.5, 0);
  
\newcommand{\MarkPoint}[2]{
\draw[thick, dashed] (#1,6) -- (#1,-0.5);
\node[anchor=north] at (#1,-0.5) {#2};
}

\MarkPoint{3}{$\rho$}
\MarkPoint{5}{$t_0$}
\MarkPoint{7}{$t_1$}
\MarkPoint{8.5}{$t_2$}

\MarkPoint{11}{$t_{i-1}$}
\MarkPoint{12.5}{$t_i$}
\MarkPoint{14}{$t_{i+1}$}
\MarkPoint{16.5}{$t_{\kappa}$}

\node[anchor=north] at (1.5,0) {\small $n^2/p$};
\node[anchor=north] at (4,0) {\small $n \log n$};
\node[anchor=north] at (6,0) {\small $n \log n$};
\node[anchor=north] at (7.75,0) {\small $n \log^{2/3} n$};

\node[anchor=north] at (11.75,0) {\small $nk_i$};
\node[anchor=north] at (13.25,0) {\small $nk_{i}^{2/3}$};

\node[fill=white] at (9.75,0) {$\ldots$};

\node[fill=white] at (15.25,0) {$\ldots$};

\draw[anchor=north east,fill=black!10!white,draw=white] (2.9,5.75) rectangle ++(13.7,0.75);
\node at (8.75,6.1) {$\Phi = \mathrm{poly}(n)$};

\node[fill=beige,draw=black,anchor=south] at (5, 4.2) {\small $\widetilde{\Phi}_{t_0} \leq n^2$};

\node [draw, shape = circle, fill = black, minimum size = 0.12cm,inner sep=0pt] at (5, 4.2) {};

\node[fill=beige,draw=black,anchor=south] at (7, 3.4) {\small $\widetilde{\Phi}_{t_1} \leq n \cdot e^{O((\log n)^{2/3})}$};

\node [draw, shape = circle, fill = black, minimum size = 0.12cm,inner sep=0pt] at (7, 3.4) {};

\node[fill=beige,draw=black,anchor=south] at (8.5, 2.6) {\small $\widetilde{\Phi}_{t_2} \leq n \cdot e^{O((\log n)^{4/9})}$};

\node [draw, shape = circle, fill = black, minimum size = 0.12cm,inner sep=0pt] at (8.5, 2.6) {};

\node[fill=beige,draw=black,anchor=south] at (12.5, 1.8) {\small $\widetilde{\Phi}_{t_i} \leq n \cdot e^{O(k_i^{2/3})}$};

\node [draw, shape = circle, fill = black, minimum size = 0.12cm,inner sep=0pt] at (12.5, 1.8) {};

\node [fill=beige,draw=black,anchor=south] at (14, 1) {\small $\widetilde{\Phi}_{t_{i+1}} \leq n \cdot e^{O(k_i^{4/9})}$};

\node [draw, shape = circle, fill = black, minimum size = 0.12cm,inner sep=0pt] at (14, 1) {};

\node [fill=beige,draw=black,anchor=south] at (16.5, 0.2) {\small $\widetilde{\Phi}_{t_\kappa} \leq O(n)$};

\node [draw, shape = circle, fill = black, minimum size = 0.12cm,inner sep=0pt] at (16.5, 0.2) {};

\node[anchor=south west,color=white,fill=green!50!black] at (9, 5.1) {Sum of large displacements:};
\node[anchor=south west,fill=green!10!white,draw=black] at (9, 4.5) {$n \cdot e^{-\Omega((\log n)^{1/3})} + n \cdot e^{-\Omega((\log n)^{2/9})} + \ldots + n \cdot e^{-\Omega(k_i^{1/3})} + \ldots  $};

\path[dashed,red,thick,->] (7, 3.4) edge[bend right]  (9.5,4.5);
\path[dashed,red,thick,->] (8.5, 2.6) edge[bend right]  (12.5,4.5);
\path[dashed,red,thick,->] (12.5, 1.8) edge[bend right]  (16.2,4.5);

 \end{tikzpicture}
 }
    \caption{The phases for the proof of the $O(n)$ bound for the total deviation.}
    \label{fig:tdev_stages}
\end{figure*}

As in \cref{thm:mdev_whp}, the first (long) phase will be to ensure that we reach a state with a small maximum displacement. The second phase is to ensure that $\Phi_t = \mathrm{poly}(n)$. In both the first and the second phase, the permutations $\tau$ and $\sigma$ are reset to the identity. At the end of each of the subsequent phases $(t_{i}, t_{i+1}]$, we reset $\sigma_{t_{i+1}}$ to the identity permutation (with $\sigma_{t_{i+1}}'$ being the permutation before the reset) and update $\tau_{t_{i+1}}$ by resetting targets using $\theta$-filtering with $\theta = 2k_i^{2/3}$ (with~$\tau_{t_{i+1}}'$ being the permutation before the filtering). 

We will prove the following claim for the $\dev$ at step~$m$. %

\begin{claim} \label{claim:bound_on_tdev_tau}
It holds that, %
\begin{align}\label{eq:total_displacement_of_large}
 \Pro{ \dev(\tau_{m}) \leq n } \geq 1 - \frac{1}{2} n^{-2}.
\end{align}
\end{claim}

And we will also prove the following claim for $\dsp$ at step $m$. %

\begin{claim} \label{claim:bound_on_phi_m}
It holds that, 
\begin{align}
 & \Pro{\dsp(l_m, \tau_m) \leq 8n \cdot d \cdot k_{\kappa}^{2/3} + 8 \cdot \dev(\tau_{m})} \geq 1 - \frac{1}{2} n^{-2}. \label{eq:tdsp_bound}
\end{align}
\end{claim}

Before proving these claims, we will show how these two imply the conclusion. By \cref{lem:dsp_bounds_dev} and the triangle inequality, we have that
\[
 \dev(\pi_m) \leq \dsp(l_m, \tau_m) + \dev(\tau_m).
\]
By taking the union-bound between~\cref{eq:total_displacement_of_large} and~\cref{eq:tdsp_bound} we get%
\begin{align*}
  \Pro{\dev(\pi_m) \leq 8n \cdot d \cdot k_{\kappa}^{2/3} + n} 
  & \geq 1 - \frac{1}{2} n^{-2} - \frac{1}{2} n^{-2} = 1 - n^{-2},
\end{align*}
and so
\[
   \Pro{\dev(\pi_m) \leq c \cdot n} \geq 1 - n^{-2},
\]
for the constant $c = 9 k_{\kappa}^{2/3}$, since $d$ and $k_{\kappa}^{2/3}$ are constants.

Now, we go ahead and prove these two claims.

\begin{proof}[Proof of \cref{claim:bound_on_tdev_tau}]
For any phase $(t_{i-1}, t_{i}]$, we have that \begin{align*}
\dev(\tau_{t_{i}})
  & \stackrel{(a)}{\leq} 4 \cdot \sum_{\stackrel{j \in [n]:}{|\tau_{t_{i}}'(j) - j| > 2k_i^{2/3}} } |\tau_{t_{i}}'(j) - j| \notag \\
  & \stackrel{(b)}{\leq} 4 \cdot \sum_{\stackrel{j \in [n] :}{ |\sigma_{t_{i}}'(j) - j| > k_i^{2/3}}} |\sigma_{t_{i}}'(j) - j|
  + 8 \cdot \sum_{j \in [n] } |\sigma_{t_{i}}'(j) - \tau_{t_i}'(j)|, \\
  & = 4 \cdot \sum_{\stackrel{j \in [n] :}{ |\sigma_{t_{i}}'(j) - j| > k_i^{2/3}}} | \sigma_{t_{i}}'(j) - j|
  + 8 \cdot \dev(\sigma_{t_i}', \tau_{t_i}'),
\end{align*}
using in $(a)$ the definition of $\theta$-filtering (\cref{def:theta-filtering}) and in $(b)$ that $(i)$ if $|\tau_{t_{i}}'(j) - j| > 2k_i^{2/3}$, then $|\sigma_{t_{i}}'(j) - j| > k_i^{2/3}$ or $|\sigma_{t_{i}}'(j) - \tau_{t_i}'(j)| > k_i^{2/3}$ and $(ii)$ if $|\sigma_{t_{i}}'(j) - j| > k_i^{2/3}$, then 
\[
  |\tau_{t_{i}}'(j) - j| \leq |\sigma_{t_{i}}'(j) - j| + |\sigma_{t_{i}}'(j) - \tau_{t_i}'(j)|,
\]
and otherwise
\begin{align*}
  |\tau_{t_{i}}'(j) - j| 
    & \leq |\sigma_{t_{i}}'(j) - j| + |\sigma_{t_{i}}'(j) - \tau_{t_i}'(j)| 
      \leq 2 \cdot |\sigma_{t_{i}}'(j) - \tau_{t_i}'(j)|.
\end{align*}
Carrying on
\begin{align}
\dev(\tau_{t_{i}}) 
  & \stackrel{(a)}{\leq} 4 \cdot \sum_{\stackrel{j \in [n] :}{ |\sigma_{t_{i}}'(j) - j| > k_i^{2/3}}} |\sigma_{t_{i}}'(j) - j|
  + 8 \cdot \dev(\sigma_{t_{i - 1}}, \tau_{t_{i - 1}}) \notag \\
  & \stackrel{(b)}{=} 4 \cdot \sum_{\stackrel{j \in [n] :}{ |\sigma_{t_{i}}'(j) - j| > k_i^{2/3}}} |\sigma_{t_{i}}'(j) - j|
  + 8 \cdot \dev(\tau_{t_{i-1}}) \notag \\
  & \stackrel{(c)}{=} 4 \cdot \sum_{\stackrel{j \in [n] :}{ \delta_{t_{i}}(j) > k_i^{2/3}}} \delta_{t_{i}}(j)
  + 8 \cdot \dev(\tau_{t_{i-1}}) \notag \\
  & \stackrel{(d)}{\leq} \sum_{s = 0}^{i} 8^{i - s + 1} \sum_{\stackrel{j \in [n] :}{ \delta_{t_s}(j) > k_s^{2/3}}} \delta_{t_s}(j), \label{eq:total_deviation_bound}
\end{align}
using in $(a)$ that the $\adm$ operations cannot increase the deviation between $\sigma$ and $\tau$ during the phase (by \cref{lem:admissible_swaps_are_good}), in $(c)$ that $\sigma_{t_i} = \id_n$, in $(d)$ that $\delta_{t_s}(j) = |(\sigma_{t_s}')^{-1}(j) - j|$
and in $(e)$ an inductive argument and that $\dev(\tau_{t_{0}}) = 0$ since $\tau_{t_0}$ is the identity distribution.
By \cref{lem:sum_large_delta_bound}, we have that
\[
  \Pr\Big[\sum_{j \in [n] : \delta_{t_i}(j) \geq k_i^{2/3}} \delta_{t_i}(j) \leq n \cdot e^{-\frac{1}{2R} \cdot k_i^{1/3}} \Big] \geq 1 - n^{-3}.
\]
By taking the union-bound over all $i = 1, \ldots, \kappa$, we have that
\begin{align*}
  & \Pr\Bigg[ \sum_{i = 0}^{\kappa} 8^{\kappa - i +1}\sum_{j \in [N] : \delta_{t_i}(j) \geq k_i^{2/3}} \delta_{t_i}(j) \leq n \cdot \sum_{i = 0}^{\kappa} 8^{\kappa - i + 1} \cdot e^{-\frac{1}{2R} \cdot k_i^{1/3}} \Bigg]  \geq 1 - n^{-3} \cdot \kappa \geq 1 - \frac{1}{4} n^{-2}.
\end{align*}
By considering the ratio of two consecutive terms for $i$ and $i+1$ in the sum (for $i \geq 1$), we have that
\[
 8 \cdot e^{\frac{1}{2R} \cdot (k_{i+1}^{1/3} - k_i^{1/3})} 
  \stackrel{(a)}{\leq} 8 \cdot e^{\frac{1}{2R} \cdot (-\frac{1}{2} k_{i}^{1/3})} 
  \stackrel{(b)}{\leq} \frac{1}{2},
\]
using in $(a)$ that $k_{i+1}^{1/3} \leq \frac{1}{2} k_i^{1/3}$ (since $k_\kappa \geq (\frac{8 \cdot 16d}{p})^4$) and in $(b)$ that $k_\kappa \geq (64 R \log (16R))^3$.
Hence,
\[
  \Pr\Big[ \sum_{i = 0}^{\kappa} 8^{\kappa - i+1} \sum_{j \in [n] : \delta_{t_i}(j) \geq k_i^{2/3}} \delta_{t_i}(j) \leq n \Big] \geq 1 - \frac{1}{2} n^{-2}. \qedhere
\]
\end{proof}

\begin{proof}[Proof outline of \cref{claim:bound_on_phi_m}]
\begin{align} \label{eq:tdsp_deterministic_bound}
\dsp(l_m, \tau_m) \leq \frac{n}{\tilde{\alpha}} \cdot \log \Big( 1 + \frac{\widetilde{\Phi}_m}{n} \Big),
\end{align}
where we used the convexity of the $\widetilde{\Phi}_m$, i.e.,
\begin{align*}
 & n \cdot \Big(\exp\Big( \tilde{\alpha} \cdot \frac{1}{n} \dsp(l_m, \tau_m) \Big) - 1\Big) = n \cdot \Big(\exp\Big(\tilde{\alpha} \cdot \frac{1}{n} \sum_{j \in [N]} |j - d \cdot \tau_m(l_m[j])| \Big) - 1\Big) \leq \widetilde{\Phi}_m.
\end{align*}

\subparagraph{Analysis for $\Phi$: Proving a $\poly(n)$ bound.} Note that by the choice of $\alpha$ and $d$, it follows that $\alpha (d-1) \geq \log 20$ and so by \cref{lem:new_sorting_step_drift} it satisfies the drop inequality for sorting steps. Further, by \cref{lem:theta-filtering}, we have that \begin{align} \label{eq:max_displacement_filtering}
    \mdev(\tau_{t_i}, \tau_{t_i}') \leq 4k_i^{2/3}.
\end{align}
As in the analysis in \cref{thm:mdev_whp}, we have that for any $r \geq \rho \geq \frac{16d}{p} \cdot n^2$,
\[
 \Ex{\Phi_{r}} \leq 2n \cdot e^{\alpha \cdot d \cdot \frac{1}{\lambda} \sqrt{108 \cdot c' \cdot k_0 \cdot \log n}}.
\]
For the next phase, because each displacement in $\tau$ with respect to $l$ changes $|l[i] - \tau(l[i])|$ by at most $k_0^{2/3}$, we have that
\[
 \Ex{\Phi_{t_0}} \leq 2n \cdot e^{\alpha \cdot d k_0^{2/3}},
\]
and similarly
\[
 \Ex{\Phi_{t_1}} \leq 2n \cdot e^{4\alpha \cdot d k_1^{2/3}}.
\]
Now, for any $i \in [1, \kappa]$, using the steps in~\cref{eq:phi_bound_with_reallocation} %
\begin{align*}
 \Ex{\Phi_{t_i}} 
  & \leq \Ex{\Phi_{t_{i-1}}} \cdot e^{-\frac{p\alpha}{4} \cdot k_i} \cdot e^{4\alpha \cdot d k_{i}^{2/3}} + n \cdot e^{4\alpha \cdot d k_{i}^{2/3}} 
    \leq 2n \cdot e^{4\alpha \cdot d k_{i}^{2/3}},
\end{align*}
where in the last inequality we used that that $k_{i} = \frac{16d}{p} \cdot k_{i-1}^{2/3}$. %
By \cref{lem:new_sorting_step_drift}, $\Phi$ can only decrease between phases, so for any $t \in [t_{i-1}, t_i)$%
\begin{align*}
 \Ex{\Phi_t} 
  & \leq \Ex{\Phi_{t_{i-1}}} \leq 2n \cdot e^{4\alpha \cdot d k_1^{2/3}} \leq n^2,
\end{align*}
using that $k_1 = O(\log n)$. By Markov's inequality for any $t \in [t_1, t_{\kappa}]$,\[
  \Pro{\Phi_t \leq n^7} \geq 1 - n^{-5},
\]
and by taking the union-bound over all steps in $[t_1, t_{\kappa}]$
\begin{align*}
  \Pr\Big[\bigcap_{t \in [t_1, t_{\kappa}]} \left\{ \Phi_t \leq n^7 \right\}\Big] 
  & \geq 1 - n^{-5} \cdot n \cdot O(\log n) \geq 1 - \frac{1}{2} n^{-3}.
\end{align*}
Therefore, for $\tilde{c} = \frac{7}{\alpha}$, it follows that 
\begin{align} \label{eq:small_max_displacement}
\Pr\Big[\bigcap_{t \in [t_1, t_{\kappa}]} \Big\{ \mdsp(l_t, \tau_t) \leq \tilde{c} \cdot \log n \Big\} \Big] \geq 1 - \frac{1}{2} n^{-3}.
\end{align}

From here onwards we will consider the \emph{capped process} that never performs a mixing operation if the maximum displacement exceeds $\tilde{c} \cdot \log n$ in any step in $[t_1, t_{\kappa}]$. In the end, by the union bound, \Whp~the two processes will agree.

\subparagraph{Analysis for $\widetilde{\Phi}$: Proving an $O(n)$ bound.} Recall that $\widetilde{\Phi} = \widetilde{\Phi}(\tilde{\alpha}, d)$ with $\tilde{\alpha} = \frac{1}{6\tilde{c}} = \frac{\alpha}{42}$ for the same $d$ as above. Note again that $\tilde{\alpha} \cdot (d-1) \geq \log 20$ and so as we did with $\Phi_{t_0}$ we deduce  that
\[
  \Ex{\widetilde{\Phi}_{t_0}} \leq 4n \cdot e^{4\alpha \cdot d \cdot k_0^{2/3}}.
\]
By applying \cref{lem:new_sorting_step_drift} and \cref{lem:many_sorting}, at step $t_1$ before the $\adm$ and $\lopt$ steps, we have that
\[
  \Ex{\widetilde{\Phi}_{t_1}} \leq 4n \cdot e^{4\tilde{\alpha} \cdot d \cdot k_0^{2/3}} \cdot e^{-\frac{\tilde{\alpha}p}{4} \cdot k_1} \leq n^{-3},
\]
using that $k_0 = k_1 = \frac{16^2 \cdot 108 \cdot d^2}{p^2} \cdot \frac{c'}{\lambda} \cdot \log n$ (and so $\frac{\tilde{\alpha} p}{4} \cdot k_1 \geq 60 \log n$).

Therefore, by Markov's inequality we have that
\[
 \Pro{\widetilde{\Phi}_{t_1} \leq 1} \geq 1 - n^{-3},
\]
and so, after moving the targets we have that
\begin{align} \label{eq:total_phi_markov}
\Pro{\widetilde{\Phi}_{t_1} \leq 4n \cdot e^{2\alpha \cdot d \cdot k_1^{2/3}} } \geq 1 - n^{-3}.
\end{align}

We will consider the interval $[t_i, t_{i+1}]$ consisting of $n \cdot k_{i+1}$ steps, where in step $t_i$ we have that $\widetilde{\Phi}_{t_i} \leq 4n \cdot e^{2\alpha \cdot d \cdot k_i^{2/3}}$. Our goal is to show that \Whp~$\widetilde{\Phi}_{t_{i+1}} \leq 4n$, by applying Azuma's inequality.

Let $\widetilde{\Phi}_t'$ be the value of the potential before applying the $\adm$ and $\lopt$ operations, which can only lead to a decrease, i.e., $\widetilde{\Phi}_t\leq\widetilde{\Phi}_t'$. 

Recall that since we are in the capped process, we have that $\mdsp(l_t, \tau_t) \leq \tilde{c} \cdot \log n$ holds in every step and so\[
  \widetilde{\phi}_t'(i)  \leq e^{\frac{1}{6 \tilde{c}} \cdot \tilde{c} \cdot \log n} = n^{1/6}.
\]
This also implies that 
\begin{align} \label{eq:phi_poly_bound}
    \widetilde{\Phi}_t' \leq n^{7/6}.
\end{align}

We define the auxiliary potential function at step $t \geq t_i$,
\[
  \widehat{\Phi}_t = \widetilde{\Phi}_t  \cdot \left( 1 - \frac{\tilde{\alpha} }{4 \cdot (n-1)} \right)^{-z(t)},
\] and \[
  \widehat{\Phi}_t' = \widetilde{\Phi}_t'  \cdot \left( 1 - \frac{\tilde{\alpha} }{4 \cdot (n-1)} \right)^{-z(t)} ,
\]
where $z(t)$ is the number of completed sorting steps from step $t_i+1$ up to step $t$. This is a super-martingale, since if $t$ is a sorting step then 
\begin{align*}
& \Ex{\widehat{\Phi}_{t+1} \, \middle| \, \mathcal{F}_t }
  \leq \widetilde{\Phi}_t' \cdot \Big( 1 - \frac{\tilde{\alpha} }{4 \cdot (n-1)} \Big)  \cdot \Big( 1 - \frac{\tilde{\alpha} }{4 \cdot (n-1)} \Big)^{-z(t)-1}  = \widehat{\Phi}_t',
\end{align*}
and if $t$ is a mixing step, then $\widehat{\Phi}_t \leq \widehat{\Phi}_{t+1}'$.

We now proceed to bound the difference $|\widehat{\Phi}_{t+1}' - \widehat{\Phi}_t|$. First, note that %
\begin{align} \label{eq:delta_tilde_phi_bound}
 \left| \widetilde{\Phi}_{t+1}' - \widetilde{\Phi}_t \right| \leq n^{1/6},
\end{align}
and further note that\begin{align} 
\Big( 1 - \frac{\tilde{\alpha} }{4 (n-1)} \Big)^{-z(t)}
  & \stackrel{(a)}{\leq} \Big( 1 + \frac{\tilde{\alpha} }{2 (n-1)} \Big)^{z(t)} 
    \stackrel{(b)}{\leq} e^{\frac{\tilde{\alpha}}{2 (n-1)} \cdot n \cdot k_{i+1}}
    \leq n^{1/6}, \label{eq:factor_bound}
\end{align}
using in $(a)$ that $(1 - \frac{\tilde{\alpha}}{4(n-1)})^{-1} \leq 1 + \frac{\tilde{\alpha}}{2 (n-1)}$ (since $\frac{1}{1 - \eps} \leq 1 + 2\eps$) and in $(b)$ that $z(t) \leq n \cdot k_{i+1}$, $k_{i} = o(\log n)$ (for $i \geq 2$) and $\tilde{\alpha}$ is a constant.

Therefore, combining the above bounds, it follows that 
\begin{align*}
  \left|\widehat{\Phi}_{t+1}' - \widehat{\Phi}_t\right| 
    & \leq \max\Big\{ \Big| \Big( 1 - \frac{\tilde{\alpha}}{4 (n{-}1)} \Big)^{-1} \cdot \widetilde{\Phi}_{t+1}' - \widetilde{\Phi}_t \Big|, \Big| \widetilde{\Phi}_{t+1}' -\widetilde{\Phi}_{t} \Big| \Big\} \cdot \Big( 1 - \frac{\tilde{\alpha}}{4 (n{-}1)} \Big)^{-z(t)}
     \\
    & \stackrel{(a)}{\leq} \Big( \big| \widetilde{\Phi}_{t+1}' - \widetilde{\Phi}_t \big| + \frac{\tilde{\alpha}}{2 (n{-}1)} \cdot \widetilde{\Phi}_{t+1}' \Big) \cdot \Big( 1 - \frac{\tilde{\alpha}}{4 (n{-}1)} \Big)^{{-}z(t)} \\
    & \stackrel{(b)}{\leq} \big( n^{1/6} + 2 n^{1/6} \big) \cdot  n^{1/6} \leq 3n^{1/3},
\end{align*}
using in $(a)$ that $(1 - \frac{\tilde{\alpha}}{4(n-1)})^{-1} \leq 1 + \frac{\tilde{\alpha}}{2 (n-1)}$ (since $\frac{1}{1 - \eps} \leq 1 + 2\eps$) and in $(b)$ the bounds \cref{eq:phi_poly_bound}, \cref{eq:delta_tilde_phi_bound} and \cref{eq:factor_bound}.

Applying Azuma's inequality (\cref{lem:azuma_modified}) for the sub-martingale $\widehat{\Phi}_{t_i}, \widehat{\Phi}_{t_i+1}', \widehat{\Phi}_{t_i+1}, \ldots$ with $\Lambda = n$ and $\Delta = 3n^{1/3}$ 
\begin{align*}
 & \Pro{ \widehat{\Phi}_{t_{i+1}} \leq \widehat{\Phi}_{t_i} + n \, \middle| \, \mathcal{F}_{t_i}, \widetilde{\Phi}_{t_i} \leq z_0} \geq 1 - \exp\Big( - \frac{n^2}{3n \cdot k_{i+1} \cdot (3n^{1/3})^{2} + 3n^{1/3}\cdot z_0} \Big),
\end{align*}
where $z_0 = 4n \cdot e^{4\tilde{\alpha} \cdot d \cdot k_i^{2/3}}$. 
By the definition of $\widehat{\Phi}_t$, we have that
\[
\widehat{\Phi}_{t_i} = \widetilde{\Phi}_{t_i} \leq 4n \cdot e^{4\tilde{\alpha} \cdot d \cdot k_i^{2/3}},
\]
and for $\widetilde{\Phi}_{t_{i+1}}''$ the value of the potential before the $\theta$-filtering, assuming that $\{ \widehat{\Phi}_{t_{i+1}} \leq \widehat{\Phi}_{t_i} + n \}$ holds, we get that
\begin{align*}
 \widetilde{\Phi}_{t_{i+1}}''
   & \leq \Big(4n \cdot e^{4\tilde{\alpha} \cdot d \cdot k_i^{2/3}} + n \Big) \cdot \Big( 1 - \frac{\tilde{\alpha}}{4 (n-1)} \Big)^{p \cdot n \cdot k_{i+1}} \leq 5n \cdot e^{4\tilde{\alpha} \cdot d \cdot k_i^{2/3}} \cdot e^{- \frac{\tilde{\alpha} p}{4 \cdot (n-1)} \cdot n \cdot k_{i+1}} \leq 4n,
\end{align*}
using that $1 + x \leq e^x$ and that $k_{i+1} = \frac{16d}{p} \cdot k_i^{2/3}$.

Therefore, after the $\theta$-filtering of the targets, using~\cref{eq:max_displacement_filtering}, it follows that 
\begin{align*}
& \Pro{\left. \widetilde{\Phi}_{t_{i+1}} \leq 4n \cdot e^{4\tilde{\alpha} \cdot d \cdot k_{i+1}^{2/3}} \, \right| \, \mathcal{F}_{t_i}, \widetilde{\Phi}_{t_i}  \leq 4n \cdot e^{4\tilde{\alpha} \cdot d \cdot k_i^{2/3}}} \geq 1 - n^{-3},
\end{align*}
using the chain rule we obtain
\begin{align*}
  & \Pr\Big[ \bigcap_{i \in [\kappa]}\Big\{ \widetilde{\Phi}_{t_{i}} \leq 4n \cdot e^{4\tilde{\alpha}\cdot d \cdot k_{i}^{2/3}} \Big\}  \Big|  \,\mathcal{F}_{t_0}, \widetilde{\Phi}_{t_0} \leq 4n\cdot e^{4\tilde{\alpha} \cdot d \cdot k_1^{2/3}}\Big]  \geq (1 - n^{-3})^{\kappa},
\end{align*}
and combining with \cref{eq:total_phi_markov}, we get that 
\begin{align*}
 & \Pr\Big[\bigcap_{i \in [\kappa]} \Big\{ \widetilde{\Phi}_{t_{i}} \leq 4n \cdot e^{4\tilde{\alpha} \cdot d \cdot k_{i}^{2/3}} \Big\}\Big] 
   \geq \left( 1 - n^{-3} \right)^{\kappa} \cdot \left( 1 - n^{-3} \right) 
   \geq 1 - \frac{1}{4}n^{-2}.
\end{align*}
Finally taking the union bound with \cref{eq:small_max_displacement} we get that the capped process agrees with the original process  and so using \cref{eq:tdsp_deterministic_bound}, we conclude that 
\[
 \Pro{\dsp(l_m, \tau_m) \leq 8n \cdot d \cdot k_{\kappa}^{2/3}} 
   \geq 1 - \frac{1}{4}n^{-2} \frac{1}{2}n^{-3}s
   \geq 1 - \frac{1}{2} n^{-2}. \qedhere
\]
\end{proof}
This completes the proof of \cref{thm:total_deviation_whp}.
\end{proof}

\section{Lower Bounds}\label{sec:lower_bounds}
In this section, first we prove a lower bound on the maximum deviation which is tight for any $b \in \big[1, \frac{n}{\log^2 n}\big]$ and a lower bound on the total deviation which is tight for any constant $b \geq 1$ (\cref{lem:max_and_total_deviation_lb}). Next, we prove a lower bound on the convergence time for Na\"ive Sort in the setting of~\cite{AnagnostopoulosKMU11} to reach a configuration with $O(b \cdot \log n)$ maximum deviation (\cref{lem:convergence_time_lb}).

\begin{lemma}[Maximum and Total Deviation] \label{lem:max_and_total_deviation_lb} 
Consider Na\"ive Sort under random adjacent rank swaps for any $b \in \big[1, \frac{n}{\log^2 n}\big]$. Then, starting from a sorted permutation $(i)$ it holds that for $m = \frac{1}{100} \cdot (b+1)^2 \cdot n \log n$,
\[
  \Pr\Big[ \bigcup_{t \in [1, m]}\Big\{ \mdev(\pi_t) \geq \frac{1}{800} \cdot (b+1) \cdot \log n\Big\} \Big] \geq 1 - o(1),
\]
and $(ii)$ at step $m = \frac{1}{100} \cdot (b+1)^2 \cdot n$, it holds that
\[
  \Pr\Big[ \dev(\pi_m) \geq \frac{1}{1600} \cdot (b+1) \cdot n \Big] \geq 1 - o(1).
\]
\end{lemma}
\begin{proof}
\emph{Statement (i).} Consider $m = c\cdot (b+1)^2 \cdot n   \log n$ where $c = \frac{1}{100}$. We start with the sorted permutation $\pi_0 = \id_n$. For each $i$, we keep track of $(j_t(i))_{t\geq 0}$ defined as follows $j_0(i) = i$, in each mixing step we stay at the same index so $j_t(i) = j_{t-1}(i)$ and in each sorting step we follow the element at $j_{t-1}(i)$, so $j_t(i) \in \{ j_{t-1}(i) - 1, j_{t-1}(i) + 1 \}$. Let $K =(n-1)/D$. We consider the starting points $\mathcal{J}$ which are multiples of $D = \frac{1}{400} \cdot (b+1) \cdot \log n$, i.e., $\mathcal{J} = \{ i \cdot D : i \in [K]\}$.

In these steps $1$ through $m$, the indices in $\mathcal{J}$ will be involved in at least $\frac{1}{2} \cdot m \cdot \frac{n}{K}$ mixing steps. Also during these steps, as long as all indices have a deviation of at most $\frac{1}{800} \cdot (b+1) \cdot \log n$ it means that each $\pi_t(j_t(i))$ is doing an unbiased random walk that does not interfere with the others. Hence, \Whp~there is a particular $j_t(i)$ with $|j_t(i) - i| \geq \frac{1}{4} \sqrt{c} \cdot (b+1) \cdot \log n$, due to the mixing steps (by taking the maximum deviation of the $K \geq \log^2 n$ unbiased random walks). Further, by a Chernoff bound (\cref{lem:chernoff_bound}), \Whp~this value will be chosen in a sorting step at most $2 \frac{m}{n \cdot (b+1)} = 2c \cdot (b+1) \cdot \log n$ times. Even if each of these sorting steps moves the item closer to the target, \Whp~we have that 
\begin{align*}
 \mdev(\pi_m) 
  & \geq \frac{1}{40} \cdot (b+1) \cdot \log n - \frac{1}{50} \cdot (b+1) \cdot \log n  = \frac{1}{200} \cdot (b+1) \cdot \log n,
\end{align*}
which concludes the proof for the first statement.

\paragraph{\it Statement (ii).} Again, we will consider the first $m=c\cdot (b+1)^2 \cdot n$ steps where $c = \frac{1}{100}$. By the Central Limit Theorem \Whp~at least $\frac{1}{4} n$ of the indices $i$ will have $|j_t(i) - i| \geq \frac{1}{4} \cdot \sqrt{c} \cdot (b+1)$ and at most a $\frac{1}{2}$ of these will be involved in more than $2 \cdot c \cdot (b+1)$ sorting steps. Hence, it follows that \Whp
\[
\dev(\pi_m) \geq \frac{1}{8} \cdot \left( \frac{1}{40} - \frac{1}{50}\right) \cdot (b+1) \cdot n,
\]
which concludes the proof for the second statement.
\end{proof}

Next, we prove a lower bound on the convergence time depending on the maximum deviation of the initial configuration.

\begin{lemma}[Convergence Time] \label{lem:convergence_time_lb}
Consider Na\"ive Sort under random adjacent rank swaps for any $b \geq 1$, starting from any configuration $\pi_0$ with $\mdev(\pi_0) \geq 96 \log n$. Then, the time to reach a configuration with $O(\log n)$ maximum deviation is at least $\frac{1}{2} (n-1) \cdot \mdev(\pi_0)$ with probability at least $1 - n^{-1}$.
\end{lemma}
\begin{proof}
Consider the first $m = \frac{1}{2} (n-1) \cdot \mdev(\pi_0)$ steps of the algorithm. We will maintain a sequence of elements $(i_t)_{t = 0}^m$, with the property that their displacement decreases by at most $1$ and this happens with probability at most $1/(n-1)$. We start with $i_0$ being an element having the maximum deviation at step $0$. Then, at step $t \geq 0$: 
\begin{itemize}
  \item If $t+1$ is a sorting step, at most one of the two swaps that involve $i_t$, either at index $\pi_t^{-1}(i_t) - 1$ or at $\pi_t^{-1}(i_t)$ will reduce the displacement of $i_t$. So we keep $i_{t+1} = i_t$.
  \item If $t+1$ is a mixing step, it could be that one of the two swaps involving $i_t$ decreases the displacement of $i_t$ by more than $1$. In this case the displacement of either $i_t - 1$ or $i_t +1$ will be at least $|i_t - \pi_t(i_t)| - 1$ (after the swap), so we update $i_{t+1}$ to that element.
\end{itemize}
So, the displacement of element $i_t$ in each step decreases by~$1$ with probability at most $1/(n-1)$. Let~$D_m$ be the amount it decreases in $m$ steps, then, using a Chernoff bound (\cref{lem:chernoff_bound})
\begin{align*}
\Pr\Big[D_m \leq \Big(1 + \frac{1}{2}\Big) \cdot \frac{m}{n-1}\Big] 
    & \geq 1 - e^{-\frac{1}{4} \cdot \frac{1}{2} \cdot\frac13\cdot \mdev(\pi_0)}  
      \geq 1 - e^{- \frac{1}{4} \cdot \frac{1}{3} \cdot \frac{1}{2} \cdot 96 \cdot \log n} 
      = 1 - n^{-1},
\end{align*}
using that $\mdev(\pi_0) \geq 96 \cdot \log n$, and so
\[
\Pro{D_m \leq \frac{3}{4}\mdev(\pi_0) } \geq 1 - n^{-1}.
\]

Finally, using that $\mdev(\pi_m) \geq \mdev(\pi_0) - D_m$, we conclude that
\[
  \Pro{\mdev(\pi_m) \geq \frac{1}{4}\mdev(\pi_0)} 
    \geq 
    1 - n^{-1},
\]
and therefore the process takes at least $m$ steps to converge. 
\end{proof}

\section*{APPENDIX}

\subsection{Tools}

We begin with a standard Chernoff bound.
\begin{lemma}[Chernoff's bound]\label{lem:chernoff_bound}
Let $X_1, \ldots, X_n$ be independent random variables taking values in $\{ 0, 1 \}$. Let $X = \sum_{i = 1}^n X_i$ and $\mu = \Ex{X}$. Then for any $\delta \in (0, 1)$, we have that
\begin{align*}
  & \max\left\{ \Pro{X \geq (1 + \delta) \cdot \mu}, \Pro{X \leq (1 - \delta) \cdot \mu} \right\} \leq e^{-\delta^2 \mu/3}.
\end{align*}
\end{lemma}

Next, we state Azuma's inequality.
\begin{lemma}[Azuma's Inequality for Super-Martingales {\cite[Problem 6.5]{DP09}}]\label{lem:azuma} %
Let $X_0, \ldots, X_n$ be a super-martingale satisfying $|X_i - X_{i-1}| \leq \Delta_i$ for any $i \in [n]$, then for any $\Lambda > 0$,
\[
\Pro{X_n \geq X_0 + \Lambda} \leq \exp\left(- \frac{\Lambda^2}{2 \cdot \sum_{i=1}^n \Delta_i^2} \right).
\]
\end{lemma}

Next, we state and prove a variant of the Azuma's inequality in which every other step the sub-martingale may decrease by an arbitrary amount.

\begin{lemma}[Modified Azuma]\label{lem:azuma_modified}
Let $X_0, \ldots, X_n$ be a non-negative super-martingale with $n$ an even integer. Further, assume it satisfies $|X_{i} - X_{i-1}| \leq \Delta$ for any odd $i \in [n]$ and $X_{i} \leq X_{i-1}$ for any even $i \in [n]$. Then for any $\Lambda > 0$,
\[
\Pro{X_n \geq X_0 + \Lambda} \leq \exp\left(- \frac{\Lambda^2}{3n\Delta^2 + 2X_0 \cdot \Delta} \right).
\]
\end{lemma}
\begin{proof}
We cannot directly apply Azuma's inequality as it requires the absolute value of the differences to be bounded, but here at an even $i$ the value of~$X_i$ could drop all the way to $0$. 

We will define an auxiliary sub-martingale $Z$ and show that this satisfies $|Z_i - Z_{i-1}| \leq \Delta$ for any $i$. To this end, for any $k \in \{0,...,n/2-1\}$, we define $\tau_k = \lfloor \frac{1}{\Delta}(X_{2k + 2} - X_{2k + 1})\rfloor$ and set \begin{align*}
Y_0^{(k)} &= X_{2k}, \\
Y_1^{(k)} &= X_{2k + 1}, \\ 
Y_2^{(k)} &= X_{2k + 1} + d_1,\\
Y_3^{(k)} &= X_{2k + 1} + d_1 + d_2, \\
& \vdots \\
Y_{2 + \tau_k}^{(k)} &= X_{2k + 2}, 
\end{align*}
where 
\[
  d_i = \begin{cases}
      \Delta, & i < \tau_k, \\
      X_{2k + 2} - X_{2k + 1}  - \Delta \cdot (\tau_k - 1), & i = \tau_k.
  \end{cases}
\]
Then, we set \[
Z = Y_0^{(0)}, Y_1^{(0)}, \ldots, Y_0^{(1)}, Y_1^{(1)}, \ldots , Y_0^{(\frac12n-1)}, Y_1^{(\frac12n-1)}.
\]
Note that $Y^{(k)}$ forms a sub-martingale with respect to $\mathcal{F}_0' = \mathcal{F}_{2k}$, $\mathcal{F}_1' = \mathcal{F}_{2k+1}$ and $\mathcal{F}_{2+t}' = (X_{2k+1}, t, X_{2k+2})$ for $t \in [0, \tau_k]$, and so $Z$ also forms a sub-martingale, and further it satisfies the bounded difference inequality for $\Delta$.

Throughout the $n$ steps $X$ can increase by at most $n/2 \cdot \Delta$. Hence, in total there can be at most $2 \cdot (\frac{n}{2} + \frac{1}{\Delta}X_0)$ extra decrease steps. So $Z$ has at most $n' = \frac{3}{2} \cdot n +\frac{1}{\Delta}X_0$ steps. By padding $Z$ with the same last value we can ensure that it always has exactly $n'$ steps, and it remains a super-martingale. Therefore, applying Azuma's inequality (\cref{lem:azuma}) for any $\Lambda > 0$ and $\Delta_i = \Delta$, we get that
\[
\Pro{Z_{n'} \geq Z_0 + \Lambda} \leq \exp\left(- \frac{\Lambda^2}{3n\Delta^2 + 2X_0 \cdot \Delta} \right).
\]
Using that $Z_0 = X_0$ and $Z_{n'} = X_n$, we obtain the conclusion.
\end{proof}

\subsection{Omitted Material}
\label{apndx}
In this Appendix, we first restate and establish a result whose proof was omitted. Then, we state and prove lemmas which were not included in the main text.

\DevDspLem*
\begin{proof}
    Let $i$ be such that $|i-\pi(i)| = \mdev(\pi)$, and suppose without loss of generality that $i < \pi(i)$.
    Suppose also that $l[j] = \pi(i)$.
    We distinguish two cases.
    If $j\leq d\cdot \pi(i) - \frac d2\cdot \mdev(\pi)$,
    then
    \[
        \mdsp(l)
        \geq
        |d\cdot l[j] - j|
        =
        d\cdot \pi(i) - j
        \geq
        \frac d2\cdot \mdev(\pi)
        .
    \]
    Suppose now that $j > d\cdot \pi(i) - \frac d2\cdot \mdev(\pi)$.
    Let $i' = \min\{\pi(k) \colon i\leq k \leq n\}$, thus $\pi(i')\leq i$.
    Let also $j'$ be such that $l[j'] = \pi(i')$.
    Note that $j'\geq j$, since the relative order of non-empty items in $l$ is the same as the order in $\pi$.
    Then
    \begin{align*}
        \mdsp(l)
        &\geq
        |d\cdot l[j'] - j'|
        =
        j' - d\cdot \pi(i')
        \geq
        j  - d\cdot i
        \geq
        \frac d2\cdot \mdev(\pi)
        ,
    \end{align*}
    since $j > d\cdot \pi(i) - \frac d2\cdot \mdev(\pi)$ and $i = \pi(i) - \mdev(\pi)$.
    This proves the first inequality of the lemma.

Next we prove the bound on the total deviation. 
Let $T = \{j \cdot d : j \in [n] \}$ be the subset of all target positions in list $l$. We first move all $n$ items to positions in $T$, such that if an item $j\in[n]$ is at a position  $x \in (i \cdot d, (i+1) \cdot d)$, we move $j$ to position $i \cdot d$ if $j \leq i$ else we move it to $(i+1) \cdot d$; this may result in more than one element mapped to the same position. Note that this step can only decrease the displacement of each element (and thus also the total displacement).

Let $L_i = \{ j \cdot d\colon j \in [1, i] \}$ and $R_i = \{ j \cdot d\colon j \in (i, n] \}$. Let $s_i$ be the \emph{total} number of items mapped to positions in $L_i$. Note that with respect to partition $\{L_i,R_i\}$ of $T$, there are at least $|s_i - i|$ elements $j$ whose target position $j\cdot d$ is in a different part of the partition than the one at which element $j$ is mapped to.

It follows that by ``ignoring'' all positions $x \in (i \cdot d, (i+1) \cdot d)$ when computing the total displacement, the total displacement reduces by at least $|s_i - i|\cdot (d-1)$. And if we ignore all positions except for those in $T$, the total displacement decreases by at least $2\sum_{i \in [n]} |s_i - i|\cdot (d-1)$.

The last step is to redistribute the items so that exactly one item is mapped to each position in $T$. This may increase the total displacement by at most $2\sum_{i \in [n]} |s_i-i|$ as we explain next, and combining this with the previous point we get that overall the total displacement decreases, since $d > 1$.

To map one item to each position in $T$, for each $i\in [n-1]$, we have to move exactly $|s_i - i|$ items between $i \cdot d$ and $(i+1) \cdot d$: from $i \cdot d$ to $(i+1) \cdot d$ if $s_i>i$ and in the opposite direction otherwise. The total displacement of the resulting list is equal to the total deviation of $\pi$, and so the conclusion follows.
\end{proof}

\begin{lemma}
\label{lem:admissible_swaps_are_good}
Consider any reals $a, b, c, d$ such that $a \leq b$ and $c \leq d$.  Then, $(i)$
\[
 |a - c| + |b - d| \leq |a - d| + |b - c|, 
\]
and $(ii)$
\[
 \max\{|a - c|, |b - d|\} \leq \max\{|a - d|, |b - c| \}.
\]
\end{lemma}
\begin{proof}
Without loss of generality assume that $a \leq c$. Then, we consider the three valid cases separately:

\medskip

\textbf{Case 1 [$a \leq b \leq c \leq d$]:} In this case,
\begin{align*}
|a - c| + |b - d| 
 & = (c - a) + (d - b) = (d - a) + (c - b) = |a - d| + |b - c|.
\end{align*}
And similarly,
\begin{align*}
\max\{|a - c|, |b - d| \} 
  & = \max\{ c - a, b -d  \} \leq d - a = \max\{|a - d|, |b - c| \}.
\end{align*}

\textbf{Case 2 [$a \leq c \leq b \leq d$]:} In this case,
\begin{align*}
|a - c| + |b - d| 
 & = (c - a) + (d - b) = (d - a) - (b - c) = |a - d| - |b - c| \leq |a - d| + |b - c|.
\end{align*}
And similarly,
\begin{align*}
\max\{|a - c|, |b - d| \} 
  & = \max\{ c - a, d - b \} \leq d - a = \max\{|a - d|, |b - c| \}.
\end{align*}

\textbf{Case 3 [$a \leq c \leq d \leq b$]:} In this case,
\begin{align*}
|a - c| + |b - d| 
 & = (c - a) + (b - d) \leq (d - a) + (b - c)  = |a - d| + |b - c|.
\end{align*}
And similarly,
\begin{align*}
\max\{|a - c|, |b - d| \} 
  & = \max\{ c - a, d - b \} 
    \leq \max\{ d - a, b - c \} 
    = \max\{|a - d|, |b - c| \}. \qedhere
\end{align*}
\end{proof}

\begin{lemma} \label{lem:mgf_taylor}
Consider any integer-valued random variable $S \geq 0$ for which there exists $\lambda >0$ and $c' > 0$ such that $\Ex{e^{3\lambda S}} \leq c'$. Then, for any $\alpha \in [-\lambda, \lambda]$, it holds that
\[
  \Ex{e^{\alpha S}} 
    \leq 1 + \alpha \cdot \Ex{S} + \alpha^2 \cdot \frac{c'}{\lambda^2}.
\]
\end{lemma}
\begin{proof}
Consider the function 
\[
 f(\alpha) = \Ex{e^{\alpha S}}.
\]
Then, by Taylor's theorem we have that for any $\alpha \in [-\lambda, \lambda]$, there exists $\xi \in (-\lambda, \lambda)$ such that
\begin{align*}
f(\alpha) & = f(0) + \alpha \cdot f'(0) + \frac{\alpha^2}{2} \cdot f''(\xi) \\
  & = 1 + \alpha \cdot \Ex{S} + \frac{\alpha^2}{2} \cdot \Ex{S^2 \cdot e^{\xi S}} \\
  & \stackrel{(a)}{\leq} 1 + \alpha \cdot \Ex{S} + \alpha^2 \cdot \Ex{S^2 \cdot e^{\lambda S}} \\
  & \stackrel{(b)}{\leq} 1 + \alpha \cdot \Ex{S} + \alpha^2 \cdot \Ex{\left( \frac{e^{\lambda S}}{\lambda} \right)^2 \cdot e^{\lambda S}} \\
  & \stackrel{(c)}{\leq} 1 + \alpha \cdot \Ex{S} + \alpha^2 \cdot \frac{c'}{\lambda^2},
\end{align*}
using in $(a)$ that $\xi < \lambda$, in $(b)$ that $S \leq \frac{e^{\lambda S}}{\lambda}$ and in $(c)$ that $\Ex{e^{3\lambda S}} \leq c'$.
\end{proof}

\begin{lemma} \label{lem:many_sorting}
Consider a process with bounded average rate of mixing steps for any $b \geq 1$. Then, for any $m \geq (b+1) \cdot n$ consecutive steps, at least a $\frac{1}{4(b+1)}$ fraction are sorting steps.
\end{lemma}
\begin{proof}
Consider any interval $[s, t)$ with length $|t - s| = (b+1) \cdot n$. Let $u$ be the last sorting step before $s$, then by \cref{def:bounded average rate if mixing steps}, it follows that the next $n-1$ sorting steps are in the interval $[s, t)$. Hence, in any interval of length $(b+1) \cdot n$, there is at least a fraction of $\frac{n-1}{(b+1) \cdot n} \geq \frac{1}{2\cdot (b+1)}$ sorting steps.

For intervals of length at least $(b+1) \cdot n$, by splitting it into sub-intervals of length exactly $(b+1) \cdot n$ (except for possibly the last one), we get that at least a $\frac{1}{2} \cdot \frac{1}{2(b+1)}$ fraction of the steps are sorting.
\end{proof}

\bibliographystyle{alphaurl}
\bibliography{sorting}

\end{document}